\documentclass{article}%
\usepackage{amsmath}
\usepackage{amsfonts}
\usepackage{amssymb}
\usepackage{graphicx}%
\providecommand{\U}[1]{\protect\rule{.1in}{.1in}}
\newtheorem{theorem}{Theorem}[section]

\newtheorem{corollary}[theorem]{Corollary}

\newtheorem{lemma}[theorem]{Lemma}

\newtheorem{proposition}[theorem]{Proposition}

\newenvironment{proof}[1][Proof]{\noindent\textbf{#1.} }{\ \rule{0.5em}{0.5em}}
\begin{document}

\title{Quantum Fidelities, Their Duals, And Convex Analysis}
\author{Keiji Matsumoto\\Quantum Computation Group, National Institute of Informatics, \\2-1-2 Hitotsubashi, Chiyoda-ku, Tokyo 101-8430\\e-mail : keiji@nii.ac.jp}
\maketitle

\begin{abstract}
We study tree kinds of quantum fidelity. Usual Uhlmann's fidelity, minus of
$f$-divergence when $f\left(  x\right)  =-\sqrt{x}$, and the one introduced by
the author via reverse test. All of them are quantum extensions of classical
fidelity, where the first one is the largest and the third one is the
smallest. We characterize them in terms of convex optimization, and introduce
their 'dual' quantity, or the polar of the minus of the fidelity. They turned
out to be monotone increasing by unital completely positive maps, concave, and
linked to its classical version via optimization about classical-to-quantum
maps and quantum-to-classical maps.

\end{abstract}

\section{Introduction}

We study tree kinds of quantum fidelity. Usual Uhlmann's fidelity, minus of
$f$-divergence when $f\left(  x\right)  =-\sqrt{x}$, and the one introduced by
the author via reverse test. All of them are quantum extensions of classical
fidelity, where the first one is the largest and the third one is the
smallest. We characterize them in terms of convex optimization, and introduce
their 'dual' quantity, or the polar of the minus of the fidelity. They turned
out to be monotone increasing by unital completely positive maps, concave, and
linked to its classical version via optimization about classical-to-quantum
maps and quantum-to-classical maps.

\section{Notations and conventions}

In the paper, it is assumed that dimensions of Hilbert spaces are finite. The
set of operators, self-adjoint operators, positive operators, and density
operators over the Hilbert space $\mathcal{H}$ will be denoted by
$\mathcal{L}\left(  \mathcal{H}\right)  $, $\mathcal{L}_{sa}\left(
\mathcal{H}\right)  $, $\mathcal{P}\left(  \mathcal{H}\right)  $, and
$\mathcal{S}\left(  \mathcal{H}\right)  $, respectively. When $\mathcal{H=}%
\mathbb{C}^{k}$ , they are denoted by $\mathcal{L}_{k}$, $\mathcal{L}_{sa,k}$,
$\mathcal{P}_{k}$ and $\mathcal{S}_{k}$. The identity operator in
$\mathbb{C}^{k}$ and identity transform in $\mathcal{L}_{k}$ will be denoted
by $I_{k}$ and $\mathbf{I}_{k}$, respectively. $\mathcal{L}$ denotes
$\bigcup_{k\in\mathbb{N}}\mathcal{L}_{k}$, and $\mathcal{L}_{sa}$,
$\mathcal{P}$, and $\mathcal{S}$ are defined similarly. We define
\[
\mathcal{L}_{sa}^{\times2}:=\bigcup_{k\in\mathbb{N}}\mathcal{L}_{sa,k}%
^{\times2}\,,\,\mathcal{P}^{\times2}:=\bigcup_{k\in\mathbb{N}}\mathcal{P}%
_{k}^{\times2}\,,
\]
etc., where
\begin{align*}
\mathcal{L}_{sa,k}^{\times2}  &  :=\mathcal{L}_{sa,k}\times\mathcal{L}%
_{sa,k},\\
\mathcal{P}_{k}^{\times2}  &  :=\mathcal{P}_{k}\times\mathcal{P}%
_{k}\mathcal{\,},
\end{align*}
etc.

We fix a standard orthonormal basis $\left\{  \left\vert i\right\rangle
\right\}  $ of $\mathbb{C}^{k}$, and denote the commutative algebra spanned by
$\left\{  \left\vert i\right\rangle \left\langle i\right\vert \right\}
_{i=1}^{k}$ \thinspace by $\mathcal{C}_{k}$. Also,
\begin{align*}
\mathcal{CP}_{k}  &  =\mathcal{C}_{k}\cap\mathcal{P}_{k}\,,\,\,\mathcal{CS}%
_{k}=\mathcal{C}_{k}\cap\mathcal{S}_{k}\,,\\
\mathcal{C}_{k}^{\times2}  &  :=\mathcal{C}_{k}\times\mathcal{C}_{k},\text{
}\mathcal{CP}_{k}^{\times2}:=\mathcal{CP}_{k}\times\mathcal{CP}_{k}%
,\mathcal{CS}_{k}^{\times2}:=\mathcal{CS}_{k}\times\mathcal{CS}_{k},\\
\mathcal{C}^{\times2}  &  :=\bigcup_{k\in\mathbb{N}}\mathcal{C}_{k}^{\times
2}\,,\,C\mathcal{P}^{\times2}:=\bigcup_{k\in\mathbb{N}}C\mathcal{P}%
_{k}^{\times2},\,C\mathcal{S}^{\times2}:=\bigcup_{k\in\mathbb{N}}%
C\mathcal{S}_{k}^{\times2}.
\end{align*}

\ Any unital completely positive (CP) map from $\mathcal{C}_{n}$ to operators
$\mathcal{L}_{k}$ is in the following form;
\[
\Phi_{M}^{\ast}\left(  \sum_{i=1}^{n}l_{i}\left\vert i\right\rangle
\left\langle i\right\vert \right)  =\sum_{i=1}^{n}l_{i}M_{i},
\]
where $M=\left\{  M_{i};M_{i}\in\mathcal{P}_{k},i=1,\cdots n\right\}  $ is a
POVM over $\mathbb{C}^{k}$. \ Since a member of $\mathcal{C}_{n}$ is
represented by an array $l=\left(  l_{1},\cdots,l_{n}\right)  $, we also
write
\[
\Phi_{M}^{\ast}\left(  l\right)  =\sum_{i=1}^{n}l_{i}M_{i}.
\]
Also, any completely positive completely positive (CPTP) map from
$\mathcal{L}_{k}$ to $\mathcal{C}_{n}$ is \ in the form of
\[
\Phi_{M}\left(  L\right)  =\sum_{i=1}^{n}\left(  \mathrm{tr\,}L\,M_{i}\right)
\,\left\vert i\right\rangle \left\langle i\right\vert .
\]
With $l=\left(  l_{i}\right)  _{i=1}^{k}=\left(  \mathrm{tr\,}L\,M_{i}\right)
_{i=1}^{k}$, \ we also write this as \
\[
\Phi_{M}\left(  L\right)  =l.
\]

Any unital CP map from $\mathcal{L}_{k}$ to $\mathcal{C}_{n}$ is in the form
of
\[
\Psi_{\vec{\rho}}^{\ast}\left(  L\right)  =\sum_{i=1}^{n}\left(
\mathrm{tr}\,\rho_{i}\,L\right)  \left\vert i\right\rangle \left\langle
i\right\vert ,
\]
where $\vec{\rho}=\left\{  \rho_{i};\rho_{i}\in\mathcal{S}_{k},i=1,\cdots
,n\right\}  $ is a set of states. \ With $l=\left(  l_{i}\right)  _{i=1}%
^{k}=\left(  \mathrm{tr}\,\rho_{i}\,L\right)  _{i=1}^{k}$, we also write this
as
\[
\Psi_{\vec{\rho}}^{\ast}\left(  L\right)  =l.
\]
Also, any CPTP map from $\mathcal{C}_{n}$ to $\mathcal{L}_{k}$ is \ in the
form of%
\[
\Psi_{\vec{\rho}}\left(  l\right)  :=\Psi_{\vec{\rho}}\left(  \sum_{i=1}%
^{n}l_{i}\left\vert i\right\rangle \left\langle i\right\vert \right)
=\sum_{i=1}^{n}l_{i}\,\rho_{i}.
\]

We denote by $\Phi_{\mathcal{C}}$ the pinching operation
\[
\Phi_{\mathcal{C}}\left(  X\right)  =\sum_{i=1}^{n}\,\left\langle i\right\vert
X\left\vert i\right\rangle \,\left\vert i\right\rangle \left\langle
i\right\vert .
\]
\thinspace\ When the operator $X$ is not invertible, $X^{-1}$ means
Moore-Penrose generalized inverse.

\section{Classical fidelity, fidelity, and minimum fidelity}

For probability distributions $p=\left(  p_{x}\right)  _{x=1}^{k}$ and
$q=\left(  q_{x}\right)  _{x=1}^{k}$, we define
\[
F^{C}\left(  p,q\right)  :=\sum_{i=1}^{k}\sqrt{p_{i}q_{i}}.
\]
For $\rho$, $\sigma\in\mathcal{S}\left(  \mathbb{C}^{k}\right)  $, Uhlmann's
fidelity is
\[
F_{\max}\left(  \rho,\sigma\right)  :=\mathrm{tr}\,\sqrt{\sqrt{\sigma}%
\rho\sqrt{\sigma}}.
\]
It is known that
\begin{equation}
F_{\max}\left(  \rho,\sigma\right)  =\min_{M:\text{measurement}}F^{C}\left(
M\left(  \rho\right)  ,M\left(  \sigma\right)  \right)  \label{F-measurement}%
\end{equation}
where $M\left(  \rho\right)  $ is the probability distribution of measurement
$M$ applied to $\rho$.

A "dual" of $F\left(  \rho,\sigma\right)  $ \cite{Matsumoto}\cite{Matsumoto-2}
is
\[
F_{\min}\left(  \rho,\sigma\right)  :=\max_{\Phi}\left\{  F^{C}\left(
p,q\right)  ;\Phi\text{ is a CPTP with }\Phi\left(  p\right)  =\rho
,\Phi\left(  q\right)  =\sigma\right\}  .
\]
When $\mathrm{supp}\,\rho\subset\mathrm{supp}\,\sigma$,
\[
F_{\min}\left(  \rho,\sigma\right)  =\mathrm{tr}\,\sigma\sqrt{\sigma
^{-1/2}\rho\sigma^{-1/2}},
\]
where $\sigma^{-1}$ is the generalized inverse. When $\mathrm{supp}%
\,\rho\not \subset \mathrm{supp}\,\sigma$,
\[
F_{\min}\left(  \rho,\sigma\right)  =\mathrm{tr}\,\sigma\sqrt{\sigma
^{-1/2}\tilde{\rho}\sigma^{-1/2}},
\]
where
\begin{align*}
\tilde{\rho}  &  =\rho_{11}-\rho_{12}\,\rho_{22}^{-1}\,\rho_{21},\\
\rho_{11}  &  :=\pi_{\sigma}\,\rho\,\pi_{\sigma},\,\rho_{12}:=\pi_{\sigma
}\,\rho\,\left(  I_{k}-\pi_{\sigma}\right)  ,\\
\rho_{21}  &  :=\rho_{12}^{\dagger},\,\,\rho_{22}:=\left(  I_{k}-\pi_{\sigma
}\right)  \rho\,\left(  I_{k}-\pi_{\sigma}\right)  ,\\
\pi_{\sigma}  &  :\text{projection onto }\mathrm{supp}\,\sigma.
\end{align*}
Also,
\[
F_{1/2}\left(  \rho,\sigma\right)  :=\mathrm{tr}\,\rho^{1/2}\sigma^{1/2}.
\]

From here, we extend $F_{\max}$ , $F_{\min}$ and $F_{1/2}$ to functionals on
$\mathcal{L}_{sa}^{\times2}$ in the following manner.
\[
F^{Q}\left(  X,Y\right)  :=\left\{
\begin{array}
[c]{cc}%
\sqrt{\mathrm{tr}\,X\,\mathrm{tr}\,Y}\,F\left(  \frac{1}{\mathrm{tr}%
\,X}X,\frac{1}{\mathrm{tr}\,Y}Y\right)  ,\, & \left(  X,Y\right)
\in\mathcal{P}^{\times2},\,X\neq0,\,Y\neq0,\\
0 & \left(  X,Y\right)  \in\mathcal{P}^{\times2},\,XY=0,\\
-\infty & \left(  X,Y\right)  \not \in \mathcal{P}^{\times2}.
\end{array}
\right.
\]
Also $F_{\mathrm{cl}}$ is extended to a functional over two signed measures,
in the analogous manner.

All of $F_{\max}$ , $F_{\min}$ and $F_{1/2}$ satisfy the following properties:

\begin{itemize}
\item (positive homogeneity)
\[
F^{Q}\left(  cX,cY\right)  =cF^{Q}\left(  X,Y\right)  ,\,\forall c\geq0.
\]

\item (concavity)
\begin{align*}
&  F^{Q}\left(  \lambda X_{1}+\left(  1-\lambda\right)  X_{2},\lambda
Y_{1}+\left(  1-\lambda\right)  Y_{2}\right) \\
&  \geq\lambda F^{Q}\left(  X_{1},Y_{1}\right)  +\left(  1-\lambda\right)
F^{Q}\left(  X_{2},Y_{2}\right)  ,\,\forall\lambda\in\left[  0,1\right]  .
\end{align*}

\item (CPTP monotonicity) $F^{Q}\left(  X,Y\right)  \leq F^{Q}\left(
\Lambda\left(  X\right)  ,\Lambda\left(  Y\right)  \right)  $ for any CPTP map
$\Lambda$.

\item (positivity) For any $X,Y\in\mathcal{P}_{k}$, $F^{Q}\left(  X,Y\right)
\geq0$ .

\item (strong homogeneity)
\[
F^{Q}\left(  \lambda X,\mu Y\right)  =\sqrt{\lambda\mu}F^{Q}\left(
X,Y\right)  .
\]

\item (normalization) for any positive vectors $x=\left(  x_{i}\right)
_{i=1}^{k}$ and $y=\left(  y_{i}\right)  _{x=1}^{k}$, and for an orthogonal
basis $\left\{  \left\vert i\right\rangle \right\}  _{i=1}^{k}$%
\[
F^{Q}\left(  \sum_{i=1}^{k}x_{i}\left\vert i\right\rangle \left\langle
i\right\vert ,\sum_{i=1}^{k}y_{i}\left\vert i\right\rangle \left\langle
i\right\vert \right)  =F^{C}\left(  x,y\right)  .
\]

\item (symmetry)
\[
F^{Q}\left(  X,Y\right)  =F^{Q}\left(  Y,X\right)
\]

\item (additivity)%
\[
F^{Q}\left(  \left[
\begin{array}
[c]{cc}%
X_{1} & 0\\
0 & X_{2}%
\end{array}
\right]  ,\left[
\begin{array}
[c]{cc}%
Y_{1} & 0\\
0 & Y_{2}%
\end{array}
\right]  \right)  =F^{Q}\left(  X_{1},Y_{1}\right)  +F^{Q}\left(  X_{2}%
,Y_{2}\right)
\]

\end{itemize}

By definition of $F_{\min}\left(  X,Y\right)  $ and (\ref{F-measurement}), it
is obvious that
\begin{equation}
F_{\min}\left(  X,Y\right)  \leq F_{\max}\left(  X,Y\right)  . \label{Fmin<F}%
\end{equation}
Also, observe that \ joint concavity and homogeneity implies
\begin{align*}
F^{Q}\left(  X_{1}+X_{2},Y_{1}+Y_{2}\right)   &  =2F^{Q}\left(  \frac{1}%
{2}\left(  X_{1}+X_{2}\right)  ,\frac{1}{2}\left(  Y_{1}+Y_{2}\right)  \right)
\\
&  \geq2\cdot\frac{1}{2}\left(  F^{Q}\left(  X_{1},Y_{1}\right)  +F^{Q}\left(
X_{2},Y_{2}\right)  \,\right) \\
&  =F^{Q}\left(  X_{1},Y_{1}\right)  +F^{Q}\left(  X_{2},Y_{2}\right)  .
\end{align*}
If $F^{Q}$ is positive in addition,
\[
F^{Q}\left(  X_{1}+X_{2},Y_{1}+Y_{2}\right)  \geq F^{Q}\left(  X_{1}%
,Y_{1}\right)  ,
\]
for any $\left(  X_{2},Y_{2}\right)  \in\mathcal{P}^{\times2}$.

Strong homogeneity and joint concavity implies strong joint concavity:
\begin{align*}
F^{Q}\left(  \sum_{i=1}^{m}\lambda_{i}X_{i},\sum_{i=1}^{m}\mu_{i}Y_{i}\right)
&  \geq\sum_{i=1}^{m}\lambda_{i}F^{Q}\left(  X_{i},\frac{\mu_{i}}{\lambda_{i}%
}Y_{i}\right) \\
&  =\sum_{i=1}^{m}\sqrt{\lambda_{i}\mu_{i}}F^{Q}\left(  X_{i},Y_{i}\right)  .
\end{align*}

We define $\mathcal{F}_{0}$ as the set of all the proper closed concave
functionals, which satisfies positive homogeneity, positivity, and
$\mathrm{dom}\,F^{Q}=\mathcal{P}^{\times2}$. $\mathcal{F}_{1}$ is the subset
of $\mathcal{F}_{0}$ whose element satisfies CPTP monotonicity, normalization,
strong homogeneity, symmetry, and additivity. The following lemma is almost
immediate from Lemma\thinspace\ref{lem:lim-conv-closed}.

\begin{lemma}
\label{lem:lim-F0}Consider a family $\left\{  F_{i}\right\}  _{i\in I}$ ,
where $F_{i}\in\mathcal{F}_{0}$. Then, $\inf_{i\in I}F_{i}\in\mathcal{F}_{0}$.
If in addition each $f_{i}$ is a member of $\mathcal{F}_{1}$, so is
$\inf_{i\in I}F_{i}$.
\end{lemma}

Also, it is known that:

\begin{theorem}
\label{th:Fmin<FQ<F}\cite{Matsumoto}\cite{Matsumoto-2}Suppose that a
functional $F^{Q}$ on $\mathcal{L}_{sa}^{\times2}$ is normalized and CPTP
monotone,
\begin{equation}
F_{\min}\left(  X,Y\right)  \leq F^{Q}\left(  X,Y\right)  \leq F_{\max}\left(
X,Y\right)  . \label{Fmin<FQ<F}%
\end{equation}

\end{theorem}

\section{Convex programing representations}

\begin{lemma}
\label{recession}A functional $F^{Q}$ on $\mathcal{L}_{sa}^{\times2}$ is a
member of $\mathcal{F}_{0}$ if and only if there is a closed convex subset
$\mathcal{M}_{F^{Q}}$ \ of \ $\mathcal{P}^{\times2}$ such that
\begin{equation}
F^{Q}\left(  X,Y\right)  =\inf_{\left(  L_{0},L_{1}\right)  \in\mathcal{M}%
_{F^{Q}}}\mathrm{tr}\,L_{0}X+\mathrm{tr}\,L_{1}Y, \label{F-inf-linear}%
\end{equation}
and $0^{+}\mathcal{M}_{F^{Q}}=\mathcal{P}^{\times2}$, or
\begin{equation}
\left(  L_{0},L_{1}\right)  \in\mathcal{M}_{F^{Q}}\Rightarrow\left(
L_{0}+M_{0},L_{1}+M_{1}\right)  \in\mathcal{M}_{F^{Q}},\,\forall M_{0}%
,M_{1}\geq0. \label{L+M-inv}%
\end{equation}
In addition, the correspondence between $F^{Q}$ and $\mathcal{M}_{F^{Q}}$ is
one-to-one. In fact,
\[
\mathcal{M}_{F^{Q}}=\left\{  \left(  L_{0},L_{1}\right)  \,;\left(
L_{0},L_{1}\right)  \in\mathcal{L}_{sa}^{\times2}\,,\forall\left(  X,Y\right)
\in\mathcal{P}^{\times2},\,\,\mathrm{tr}\,L_{0}X+\mathrm{tr}\,L_{1}Y\geq
F^{Q}\left(  X,Y\right)  \right\}
\]

\end{lemma}

\begin{proof}
By Lemma\thinspace\ref{lem:sublinear-2}, it is obvious that
(\ref{F-inf-linear}) holds for a closed convex set $\mathcal{M}_{F^{Q}}$. Let
$\mathcal{M}_{F^{Q}}$ be a closed convex set which may not satisfy
(\ref{L+M-inv}). Then,
\[
\mathcal{\tilde{M}}_{F^{Q}}:=\left\{  \left(  L_{0}+M_{0},L_{1}+M_{1}\right)
;\left(  L_{0},L_{1}\right)  \in\mathcal{M}_{F^{Q}},M_{0},M_{1}\geq0\right\}
\]
is a closed convex set satisfying (\ref{L+M-inv}). Also, if $\left(
X,Y\right)  \in$ $\mathcal{P}^{\times2}$,
\[
\inf_{\left(  L_{0},L_{1}\right)  \in\mathcal{\tilde{M}}_{F^{Q}}}%
\mathrm{tr}\,L_{0}X+\mathrm{tr}\,L_{1}Y=\inf_{\left(  L_{0},L_{1}\right)
\in\mathcal{M}_{F^{Q}}}\mathrm{tr}\,L_{0}X+\mathrm{tr}\,L_{1}Y=F^{Q}\left(
X,Y\right)  .
\]
If $\left(  X,Y\right)  \not \in $ $\mathcal{P}^{\times2}$,
\[
\inf_{\left(  L_{0},L_{1}\right)  \in\mathcal{\tilde{M}}_{F^{Q}}}%
\mathrm{tr}\,L_{0}X+\mathrm{tr}\,L_{1}Y=-\infty=F^{Q}\left(  X,Y\right)  .
\]
Thus, for a given $F^{Q}$, there is a closed convex set $\mathcal{M}_{F^{Q}}$
satisfying (\ref{F-inf-linear}) and (\ref{L+M-inv}). By positivity of $F^{Q}$,
$\mathcal{M}_{F^{Q}}\ \subset\mathcal{P}^{\times2}$.

That the correspondence between $F^{Q}$ and $\mathcal{M}_{F^{Q}}$ is
one-to-one is obvious by Lemma\thinspace\ref{lem:sublinear-2}.
\end{proof}

\begin{lemma}
\label{lem:recession}Suppose a closed convex set $\mathcal{M}_{F^{Q}}%
\subset\mathcal{P}^{\times2}$ satisfies (\ref{L+M-inv}). Then, for any $M_{0}%
$, $M_{1}>0$, there is a positive number $t_{0}$ such that
\begin{align*}
\forall\,t  &  \geq t_{0}\,\,\left(  t\,M_{0},t\,M_{1}\right)  \in
\mathcal{M}_{F^{Q}},\\
\forall\,t  &  <t_{0}\,\,\left(  t\,M_{0},t\,M_{1}\right)  \not \in
\mathcal{M}_{F^{Q}}%
\end{align*}

\end{lemma}

\begin{proof}
To prove the statement, suppose $\left(  L_{0},L_{1}\right)  \in$
\ $\mathcal{M}_{F^{Q}}$ and $M_{0}$, $M_{1}>0$. Then there is $t_{0}\geq0$
such that
\[
L_{0}\leq t_{0}M_{0},\,\,L_{1}\leq t_{0}M_{1}\,.
\]
Since $\left(  L_{0},L_{1}\right)  \in$ \ $\mathcal{M}_{F^{Q}}$, $\left(
t_{0}\,M_{0},t_{0}\,M_{1}\right)  \in$ \ $\mathcal{M}_{F^{Q}}$ by
(\ref{L+M-inv}). Thus, for any $t\geq t_{0}\,$, we have $\left(
t\,M_{0},t\,M_{1}\right)  \in\mathcal{M}_{F^{Q}}$. Since the set $\left\{
\left(  t\,M_{0},t\,M_{1}\right)  ;\,t\geq0\right\}  $ is closed, its
intersection with $\mathcal{M}_{F^{Q}}$ is also closed. So the minimum
\[
\,t_{0}=\min\left\{  t;\left(  t\,M_{0},t\,M_{1}\right)  \in\mathcal{M}%
_{F^{Q}}\right\}
\]
exists.
\end{proof}

\begin{lemma}
Let $\mathcal{M}_{F^{Q}}\subset\mathcal{P}^{\times2}$ be a closed convex set
with $0^{+}\mathcal{M}_{F^{Q}}=\mathcal{P}^{\times2}$. Then for any $X>0$,
$Y>0$, $\left(  X,Y\right)  \rightarrow\mathrm{tr}\,L_{0}X+\mathrm{tr}%
\,L_{1}Y$ has minimum in $\mathcal{M}_{F^{Q}}$. Also, its infimum is finite if
and only if $\left(  X,Y\right)  \in\mathcal{P}^{\times2}$.
\end{lemma}

\begin{proof}
The second statement is trivial. So we prove the only first one. Choose $a$
which is strictly larger then the infimum, and consider a level set
\begin{align*}
&  \left\{  \left(  L_{0},L_{1}\right)  \in\mathcal{M}_{F^{Q}};\mathrm{tr}%
\,L_{0}X+\mathrm{tr}\,L_{1}Y\leq\alpha\right\} \\
&  =\mathcal{M}_{F^{Q}}\cap\left\{  \left(  L_{0},L_{1}\right)  \in
\mathcal{P}^{\times2};\mathrm{tr}\,L_{0}X+\mathrm{tr}\,L_{1}Y\leq
\alpha\right\}  ,
\end{align*}
which is closed. If $X>0$ and $Y>0$, the recession cone of this is empty, due
to the following reasons. If it has direction of recession, it should be a
member of $\mathcal{P}^{\times2}$, because the set is subset of $\mathcal{P}%
^{\times2}$. But, for any $\left(  L_{0}^{\prime},L_{1}^{\prime}\right)
\in\mathcal{P}^{\times2}$, there is $t$ such that
\[
\mathrm{tr}\,\left(  L_{0}+tL_{0}^{\prime}\right)  X+\mathrm{tr}\,\left(
L_{1}+tL_{1}^{\prime}\right)  Y>\alpha.
\]
So there is no direction of recession. Therefore, the set is bounded.
Therefore, $\left(  X,Y\right)  \rightarrow\mathrm{tr}\,L_{0}X+\mathrm{tr}%
\,L_{1}Y$ \ has minimum over the set, which coincide with the minimum over
$\mathcal{M}_{F^{Q}}$.
\end{proof}

The proof of the following two propositions are immediate, thus omitted.

\begin{proposition}
\label{prop:monotone}Suppose $F^{Q}$ is a member of $\mathcal{F}_{0}$. Then,
$F^{Q}$ is CPTP monotone if and only if $\mathcal{M}_{F^{Q}}$ satisfies%
\[
\left(  L_{0},L_{1}\right)  \in\mathcal{M}_{F^{Q}}\Rightarrow\left(
\Lambda^{\ast}\left(  L_{0}\right)  ,\Lambda^{\ast}\left(  L_{1}\right)
\right)  \in\mathcal{M}_{F^{Q}}%
\]
for any CPTP map $\Lambda$.
\end{proposition}

\begin{proposition}
\label{prop:t-1/t}Suppose $F^{Q}$ is a member of $\mathcal{F}_{0}$. Then,
$F^{Q}$ is strongly homogeneous if and only if $\mathcal{M}_{F^{Q}}$ satisfies%
\[
\,\left(  L_{0},L_{1}\right)  \in\mathcal{M}_{F^{Q}}\Rightarrow\left(
t\,L_{0},\frac{1}{t}L_{1}\right)  \in\mathcal{M}_{F^{Q}}%
\]
for any $t>0$.
\end{proposition}

\begin{proposition}
\label{prop:commute}Suppose $F^{Q}$ is a member of $\mathcal{F}_{0}$ that is
CPTP monotone and normalized. Then,
\[
\mathcal{M}_{F^{Q}}\cap\mathcal{C}^{\times2}\mathcal{=}\left\{  \,\left(
L_{0},L_{1}\right)  \in\mathcal{C\,};\,L_{0}>0,L_{1}\geq\frac{1}{4}L_{0}%
^{-1}\right\}  .
\]

\end{proposition}

\begin{proof}
Let $l_{0}:=\left(  l_{0,i}\right)  _{i=1}^{k}$, $X=\sum_{i=1}^{k}%
x_{i}\left\vert i\right\rangle \left\langle i\right\vert ,Y=\sum_{i=1}%
^{k}y_{i}\left\vert i\right\rangle \left\langle i\right\vert $. Then, by
normalization,
\begin{align}
F^{Q}\left(  X,Y\right)   &  =\sum_{i=1}^{k}\sqrt{x_{i}y_{i}}=\min_{l_{0}%
,}\sum_{i=1}^{k}\left(  x_{i}l_{0,i}+\frac{y_{i}}{4\,l_{0,i}}\right)
\nonumber\\
&  =\min_{\left(  l_{0},l_{1}\right)  \in\mathcal{CP}}\left\{  \sum_{i=1}%
^{k}\left(  x_{i}l_{0,i}+y_{i}l_{1,i}\right)  \,;\,l_{1,i}\geq\frac
{1}{4l_{0,i}}\right\} \nonumber\\
&  =\min\left\{  \,\left(  \mathrm{tr}\,L_{0}X+\mathrm{tr}\,L_{1}Y\right)
\,;\left(  L_{0},L_{1}\right)  \in\mathcal{C\,};\,L_{0}>0,L_{1}\geq\frac{1}%
{4}L_{0}^{-1}\right\}  . \label{FC-1}%
\end{align}

By CPTP monotonicity, $\Phi_{\mathcal{C}}\left(  \mathcal{M}_{F^{Q}}\right)
\subset\mathcal{M}_{F^{Q}}$. Thus,
\[
\Phi_{\mathcal{C}}\left(  \mathcal{M}_{F^{Q}}\right)  \subset\mathcal{M}%
_{F^{Q}}\cap\mathcal{C}^{\times2}.
\]
Since each element $X$ of $\mathcal{C}$ is unchanged by $\Phi_{\mathcal{C}}$,
$\Phi_{\mathcal{C}}\left(  X\right)  =X$, the opposite inclusion is also true:%
\[
\mathcal{M}_{F^{Q}}\cap\mathcal{C}^{\times2}=\Phi_{\mathcal{C}}\left(
\mathcal{M}_{F^{Q}}\cap\mathcal{C}^{\times2}\right)  \subset\Phi_{\mathcal{C}%
}\left(  \mathcal{M}_{F^{Q}}\right)  .
\]
Therefore, we have
\[
\Phi_{\mathcal{C}}\left(  \mathcal{M}_{F^{Q}}\right)  =\mathcal{M}_{F^{Q}}%
\cap\mathcal{C}^{\times2}.
\]

Observe
\begin{align*}
F^{Q}\left(  X,Y\right)   &  =\min_{\left(  L_{0},L_{1}\right)  \in
\mathcal{M}_{F^{Q}}}\left(  \mathrm{tr}\,L_{0}X+\mathrm{tr}\,L_{1}Y\right) \\
&  =\min_{\left(  L_{0},L_{1}\right)  \in\mathcal{M}_{F^{Q}}}\left(
\mathrm{tr}\,L_{0}\Phi_{\mathcal{C}}\left(  X\right)  +\mathrm{tr}\,L_{1}%
\Phi_{\mathcal{C}}\left(  Y\right)  \right) \\
&  =\min_{\left(  L_{0},L_{1}\right)  \in\mathcal{M}_{F^{Q}}}\left(
\mathrm{tr}\,\Phi_{\mathcal{C}}\left(  L_{0}\right)  X+\mathrm{tr}%
\,\Phi_{\mathcal{C}}\left(  L_{1}\right)  Y\right) \\
&  =\min_{\left(  L_{0},L_{1}\right)  \in\Phi_{\mathcal{C}}\left(
\mathcal{M}_{F^{Q}}\right)  }\left(  \mathrm{tr}\,L_{0}X+\mathrm{tr}%
\,L_{1}Y\right) \\
&  =\min_{\left(  L_{0},L_{1}\right)  \in\mathcal{M}_{F^{Q}}\cap
\mathcal{C}^{\times2}}\left(  \mathrm{tr}\,L_{0}X+\mathrm{tr}\,L_{1}Y\right)
.
\end{align*}
Since this and (\ref{FC-1}) holds for any $\left(  X,Y\right)  \in
\mathcal{PC}^{\times2}$, by Lemma\thinspace\ref{lem:separation}, we have the assertion.
\end{proof}

\section{The minimum points of convex programs}

Suppose a member $F^{Q}$ of $\mathcal{F}_{0}$ has the derivative
\[
\mathrm{D}F^{Q}\left(  X,Y\right)  \left(  T,S\right)  =\mathrm{tr}%
\,L_{0,\ast}T+\mathrm{tr}\,L_{1,\ast}S.
\]
Then, for any $\lambda>0$,
\begin{align*}
\mathrm{D}F^{Q}\left(  \lambda X,\lambda Y\right)  \left(  T,S\right)   &
=\left.  \frac{\mathrm{d}}{\mathrm{d}t}\,F^{Q}\left(  \lambda X+tT,\lambda
Y+tS\right)  \right\vert _{t=0}\\
&  =\lambda\left.  \frac{\mathrm{d}}{\mathrm{d}t}\,F^{Q}\left(  X+\frac
{t}{\lambda}T,Y+\frac{t}{\lambda}S\right)  \right\vert _{t=0}\\
&  =\lambda\mathrm{D}F^{Q}\left(  X,Y\right)  \left(  \frac{1}{\lambda}%
T,\frac{1}{\lambda}S\right) \\
&  =\mathrm{D}F^{Q}\left(  X,Y\right)  \left(  T,S\right)  .
\end{align*}
\ Also, since $F^{Q}$ is concave,%
\begin{align*}
&  F^{Q}\left(  \lambda X+T,\lambda Y+S\right)  -F^{Q}\left(  \lambda
X,\lambda Y\right) \\
&  =F^{Q}\left(  \lambda X+T,\lambda Y+S\right)  -\lambda F^{Q}\left(
X,Y\right) \\
&  \leq\mathrm{tr}\,L_{0,\ast}T+\mathrm{tr}\,L_{1,\ast}S.
\end{align*}
Since $F^{Q}$ is closed, it is upper semi continuous. Thus, taking
$\varlimsup_{\lambda\rightarrow0}$ of both ends, for any \ $T\geq0$, $S\geq
0$,
\begin{align*}
&  \varlimsup_{\lambda\rightarrow0}\left\{  F^{Q}\left(  \lambda X+T,\lambda
Y+S\right)  -\lambda F^{Q}\left(  X,Y\right)  \right\} \\
&  =F^{Q}\left(  T,S\right) \\
&  \leq\mathrm{tr}\,L_{0,\ast}T+\mathrm{tr}\,L_{1,\ast}S,
\end{align*}
which means%
\[
\left(  L_{0,\ast},L_{1,\ast}\right)  \in\mathcal{M}_{F^{Q}}\,.
\]

Also, since $F^{Q}$ is positively homogeneous,
\[
F^{Q}\left(  0,0\right)  -F^{Q}\left(  X,Y\right)  =\mathrm{D}F^{Q}\left(
X,Y\right)  \left(  -X,-Y\right)
\]
holds for any \ $X\geq0$, $Y\geq0$. Thus,
\[
F^{Q}\left(  X,Y\right)  =\mathrm{D}F^{Q}\left(  X,Y\right)  \left(
X,Y\right)  =\mathrm{tr}\,L_{0,\ast}X+\mathrm{tr}\,L_{1,\ast}Y.
\]
So $\left(  L_{0,\ast},L_{1,\ast}\right)  $ achieves (\ref{F-inf-linear}).

Define, for each $Z\in\mathcal{P}_{k}$, the linear transform $\mathbf{S}_{Z}$
on $\mathcal{L}_{sa,k}$ by the equation%
\[
X=\mathbf{S}_{Z}\left(  X\right)  \,Z+Z\,\,\mathbf{S}_{Z}\left(  X\right)  \,.
\]
When $Z>0$,
\begin{equation}
\mathbf{S}_{Z}\left(  X\right)  =\int_{0}^{\infty}e^{-t\,Z}\,X\,e^{-tZ}%
\,\mathrm{d\,}t.\label{S-int}%
\end{equation}
In fact $\mathbf{S}_{Z}$ is self-dual with respect to Hilbert-Schmidt inner
product,
\[
\mathbf{S}_{Z}=\mathbf{S}_{Z}^{\ast}.
\]
When $Z>0$, this is obvious from the second expression of $\mathbf{S}_{Z}$.
When is positive but may have null eigenspace,
\begin{align*}
\mathrm{tr}\,\mathbf{S}_{Z}\left(  X\right)  Y &  =\mathrm{tr}\,\mathbf{S}%
_{Z}\left(  X\right)  \left(  \mathbf{S}_{Z}\left(  Y\right)  Z+Z\mathbf{S}%
_{Z}\left(  Y\right)  \right)  \\
&  =\mathrm{tr}\,\left(  \mathbf{S}_{Z}\left(  X\right)  Z+Z\mathbf{S}%
_{Z}\left(  X\right)  \right)  \mathbf{S}_{Z}\left(  Y\right)  \\
&  =\mathrm{tr}\,X\mathbf{S}_{Z}\left(  Y\right)  \text{.}%
\end{align*}
So \ $\mathbf{S}_{Z}$ is self-dual, if viewed as a linear transform on
$\mathcal{L}_{sa,k}$.

The derivative of $f_{1}\left(  X\right)  =\sqrt{X}$ is
\[
\mathrm{D}f_{1}\left(  X\right)  \left(  T\right)  =\mathbf{S}_{\sqrt{X}%
}\left(  T\right)  ,
\]
since the differentiation of both sides of $X=\left\{  f_{1}\left(  X\right)
\right\}  ^{2}$ yields
\[
T=\left\{  \mathrm{D}f_{1}\left(  X\right)  \left(  T\right)  \right\}
\,X+X\,\left\{  \mathrm{D}f_{1}\left(  X\right)  \left(  T\right)  \right\}
.\,
\]

First,, consider $F_{\max}\left(  X,Y\right)  =\mathrm{tr}\,\sqrt
{Y^{1/2}XY^{1/2}}$. The derivative of $f_{2}\left(  X\right)  =\sqrt
{Y^{1/2}XY^{1/2}}$
\[
\mathrm{D}f_{2}\left(  X\right)  \left(  T\right)  =\mathbf{S}_{Y^{1/2}%
XY^{1/2}}\left(  Y^{1/2}\,T\,Y^{1/2}\right)  .
\]
Therefore,%

\begin{align*}
&  \mathrm{D}F_{\max}\left(  X,Y\right)  \left(  T,S\right)  \\
&  =\mathrm{tr}\,\,\mathbf{S}_{\sqrt{Y^{1/2}XY^{1/2}}}\left(  Y^{1/2}%
\,T\,Y^{1/2}\right)  +\mathrm{tr}\,\,\mathbf{S}_{\sqrt{X^{1/2}YX^{1/2}}%
}\left(  X^{1/2}S\,X^{1/2}\right)  \\
&  =\mathrm{tr}\,\mathbf{S}_{\sqrt{Y^{1/2}XY^{1/2}}}\left(  I\right)
Y^{1/2}TY^{1/2}+\mathrm{tr}\,\ \mathbf{S}_{\sqrt{X^{1/2}YX^{1/2}}}\left(
I\right)  X^{1/2}\,SX^{1/2}\\
&  =\mathrm{tr}\,Y^{1/2}\mathbf{S}_{\sqrt{Y^{1/2}XY^{1/2}}}\left(  I\right)
Y^{1/2}T+\mathrm{tr}\,\ X^{1/2}\mathbf{S}_{\sqrt{X^{1/2}YX^{1/2}}}\left(
I\right)  X^{1/2}\,S\\
&  =\frac{1}{2}\mathrm{tr}\,\,T\,Y^{1/2}\left(  Y^{1/2}X\,Y^{1/2}\right)
^{-1/2}Y^{1/2}+\frac{1}{2}\mathrm{tr}\,\,S\,X^{1/2}\left(  X^{1/2}%
Y\,X^{1/2}\right)  ^{-1/2}X^{1/2},
\end{align*}
and
\begin{align*}
L_{0,\ast} &  =\frac{1}{2}Y^{1/2}\left(  Y^{1/2}X\,Y^{1/2}\right)
^{-1/2}Y^{1/2},\\
L_{1,\ast} &  =\frac{1}{2}X^{1/2}\left(  X^{1/2}Y\,X^{1/2}\right)
^{-1/2}X^{1/2}.
\end{align*}
Here, `$\,\cdot^{-1}$' stands for generalized inverse. Observe
\begin{align*}
Y &  =4L_{0,\ast}\,X\,L_{0,\ast}\,,\\
X &  =4L_{1,\ast}\,Y\,L_{1,\ast}\,.\,
\end{align*}
Thus, \ $2L_{0,\ast}$ and $2L_{1,\ast}$ is non-commutative version of
Radon-Nikodym derivative $`\mathrm{d}\,\sqrt{Y}/\mathrm{d}\,\sqrt{X}^{\prime}$
and $`\mathrm{d}\,\sqrt{X}/\mathrm{d}\,\sqrt{Y}^{\prime}$, respectively.
Also,
\[
\left(  2L_{0,\ast}\right)  \left(  2L_{1,\ast}\right)  =I\mathbf{.}%
\]
Indeed,
\begin{equation}
F_{\max}\left(  X,Y\right)  =\min_{L:L>0}\mathrm{tr}\,LX+\mathrm{tr}%
\,L^{-1}Y.\label{F=L+1/L}%
\end{equation}
This is verified by differentiation of the right hand side:%
\[
\frac{\partial}{\partial L}\left\{  \mathrm{tr}\,LX+\mathrm{tr}\,L^{-1}%
Y\right\}  =X-L^{-1}YL^{-1}.
\]
So the minimum is achieved by a positive $L$ with
\[
Y=L\,X\,L.
\]
Thus, $L=2L_{0,\ast}$ .

Next, consider $F_{\min}\left(  X,Y\right)  $, supposing that $X>0$ and
$Y>0$,
\[
F_{\min}\left(  X,Y\right)  =\mathrm{tr}\,Y\sqrt{Y^{-1/2}XY^{-1/2}%
}=\mathrm{tr}\,X\sqrt{X^{-1/2}YX^{-1/2}}.
\]
So,
\begin{align*}
&  \mathrm{D}F_{\min}\left(  X,Y\right)  \left(  T,S\right) \\
&  =\mathrm{tr}\,\,Y\,\mathbf{S}_{\sqrt{Y^{-1/2}XY^{-1/2}}}\left(
Y^{-1/2}\,T\,Y^{-1/2}\right)  +\mathrm{tr}\,\,X\,\mathbf{S}_{\sqrt
{X^{-1/2}YX^{-1/2}}}\left(  X^{-1/2}S\,X^{-1/2}\right) \\
&  =\mathrm{tr}\,\,Y^{-1/2}\mathbf{S}_{\sqrt{Y^{-1/2}XY^{-1/2}}}\left(
Y\right)  \,Y^{-1/2}\,T\,+\mathrm{tr}\,\,X^{-1/2}\mathbf{S}_{\sqrt
{X^{-1/2}YX^{-1/2}}}\left(  X\right)  \,X^{-1/2}S.
\end{align*}
This means
\begin{align*}
L_{0,\ast}  &  =Y^{-1/2}\mathbf{S}_{\sqrt{Y^{-1/2}XY^{-1/2}}}\left(  Y\right)
\,Y^{-1/2},\\
L_{1,\ast}  &  =X^{-1/2}\mathbf{S}_{\sqrt{X^{-1/2}YX^{-1/2}}}\left(  X\right)
\,X^{-1/2}\,.
\end{align*}

Lastly, we consider $F_{1/2}\left(  X,Y\right)  $, where $\left(  X,Y\right)
\in\mathcal{P}_{k}^{\times2}$.%
\begin{align*}
&  \mathrm{D}F_{1/2}\left(  X,Y\right)  \left(  T,S\right)  \\
&  =\mathrm{tr}\,\mathbf{S}_{\sqrt{X}}\left(  T\right)  \sqrt{Y}%
+\mathrm{tr}\,\mathbf{S}_{\sqrt{Y}}\left(  S\right)  \sqrt{X}\\
&  =\mathrm{tr}\,T\mathbf{S}_{\sqrt{X}}\left(  \sqrt{Y}\right)  +\mathrm{tr}%
\,S\,\mathbf{S}_{\sqrt{Y}}\left(  \sqrt{X}\right)  .
\end{align*}
So,
\begin{equation}
L_{0,\ast}=\mathbf{S}_{\sqrt{X}}\left(  \sqrt{Y}\right)  ,\,L_{1,\ast
}=\mathbf{S}_{\sqrt{Y}}\left(  \sqrt{X}\right)  .\label{L=SY}%
\end{equation}
They give another non-commutative version of Radon-Nikodym derivative
$`\mathrm{d}\,\sqrt{Y}/\mathrm{d}\,\sqrt{X}^{\prime}$ and $`\mathrm{d}%
\,\sqrt{X}/\mathrm{d}\,\sqrt{Y}^{\prime}$.

\section{SDP representations}

It is known \cite{Killoran}\cite{Watrous}\ that
\begin{align}
F_{\max}\left(  X,Y\right)   &  =\max\,\left\{  \frac{1}{2}\left(
\mathrm{tr}\,C+\mathrm{tr}\,C^{\dagger}\right)  \,;\,\left[
\begin{array}
[c]{cc}%
X & C\\
C^{\dagger} & Y
\end{array}
\right]  \geq0\right\}  ,\label{Fmax-primal}\\
&  =\min\left\{  \mathrm{tr}\,X\,L_{0}+\mathrm{tr}\,Y\,L_{1}\,;\left(
L_{0},L_{1}\right)  \in\mathcal{M}_{F_{\max}}\right\}  , \label{Fmax-dual}%
\end{align}
where
\begin{equation}
\mathcal{M}_{F_{\max}}=\left\{  \left(  L_{0},L_{1}\right)  \,;\,\,\left[
\begin{array}
[c]{cc}%
2L_{0} & -I_{k}\\
-I_{k} & 2L_{1}%
\end{array}
\right]  \geq0,\,L_{0},L_{1}\in\mathcal{L}_{sa,k}\right\}  . \label{MFmax-2}%
\end{equation}
The equality between (\ref{Fmax-primal}) and (\ref{Fmax-dual}) is due to
duality theorem of semi definite programing. By Lemma\thinspace
\ref{lem:block-positive}, it is easy to verify
\begin{equation}
\mathcal{M}_{F_{\max}}=\left\{  \left(  L_{0},L_{1}\right)  \,;\,\,2L_{0}\geq
L,\,2L_{1}\geq L^{-1},\,\exists L\geq0\right\}  , \label{MFmax}%
\end{equation}
which leads to (\ref{F=L+1/L}). Conversely, (\ref{F=L+1/L}) leads to
(\ref{MFmax}).

Also, in the case of $\mathrm{supp}\,X\subset\mathrm{supp}\,Y$, it is known
\cite{Bhatia} that
\begin{equation}
\,\left[
\begin{array}
[c]{cc}%
X & C\\
C & Y
\end{array}
\right]  \geq0,\,C=C^{\dagger} \label{rho-sigma>0}%
\end{equation}
holds if and only if
\[
C\geq\sqrt{Y}\sqrt{Y^{-1/2}XY^{-1/2}}\sqrt{Y}\,.
\]
Therefore,
\begin{align}
F_{\min}\left(  X,Y\right)   &  =\max\,\left\{  \mathrm{tr}%
\,C\,;\,\text{(\ref{rho-sigma>0})},\,C\in\mathcal{L}_{sa,k}\right\}
\label{Fmin-primal}\\
&  =\min\left\{  \mathrm{tr}\,X\,L_{0}+\mathrm{tr}\,Y\,L_{1}\,;\left(
L_{0},L_{1}\right)  \in\mathcal{M}_{F_{\min}}\right\}  ,\label{Fmin-dual}\\
\mathcal{M}_{F_{\min}}  &  =\left\{  \left(  L_{0},L_{1}\right)
\,;\,\,\left[
\begin{array}
[c]{cc}%
2L_{0} & -I_{k}-\sqrt{-1}A\\
-I_{k}+\sqrt{-1}A & 2L_{1}%
\end{array}
\right]  \geq0,\,L_{0},L_{1},A\in\mathcal{L}_{sa,k}\right\}  , \label{MFmin}%
\end{align}
where the second identity is by the duality theorem of SDP.

In the case of $\mathrm{supp}\,X\not \subset \mathrm{supp}\,Y$, we still have
(\ref{Fmin-primal}), as proved in the following. By Lemma\thinspace
\ref{lem:block-positive}, $C$ should be supported on $\mathrm{supp}\,Y$, for
(\ref{rho-sigma>0}) to hold. Therefore,
\[
\,\left[
\begin{array}
[c]{cc}%
X & C\\
C & Y
\end{array}
\right]  =\left[
\begin{array}
[c]{cccc}%
X_{11} & X_{12} & C & 0\\
X_{21} & X_{22} & 0 & 0\\
C & 0 & Y & 0\\
0 & 0 & 0 & 0
\end{array}
\right]  \geq0.
\]
Because of $X\geq0$ and Lemma\thinspace\ref{lem:block-positive},
$X_{12}\,\left(  I_{k_{0}}-\pi_{\mathrm{supp}\,X_{22}}\right)  =0$, where
$k_{0}=\dim$ $\mathrm{supp}\,Y$. Therefore,
\begin{align*}
&  \left[
\begin{array}
[c]{cccc}%
I_{k_{0}} & -X_{12}X_{22}^{-1} & C & 0\\
0 & I_{k-k_{0}} & 0 & 0\\
C & 0 & I_{k_{0}} & 0\\
0 & 0 & 0 & I_{k-k_{0}}%
\end{array}
\right]  \left[
\begin{array}
[c]{cc}%
X & C\\
C & Y
\end{array}
\right]  \left[
\begin{array}
[c]{cccc}%
I_{k_{0}} & 0 & C & 0\\
-X_{22}^{-1}X_{21} & I_{k-k_{0}} & 0 & 0\\
C & 0 & I_{k_{0}} & 0\\
0 & 0 & 0 & I_{k-k_{0}}%
\end{array}
\right] \\
&  =\left[
\begin{array}
[c]{cccc}%
X_{11}-X_{12}X_{22}^{-1}X_{21} & X_{12}\,\left(  I_{k_{0}}-\pi_{\mathrm{supp}%
\,X_{22}}\right)  & C & 0\\
\left(  I_{k_{0}}-\pi_{\mathrm{supp}\,X_{22}}\right)  \,X_{21} & X_{22} & 0 &
0\\
C & 0 & Y & 0\\
0 & 0 & 0 & 0
\end{array}
\right] \\
&  =\left[
\begin{array}
[c]{cccc}%
\tilde{X} & 0 & C & 0\\
0 & X_{22} & 0 & 0\\
C & 0 & Y & 0\\
0 & 0 & 0 & 0
\end{array}
\right]  .
\end{align*}
Therefore, (\ref{rho-sigma>0}) is equivalent to
\begin{equation}
\left[
\begin{array}
[c]{cc}%
\tilde{X} & C\\
C & Y
\end{array}
\right]  \geq0,\,C=C^{\dagger}. \label{tilde-rho-sigma>0}%
\end{equation}
Thus,
\begin{align*}
&  \max\,\left\{  \mathrm{tr}\,C\,;\,\text{(\ref{rho-sigma>0})},\,C\in
\mathcal{L}_{sa,k}\right\} \\
&  =\max\,\left\{  \mathrm{tr}\,C\,;\,\text{(\ref{tilde-rho-sigma>0})}%
,\,C\in\mathcal{L}_{sa,k}\right\} \\
&  =\mathrm{tr}\,Y\sqrt{Y^{-1/2}\tilde{X}Y^{-1/2}}=F_{\min}\left(  X,Y\right)
,
\end{align*}
and our assertion is proved.

Note that, if $\left(  L_{0},L_{1}\right)  \in\mathcal{M}_{F_{\min}}$,
\ $L_{0}>0$ and $L_{1}>0$, thus
\begin{equation}
\mathcal{M}_{F_{\min}}=\left\{  \left(  L_{0},L_{1}\right)  \,;\,\,\left[
\begin{array}
[c]{cc}%
2L_{0} & -I_{k}-\sqrt{-1}A\\
-I_{k}+\sqrt{-1}A & 2L_{1}%
\end{array}
\right]  \geq0,\,L_{0}>0,L_{1}>0,A\in\mathcal{L}_{sa}\right\}  .
\label{MFmin-2}%
\end{equation}
Suppose otherwise, that is, $L_{0}$ has null eigenspace, and let $\left\vert
\psi\right\rangle $ be a member of it with unit length $\left\Vert
\psi\right\Vert =1$. Then,%

\begin{align*}
&  \left[
\begin{array}
[c]{cc}%
\left\langle \psi\right\vert  & c\left\langle \psi\right\vert
\end{array}
\right]  \left[
\begin{array}
[c]{cc}%
2L_{0} & -I_{k}-\sqrt{-1}A\\
-I_{k}+\sqrt{-1}A & 2L_{1}%
\end{array}
\right]  \left[
\begin{array}
[c]{c}%
\left\vert \psi\right\rangle \\
c\left\vert \psi\right\rangle
\end{array}
\right] \\
&  =-2c+2c^{2}\left\langle \psi\right\vert PL_{1}P\left\vert \psi
\right\rangle
\end{align*}
is negative if $c$ is sufficiently large positive number. So $L_{0}$ should be
strictly positive, and so should be $L_{1}$.

A consequence of SDP representations for $F_{\max}$ and $F_{\min}$ is
\begin{align}
F_{\min}\left(  X,Y\right)   &  =\min\left\{  \mathrm{tr}\,X\,L_{0}%
+\mathrm{tr}\,Y\,L_{1}\,;\left(  L_{0},\left(  I_{k}-\sqrt{-1}A\right)
^{-1}L_{1}\left(  I_{k}+\sqrt{-1}A\right)  ^{-1}\right)  \in\mathcal{M}%
_{F_{\max}},A\in\mathcal{L}_{sa,}\,\right\} \label{Fmin-I-A}\\
&  =\min\left\{  F_{\max}\left(  X,\left(  I_{k}-\sqrt{-1}A\right)  Y\left(
I_{k}+\sqrt{-1}A\right)  \right)  ;A\in\mathcal{L}_{sa}\,\right\}  .\nonumber
\end{align}
To show these, note that
\begin{align*}
&  \left[
\begin{array}
[c]{cc}%
I_{k} & 0\\
0 & \left(  I_{k}-\sqrt{-1}A\right)  ^{-1}%
\end{array}
\right]  \left[
\begin{array}
[c]{cc}%
2L_{0} & -I_{k}-\sqrt{-1}A\\
-I_{k}+\sqrt{-1}A & 2L_{1}%
\end{array}
\right]  \left[
\begin{array}
[c]{cc}%
I_{k} & 0\\
0 & \left(  I_{k}+\sqrt{-1}A\right)  ^{-1}%
\end{array}
\right] \\
&  =\left[
\begin{array}
[c]{cc}%
2L_{0} & -I_{k}\\
-I_{k} & \left(  I_{k}-\sqrt{-1}A\right)  ^{-1}\left(  2L_{1}\right)  \left(
I_{k}+\sqrt{-1}A\right)  ^{-1}%
\end{array}
\right]  .
\end{align*}
Here, $I_{k}-\sqrt{-1}A$ is invertible because $\left(  I_{k}-\sqrt
{-1}A\right)  ^{\dagger}\left(  I_{k}-\sqrt{-1}A\right)  $ is invertible,
\begin{align*}
&  \left(  I_{k}-\sqrt{-1}A\right)  ^{\dagger}\left(  I_{k}-\sqrt{-1}A\right)
\\
&  =I_{k}+A^{2}\geq I_{k}.
\end{align*}
Therefore,
\begin{align}
&  \left(  L_{0},L_{1}\right)  \in\mathcal{M}_{F_{\min}}\nonumber\\
&  \Leftrightarrow\exists A\in\mathcal{L}_{sa,k}\,\,\left(  L_{0},\left(
I_{k}-\sqrt{-1}A\right)  ^{-1}L_{1}\left(  I_{k}+\sqrt{-1}A\right)
^{-1}\right)  \in\mathcal{M}_{F_{\max}}. \label{MFmin-MFmax}%
\end{align}
Therefore, by (\ref{Fmax-dual}) and (\ref{Fmin-dual}), we have the asserted
identity. Similarly, we have
\[
F_{\min}\left(  X,Y\right)  =\min_{A\in\mathcal{L}_{sa,k}}F_{\max}\left(
\left(  I_{k}+\sqrt{-1}A\right)  X\left(  I_{k}-\sqrt{-1}A\right)  ,Y\right)
\,.
\]

Finally, we present a SDP representation of $F_{1/2}$. Define linear operators
$\mathbf{L}_{X}$ and $\mathbf{R}_{Y}$ on $\mathcal{L}_{k}$ by%
\[
\mathbf{L}_{X}\left(  A\right)  =XA,\,\mathbf{R}_{Y}\left(  A\right)  =AY.
\]
Then, if $X$ and $Y$ are self-adjoint, $\mathbf{L}_{X}$ and $\mathbf{R}_{Y}$
are self-adjoint when $\mathcal{L}_{k}$ is equipped with the Hilbert-Schmidt
inner product $\left\langle A,B\right\rangle _{\mathrm{HS}}:=\mathrm{tr}%
\,A^{\dagger}B$. Also, $\mathbf{L}_{X}$ and $\mathbf{R}_{Y}$ commutes, and%

\[
F_{1/2}\left(  X,Y\right)  =\left\langle I,\mathbf{L}_{X}^{1/2}\mathbf{R}%
_{Y}^{1/2}I\right\rangle _{\mathrm{HS}}.
\]
Hence,
\begin{align*}
F_{1/2}\left(  X,Y\right)   &  =\max\left\{  \operatorname{Re}\left\langle
I,\mathbf{C}I\right\rangle _{\mathrm{HS}}\,;\mathbf{C\leq L}_{X}%
^{1/2}\mathbf{R}_{Y}^{1/2},\right\} \\
&  =\max\left\{  \operatorname{Re}\left\langle I,\mathbf{C}I\right\rangle
_{\mathrm{HS}}\,;\left[
\begin{array}
[c]{cc}%
\mathbf{L}_{X} & \mathbf{C}\\
\mathbf{C} & \mathbf{R}_{Y}%
\end{array}
\right]  \geq0\right\}  .
\end{align*}

\section{Polar}

\subsection{Definition and Basic Properties}

We define \textit{polar} of $F^{Q}\in\mathcal{F}_{0}$ by
\[
\hat{F}^{Q}\left(  L_{0},L_{1}\right)  :=\left\{
\begin{array}
[c]{cc}%
\inf\left\{  s\,;\,\left(  L_{0},L_{1}\right)  \in s\,\left(  \mathcal{M}%
_{F^{Q}}\right)  ^{c},\,s>0\right\}  , & \left(  L_{0},\,L_{1}\right)
\in\mathcal{P}^{\times2},\\
-\infty, & \left(  L_{0},\,L_{1}\right)  \notin\mathcal{P}^{\times2},
\end{array}
\right.
\]
in analogy with a polar of a convex function. This is a sort of `dual' of
$F^{Q}$, and as is shown later, is CPTP monotone non-decreasing by
applications of any CP unital maps. We study the property of this quantity
rather meticulously.

\begin{proposition}
\label{prop:hatF>1}Suppose $F^{Q}\in\mathcal{F}_{0}$. If $\left(  L_{0}%
,L_{1}\right)  $ satisfies $\hat{F}^{Q}\left(  L_{0},L_{1}\right)  >0$, then
\ \ \
\begin{equation}
\hat{F}^{Q}\left(  L_{0},L_{1}\right)  =\max\left\{  s\,;\,\left(  L_{0}%
,L_{1}\right)  \in s\,\left(  \mathcal{M}_{F^{Q}}\right)  ,\,s>0\right\}  .
\label{hatF}%
\end{equation}
In particular, if $L_{0}$, $L_{1}>0$, we have (\ref{hatF}). Also,
\begin{equation}
\hat{F}^{Q}\left(  L_{0},L_{1}\right)  \geq1\,\Leftrightarrow\,\left(
L_{0},L_{1}\right)  \in\mathcal{M}_{F^{Q}}. \label{hatF>1}%
\end{equation}

If $F^{Q}\in\mathcal{F}_{0}$ satisfies normalization and CPTP monotonicity and
$L_{0}$ or $L_{1}$ has an eigenvalue $0$,
\begin{equation}
\hat{F}^{Q}\left(  L_{0},L_{1}\right)  =0. \label{hatF=0}%
\end{equation}

\end{proposition}

\begin{proof}
If $\left(  L_{0},L_{1}\right)  $ satisfies $\hat{F}^{Q}\left(  L_{0}%
,L_{1}\right)  >0$, the set
\[
\mathcal{T}:=\left\{  t;t\geq0,\,\left(  tL_{0},tL_{1}\right)  \in
\mathcal{M}_{F^{Q}}\right\}
\]
is not empty. Also, $\mathcal{T}$ is closed, since \ both $\mathcal{M}_{F^{Q}%
}$ and $\left\{  \left(  tL_{0},tL_{1}\right)  ;t\geq0\right\}  $ are closed.
By (\ref{L+M-inv}), if $t\in\mathcal{T}$, any $t^{\prime}\geq t$ is also an
element of the set $\mathcal{T}$. Therefore, there is $t_{0}>0$ such that
$t\in\mathcal{T}$ \ is equivalent to $t\geq t_{0}$. Therefore, we have
(\ref{hatF}). \ In particular, if $L_{0},L_{1}>0$, by Lemma\thinspace
\ref{lem:recession}, $\hat{F}^{Q}\left(  L_{0},L_{1}\right)  >0$.

Suppose $\hat{F}^{Q}\left(  L_{0},L_{1}\right)  \geq1\,$. Then, by definition,
$\left(  L_{0},L_{1}\right)  \in\mathcal{P}^{\times2}$. By (\ref{hatF}),
$\left(  L_{0},L_{1}\right)  \in\mathcal{M}_{F^{Q}}$. On the other hand,
suppose $\left(  L_{0},L_{1}\right)  \in\mathcal{M}_{F^{Q}}$. Then, by
(\ref{L+M-inv}),%
\[
\frac{1}{s}\left(  L_{0},L_{1}\right)  \in\mathcal{M}_{F^{Q}},\,0<\forall
s\leq1,
\]
which leads to $\hat{F}^{Q}\left(  L_{0},L_{1}\right)  \geq1\,$.

Suppose in addition $F^{Q}\in\mathcal{F}_{0}$ satisfies normalization and CPTP
monotonicity. Then, by Theorem\thinspace\ref{th:Fmin<FQ<F}, $\mathcal{M}%
_{F_{\max}}\subset\mathcal{M}_{F^{Q}}$. Therefore for any $L_{0}$ or $L_{1}$
with an eigenvalue $0$,
\[
\left(  tL_{0},tL_{1}\right)  \notin\mathcal{M}_{F_{\max}}\subset
\mathcal{M}_{F^{Q}},\forall t\geq0,
\]
which leads to (\ref{hatF=0}).
\end{proof}

\begin{proposition}
\label{prop:polar}For any $\left(  L_{0},L_{1}\right)  $,$\,\ \left(
X,Y\right)  \in$ $\mathcal{P}^{\times2}$, and $F^{Q}\in\mathcal{F}_{0}$,
\begin{equation}
\hat{F}^{Q}\left(  L_{0},L_{1}\right)  F^{Q}\left(  X,Y\right)  \leq
\mathrm{tr}\,L_{0}X+\mathrm{tr}\,L_{1}Y. \label{holder-ineq}%
\end{equation}
Also, for any $\left(  L_{0},L_{1}\right)  \in$ $\mathcal{P}^{\times2}$,
\begin{align}
\hat{F}^{Q}\left(  L_{0},L_{1}\right)   &  =\inf\left\{  \frac{1}{F^{Q}\left(
X,Y\right)  }\left(  \mathrm{tr}\,L_{0}X+\mathrm{tr}\,L_{1}Y\right)  ;\left(
X,Y\right)  \in\mathcal{P}^{\times2},F^{Q}\left(  X,Y\right)  >0\right\}
,\label{holder-1}\\
&  =\inf\left\{  \mathrm{tr}\,L_{0}X+\mathrm{tr}\,L_{1}Y;\left(  X,Y\right)
\in\mathcal{P}^{\times2},F^{Q}\left(  X,Y\right)  =1\right\}
\label{holder-1-1}\\
&  =\inf\left\{  \mathrm{tr}\,L_{0}X+\mathrm{tr}\,L_{1}Y;\left(  X,Y\right)
\in\mathcal{P}^{\times2},F^{Q}\left(  X,Y\right)  \geq1\right\}  ,
\label{holder-1-2}%
\end{align}
and \ $\hat{F}^{Q}$ is a member of $\mathcal{F}_{0}$. \ If in addition $F^{Q}$
is not identically $0$ on $\mathcal{P}^{\times2}$, $\hat{F}^{Q}$ also has that property.
\end{proposition}

\begin{proof}
First, we show (\ref{holder-1}). Observe
\begin{align*}
\hat{F}^{Q}\left(  L_{0},L_{1}\right)   &  =\inf\left\{  s\,;\,\left(
L_{0},L_{1}\right)  \in s\,\left(  \mathcal{M}_{F^{Q}}\right)  ^{c}%
,\,s>0\right\} \\
&  =\inf\left\{  s\,;\,\frac{1}{s}\left(  L_{0},L_{1}\right)  \in\left(
\mathcal{M}_{F^{Q}}\right)  ^{c},\,s>0\right\} \\
&  =\inf\left\{  s\,;\,\exists\left(  X,Y\right)  \in\mathcal{P}^{\times
2}\,,\,\frac{1}{s}\left(  \mathrm{tr}\,L_{0}X+\mathrm{tr}\,L_{1}Y\right)
<F^{Q}\left(  X,Y\right)  ,\,s>0\right\}  .
\end{align*}
Since $\frac{1}{s}\left(  \mathrm{tr}\,L_{0}X+\mathrm{tr}\,L_{1}Y\right)
\geq0$, $\frac{1}{s}\left(  \mathrm{tr}\,L_{0}X+\mathrm{tr}\,L_{1}Y\right)
<F^{Q}\left(  X,Y\right)  $ holds only if $F^{Q}\left(  X,Y\right)  $ is
positive. Therefore,
\[
\hat{F}^{Q}\left(  L_{0},L_{1}\right)  =\inf\left\{  s\,;\exists\left(
X,Y\right)  \in\mathcal{P}^{\times2}\,,\,\,s\,>\frac{1}{F^{Q}\left(
X,Y\right)  }\left(  \mathrm{tr}\,L_{0}X+\mathrm{tr}\,L_{1}Y\right)
,\,F^{Q}\left(  X,Y\right)  >0\right\}  ,
\]
which implies (\ref{holder-1}). (\ref{holder-ineq}) results from
(\ref{holder-1}).

Next, (\ref{holder-1}) is equivalent to
\[
\hat{F}^{Q}\left(  L_{0},L_{1}\right)  =\inf\left\{  \mathrm{tr}%
\,L_{0}X+\mathrm{tr}\,L_{1}Y;\left(  X,Y\right)  \in\mathcal{P}^{\times
2},F^{Q}\left(  X,Y\right)  \neq0,-\infty\right\}  ,
\]
which is equal to (\ref{holder-1-1}). Let $\mathcal{P}_{F^{Q}}$ be the set of
$\left(  X,Y\right)  \in\mathcal{P}^{\times2}$ such that $F^{Q}\left(
X,Y\right)  =1$. Then we obtain (\ref{holder-1-2}) as follows.
\begin{align*}
\hat{F}^{Q}\left(  L_{0},L_{1}\right)   &  =\inf\left\{  \mathrm{tr}%
\,L_{0}X+\mathrm{tr}\,L_{1}Y;\left(  X,Y\right)  \in\mathcal{P}_{F^{Q}%
}\right\} \\
&  =\inf\left\{  \mathrm{tr}\,L_{0}X+\mathrm{tr}\,L_{1}Y;\left(  X,Y\right)
\in\mathrm{conv}\,\mathcal{P}_{F^{Q}}\right\} \\
&  =\inf\left\{  \mathrm{tr}\,L_{0}X+\mathrm{tr}\,L_{1}Y;\left(  X,Y\right)
\in\mathcal{P}^{\times2},\,F^{Q}\left(  X,Y\right)  \geq1\right\}  ,
\end{align*}
where the last identity is due to concavity of $F^{Q}$.

By (\ref{holder-1-2}) and Lemma\thinspace\ref{lem:sublinear-2}, $\hat{F}^{Q}$
is closed, proper, concave, and positively homogenous. Since $\hat{F}%
^{Q}\left(  L_{0},L_{1}\right)  \geq0$ for all $\left(  L_{0},L_{1}\right)
\in\mathcal{P}^{\times2}$, $\ \hat{F}^{Q}$ is a member of $F^{Q}$.

Next, we show $\hat{F}^{Q}$ is not identically $0$ on $\mathcal{P}^{\times2}$.
If $\left(  L_{0},L_{1}\right)  \in\mathcal{M}_{F^{Q}}\,$, by
Proposition\thinspace\ref{prop:hatF>1}, $\hat{F}^{Q}\left(  L_{0}%
,L_{1}\right)  \geq1$. Therefore, if $\hat{F}^{Q}$ is identically $0$ on
$\mathcal{P}^{\times2}$, $\mathcal{M}_{F^{Q}}=\mathcal{M}_{F^{Q}}\cap$
$\mathcal{P}^{\times2}$ is empty. By \ref{lem:sublinear}, this contradicts
with the assumption that $F^{Q}\left(  X,Y\right)  $ is a member of
$\mathcal{F}_{0}$.
\end{proof}

\begin{theorem}
\label{th:hatF-F0}If $F^{Q}$ is a member of $\mathcal{F}_{0}$ and not
identically $0$ on $\mathcal{P}^{\times2}$, so is $\hat{F}^{Q}$. Also,%
\begin{align}
F^{Q}\left(  X,Y\right)   &  =\inf\left\{  \frac{1}{\hat{F}^{Q}\left(
L_{0},L_{1}\right)  }\left(  \mathrm{tr}\,L_{0}X+\mathrm{tr}\,L_{1}Y\right)
;\left(  L_{0},L_{1}\right)  \in\mathcal{P}^{\times2},\hat{F}^{Q}\left(
L_{0},L_{1}\right)  >0\right\} \label{holder-2}\\
&  =\inf\left\{  \mathrm{tr}\,L_{0}X+\mathrm{tr}\,L_{1}Y;\left(  L_{0}%
,L_{1}\right)  \in\mathcal{P}^{\times2},\hat{F}^{Q}\left(  L_{0},L_{1}\right)
=1\right\}  \label{holder-2-1}%
\end{align}

\end{theorem}

\begin{equation}
F^{Q}\left(  X,Y\right)  =\inf\left\{  s\,;\,\left(  X,Y\right)  \in
s\,\left(  \mathcal{M}_{\hat{F}^{Q}}\right)  ^{c},\,s>0\right\}  ,
\label{hathatF}%
\end{equation}
and
\begin{equation}
\,F^{Q}\left(  X,Y\right)  \geq1\Leftrightarrow\left(  X,Y\right)
\in\mathcal{M}_{\hat{F}^{Q}}. \label{F>1}%
\end{equation}

\begin{proof}
By (\ref{hatF>1}) and the definition of $\mathcal{M}_{F^{Q}}$, \
\begin{align*}
F^{Q}\left(  X,Y\right)   &  =\inf\left\{  \mathrm{tr}\,L_{0}X+\mathrm{tr}%
\,L_{1}Y\,;\,\left(  L_{0},L_{1}\right)  \in\mathcal{M}_{F^{Q}}\right\} \\
&  =\inf\left\{  \mathrm{tr}\,L_{0}X+\mathrm{tr}\,L_{1}Y\,;\,\hat{F}%
^{Q}\left(  L_{0},L_{1}\right)  \geq1\right\} \\
&  =\inf\left\{  \mathrm{tr}\,L_{0}X+\mathrm{tr}\,L_{1}Y\,;\,\hat{F}%
^{Q}\left(  L_{0},L_{1}\right)  =1\right\}  .
\end{align*}
The last end of this is equal to the RHS of (\ref{holder-2}), due to the
almost parallel reason as (\ref{holder-1-1}) equals (\ref{holder-1-2}).

(\ref{hathatF}) is derived from (\ref{holder-2}) in almost parallel manner as
the proof of (\ref{holder-1}). (\ref{F>1}) results from (\ref{holder-2-1}) and
concavity of $\hat{F}^{Q}\left(  L_{0},L_{1}\right)  $.
\end{proof}

\begin{corollary}
\label{cor:1-to-1} (\ref{holder-1}) establishes one-to-one map from
$\mathcal{F}_{0}$ to itself, whose inverse map is given by (\ref{holder-2}).
\end{corollary}

\begin{proof}
The mapping from $F^{Q}\in\mathcal{F}_{0}$ to $\hat{F}^{Q}\in\mathcal{F}_{0}$
is one-to-one since $F^{Q}$ is recovered from $\hat{F}^{Q}$ using
(\ref{holder-2}).
\end{proof}

\begin{proposition}
\label{prop:hatF-s-hom}$F^{Q}\in\mathcal{F}_{0}$ is strongly homogeneous if
and only if $\hat{F}^{Q}\in\mathcal{F}_{0}$ is strongly homogeneous.
\end{proposition}

\begin{proof}
Suppose $F^{Q}\in\mathcal{F}_{0}$ is strongly homogeneous. With $t_{0}>0$, and
$t_{1}>0$,
\begin{align*}
\hat{F}^{Q}\left(  t_{0}L_{0},t_{1}L_{1}\right)   &  =\inf\left\{
s\,;\,\left(  t_{0}L_{0},t_{1}L_{1}\right)  \in s\,\left(  \mathcal{M}_{F^{Q}%
}\right)  ^{c},\,s>0\right\} \\
&  =\inf\left\{  s\,;\,\sqrt{t_{0}t_{1}}\left(  \sqrt{\frac{t_{0}}{t_{1}}%
}L_{0},\sqrt{\frac{t_{1}}{t_{0}}}L_{1}\right)  \in s\,\left(  \mathcal{M}%
_{F^{Q}}\right)  ^{c},\,s>0\right\} \\
&  =\inf\left\{  s\,;\,\sqrt{t_{0}t_{1}}\left(  L_{0},L_{1}\right)  \in
s\,\left(  \mathcal{M}_{F^{Q}}\right)  ^{c},\,s>0\right\} \\
&  =\,\sqrt{t_{0}t_{1}}\hat{F}^{Q}\left(  L_{0},L_{1}\right)  ,
\end{align*}
where the third identity is due to Proposition\thinspace\ref{prop:t-1/t}. The
opposite implication is proved in almost parallel manner.
\end{proof}

\begin{proposition}
\label{prop:hatF-monotone} $F^{Q}\in\mathcal{F}_{0}$ satisfies CPTP
monotonicity if and only if $\hat{F}^{Q}$ is monotone increasing by
application of any unital CP map $\Lambda^{\ast}$,
\[
\hat{F}^{Q}\left(  \Lambda^{\ast}\left(  L_{0}\right)  ,\Lambda^{\ast}\left(
L_{1}\right)  \right)  \geq\hat{F}^{Q}\left(  L_{0},L_{1}\right)  .
\]

\end{proposition}

\begin{proof}
Suppose $F^{Q}\in\mathcal{F}_{0}$ satisfies CPTP monotonicity. If $\hat{F}%
^{Q}\left(  L_{0},L_{1}\right)  =0$, the assertion is trivial. Thus, we
suppose $\hat{F}^{Q}\left(  L_{0},L_{1}\right)  >0$ and use (\ref{hatF}). If
\ $\left(  L_{0},L_{1}\right)  \in s\,\mathcal{M}_{F^{Q}}$, \ by
Proposition\thinspace\ref{prop:monotone}, \thinspace we have%
\begin{align*}
\left(  \Lambda^{\ast}\left(  L_{0}\right)  ,\Lambda^{\ast}\left(
L_{1}\right)  \right)   &  \in\Lambda^{\ast}\left(  s\,\mathcal{M}_{F^{Q}%
}\right)  \\
&  =s\Lambda^{\ast}\left(  \,\mathcal{M}_{F^{Q}}\right)  \subset
s\,\mathcal{M}_{F^{Q}},
\end{align*}
which implies the assertion. The opposite implication is proved in almost
parallel manner.
\end{proof}

\begin{proposition}
\label{prp:direct-min} $F^{Q}\in\mathcal{F}_{0}$ satisfies CPTP monotonicity
and additivity if and only if $\hat{F}^{Q}\in\mathcal{F}_{0}$ is monotone
increasing by any unital CP map and satisfies%
\begin{equation}
\hat{F}^{Q}\left(  L_{0}^{\left(  1\right)  }\oplus L_{0}^{\left(  2\right)
},L_{1}^{\left(  1\right)  }\oplus L_{1}^{\left(  2\right)  }\right)
=\min\left\{  \hat{F}^{Q}\left(  L_{0}^{\left(  1\right)  },L_{1}^{\left(
1\right)  }\right)  ,\hat{F}^{Q}\left(  L_{0}^{\left(  2\right)  }%
,L_{1}^{\left(  2\right)  }\right)  \right\}  . \label{direct-min}%
\end{equation}

\end{proposition}

\begin{proof}
Suppose $F^{Q}\in\mathcal{F}_{0}$ is CPTP monotone and additive. Then by
Proposition\thinspace\ref{prop:hatF-monotone}, $\hat{F}^{Q}$ is monotone
increasing by any unital CP map. $\mathcal{M}_{\hat{F}^{Q}}$ is invariant by
any unital CP map. Therefore, for any $\left(  X,Y\right)  \in\mathcal{M}%
_{\hat{F}^{Q}}$ with
\begin{equation}
X=\left[
\begin{array}
[c]{cc}%
X_{1} & \ast\\
\ast & X_{2}%
\end{array}
\right]  ,\,Y=\left[
\begin{array}
[c]{cc}%
Y_{1} & \ast\\
\ast & Y_{2}%
\end{array}
\right]  , \label{XY}%
\end{equation}
we have $\left(  X_{1}\oplus X_{2},Y_{1}\oplus Y_{2}\right)  \in
\mathcal{M}_{\hat{F}^{Q}}$ (Consider the pinching operation that maps $\left(
X,Y\right)  $ to $\left(  X_{1}\oplus X_{2},Y_{1}\oplus Y_{2}\right)  $).
Therefore,
\begin{align*}
&  \hat{F}^{Q}\left(  L_{0}^{\left(  1\right)  }\oplus L_{0}^{\left(
2\right)  },L_{1}^{\left(  1\right)  }\oplus L_{1}^{\left(  2\right)  }\right)
\\
&  =\inf_{\left(  X,Y\right)  \in\mathcal{M}_{\hat{F}^{Q}}}\mathrm{tr}%
\,\left(  L_{0}^{\left(  1\right)  }\oplus L_{0}^{\left(  2\right)  }\right)
X+\mathrm{tr}\,\left(  L_{1}^{\left(  1\right)  }\oplus L_{1}^{\left(
2\right)  }\right)  Y\\
&  =\inf_{\left(  X,Y\right)  \in\mathcal{M}_{\hat{F}^{Q}}}\mathrm{tr}%
\,L_{0}^{\left(  1\right)  }X_{1}+\mathrm{tr}\,L_{0}^{\left(  2\right)  }%
X_{2}+\mathrm{tr}\,L_{1}^{\left(  1\right)  }Y_{1}+\mathrm{tr}\,L_{1}^{\left(
2\right)  }Y_{2}.\\
&  =\inf_{\left(  X_{1}\oplus X_{2},Y_{1}\oplus Y_{2}\right)  \in
\mathcal{M}_{\hat{F}^{Q}}}\mathrm{tr}\,L_{0}^{\left(  1\right)  }%
X_{1}+\mathrm{tr}\,L_{0}^{\left(  2\right)  }X_{2}+\mathrm{tr}\,L_{1}^{\left(
1\right)  }Y_{1}+\mathrm{tr}\,L_{1}^{\left(  2\right)  }Y_{2}%
\end{align*}

Also, by (\ref{hatF>1}),
\begin{align*}
&  \left(  X_{1}\oplus X_{2},Y_{1}\oplus Y_{2}\right)  \in\mathcal{M}_{\hat
{F}^{Q}}\\
&  \Leftrightarrow F^{Q}\left(  X_{1}\oplus X_{2},Y_{1}\oplus Y_{2}\right)
=F^{Q}\left(  X_{1},Y_{1}\right)  +F^{Q}\left(  X_{2},Y_{2}\right)  \geq1,\\
&  \Leftrightarrow F^{Q}\left(  X_{1},Y_{1}\right)  \geq\lambda,\,F^{Q}\left(
X_{2},Y_{2}\right)  \geq1-\lambda,\,0\leq\exists\lambda\leq1.\\
&  \Leftrightarrow\left(  X_{1},Y_{1}\right)  =\lambda\left(  \tilde{X}%
_{1},\tilde{Y}_{1}\right)  ,\,\,\left(  X_{2},Y_{2}\right)  =\left(
1-\lambda\right)  \left(  \tilde{X}_{2},\tilde{Y}_{2}\right)  ,\\
&  \,\,\exists\left(  \tilde{X}_{1},\tilde{Y}_{1}\right)  ,\left(  \tilde
{X}_{2},\tilde{Y}_{2}\right)  \in\mathcal{M}_{\hat{F}^{Q}},\,0\leq
\exists\lambda\leq1.
\end{align*}
Therefore,
\begin{align*}
&  \hat{F}^{Q}\left(  L_{0}^{\left(  1\right)  }\oplus L_{0}^{\left(
2\right)  },L_{1}^{\left(  1\right)  }\oplus L_{1}^{\left(  2\right)  }\right)
\\
&  =\inf_{0\leq\lambda\leq1}\inf_{\left(  \tilde{X}_{1},\tilde{Y}_{1}\right)
\in\mathcal{M}_{\hat{F}^{Q}}}\inf_{\left(  \tilde{X}_{2},\tilde{Y}_{2}\right)
\in\mathcal{M}_{\hat{F}^{Q}}}\lambda\left(  \mathrm{tr}\,L_{0}^{\left(
1\right)  }\tilde{X}_{1}+\mathrm{tr}L_{1}^{\left(  1\right)  }\tilde{Y}%
_{1}\right)  +\left(  1-\lambda\right)  \left(  \mathrm{tr}\,L_{0}^{\left(
2\right)  }\tilde{X}_{2}+\mathrm{tr}\,L_{1}^{\left(  2\right)  }\tilde{Y}%
_{2}\right) \\
&  =\inf_{0\leq\lambda\leq1}\left\{  \lambda\hat{F}^{Q}\left(  L_{0}^{\left(
1\right)  },L_{1}^{\left(  1\right)  }\right)  +\left(  1-\lambda\right)
\hat{F}^{Q}\left(  L_{0}^{\left(  2\right)  },L_{1}^{\left(  2\right)
}\right)  \right\} \\
&  =\min\left\{  \hat{F}^{Q}\left(  L_{0}^{\left(  1\right)  },L_{1}^{\left(
1\right)  }\right)  ,\hat{F}^{Q}\left(  L_{0}^{\left(  2\right)  }%
,L_{1}^{\left(  2\right)  }\right)  \right\}  .
\end{align*}

Conversely, suppose $\hat{F}^{Q}\in\mathcal{F}_{0}$ is monotone increasing by
any unital CP map and satisfies (\ref{direct-min}). Then by
Proposition\thinspace\ref{prop:hatF-monotone}, $F^{Q}\in\mathcal{F}_{0}$ is
CPTP monotone. Therefore, $\left(  L_{0}^{\left(  1\right)  }\oplus
L_{0}^{\left(  2\right)  },L_{1}^{\left(  1\right)  }\oplus L_{1}^{\left(
2\right)  }\right)  $ is a member of $\mathcal{M}_{F^{Q}}$ if and only if
there exists $\left(  L_{0},L_{1}\right)  \in\mathcal{M}_{\hat{F}^{Q}}$ such
that
\[
L_{\theta}=\left[
\begin{array}
[c]{cc}%
L_{\theta}^{\left(  1\right)  } & \ast\\
\ast & L_{\theta}^{\left(  2\right)  }%
\end{array}
\right]  ,\theta\in\left\{  0,1\right\}  .\,
\]
Also, by (\ref{hatF>1}),
\begin{align*}
&  \left(  L_{0}^{\left(  1\right)  }\oplus L_{0}^{\left(  2\right)  }%
,L_{1}^{\left(  1\right)  }\oplus L_{1}^{\left(  2\right)  }\right)
\in\mathcal{M}_{F^{Q}}\\
&  \Leftrightarrow\hat{F}^{Q}\left(  L_{0}^{\left(  1\right)  }\oplus
L_{0}^{\left(  2\right)  },L_{1}^{\left(  1\right)  }\oplus L_{1}^{\left(
2\right)  }\right)  =\min\left\{  \hat{F}^{Q}\left(  L_{0}^{\left(  1\right)
},L_{1}^{\left(  1\right)  }\right)  ,\hat{F}^{Q}\left(  L_{0}^{\left(
2\right)  },L_{1}^{\left(  2\right)  }\right)  \right\}  \geq1\\
&  \Leftrightarrow\hat{F}^{Q}\left(  L_{0}^{\left(  1\right)  },L_{1}^{\left(
1\right)  }\right)  \geq1\text{ and }\hat{F}^{Q}\left(  L_{0}^{\left(
2\right)  },L_{1}^{\left(  2\right)  }\right)  \geq1\\
&  \Leftrightarrow\left(  L_{0}^{\left(  1\right)  },L_{1}^{\left(  1\right)
}\right)  ,\left(  L_{0}^{\left(  2\right)  },L_{1}^{\left(  2\right)
}\right)  \in\mathcal{M}_{F^{Q}}.
\end{align*}
Therefore,
\begin{align*}
&  F^{Q}\left(  X_{1}\oplus X_{2},Y_{1}\oplus Y_{2}\right) \\
&  =\inf_{\left(  L_{0},L_{1}\right)  \in\mathcal{M}_{F^{Q}}}\mathrm{tr}%
\,L_{0}\left(  X_{1}\oplus X_{2}\right)  +\mathrm{tr}\,L_{1}\left(
Y_{1}\oplus Y_{2}\right) \\
&  =\inf_{\left(  L_{0}^{\left(  1\right)  }\oplus L_{0}^{\left(  2\right)
},L_{1}^{\left(  1\right)  }\oplus L_{1}^{\left(  2\right)  }\right)
\in\mathcal{M}_{F^{Q}}}\mathrm{tr}\,L_{0}^{\left(  1\right)  }X_{1}%
+\mathrm{tr}\,L_{0}^{\left(  2\right)  }X_{2}+\mathrm{tr}\,L_{1}^{\left(
1\right)  }Y_{1}+\mathrm{tr}\,L_{1}^{\left(  2\right)  }Y_{2}\\
&  =\inf_{\left(  L_{1}^{\left(  1\right)  },L_{1}^{\left(  2\right)
}\right)  ,\left(  L_{1}^{\left(  1\right)  },L_{1}^{\left(  2\right)
}\right)  \in\mathcal{M}_{F^{Q}}}\left(  \mathrm{tr}\,L_{0}^{\left(  1\right)
}X_{1}+\mathrm{tr}\,L_{1}^{\left(  1\right)  }Y_{1}\right)  +\left(
\mathrm{tr}\,L_{0}^{\left(  2\right)  }X_{2}+\mathrm{tr}\,L_{1}^{\left(
2\right)  }Y_{2}\right) \\
&  =F^{Q}\left(  X_{1},Y_{1}\right)  +F^{Q}\left(  X_{2},Y_{2}\right)  .
\end{align*}

\end{proof}

\subsection{Classical version and $\hat{F}_{\min}$, $\hat{F}_{\max}$}

For real vectors $l_{0}:=\left(  l_{0,i}\right)  _{i=1}^{k}$ and
$l_{1}:=\left(  l_{1,i}\right)  _{i=1}^{k}$ with positive components we can
define $\hat{F}^{C}\left(  l_{0},l_{1}\right)  $ in analogy with $\hat{F}^{Q}%
$; First, let
\[
\mathcal{M}_{C}:=%
{\displaystyle\bigcup\limits_{k=1}^{\infty}}
\left\{  \left(  l_{0},l_{1}\right)  ;\,F^{C}\left(  x,y\right)  \leq
\sum_{i=1}^{k}l_{0,i}x_{i}+\sum_{i=1}^{k}l_{1,i}y_{i}\right\}  .
\]
Then
\[
F^{C}\left(  x,y\right)  =\inf\left\{  \sum_{i=1}^{k}l_{0,i}x_{i}+\sum
_{i=1}^{k}l_{1,i}y_{i}\,;\,\left(  l_{0},l_{1}\right)  \in\mathcal{M}%
_{C}\right\}  .
\]
Thus we define
\[
\hat{F}^{C}\left(  l_{0},l_{1}\right)  :=\inf\left\{  s\,;\,\left(
l_{0},l_{1}\right)  \in s\,\left(  \mathcal{M}_{C}\right)  ^{c},\,s>0\right\}
.
\]
$\hat{F}^{C}$ is concave, positively homogeneous, proper, and monotone
increasing by transpose of stochastic map. Also,
\begin{equation}
\hat{F}^{C}\left(  l_{0},l_{1}\right)  =\inf\left\{  \frac{1}{\,F^{C}\left(
x,y\right)  }\left(  \sum_{i=1}^{k}l_{0,i}x_{i}+\sum_{i=1}^{k}l_{1,i}%
y_{i}\right)  ;x,y\in%
\mathbb{R}
_{\geq0}^{k},\,F^{C}\left(  x,y\right)  >0\right\}  . \label{holder-cl}%
\end{equation}
In addition, it is additive in the sense that%

\begin{align}
\hat{F}^{C}\left(  l_{0},l_{1}\right)   &  =\hat{F}^{C}\left(  l_{0}^{\left(
1\right)  },l_{1}^{\left(  2\right)  }\right)  +\hat{F}^{C}\left(
l_{0}^{\left(  1\right)  },l_{1}^{\left(  2\right)  }\right)
,\label{additive-c}\\
l_{\theta}^{\left(  1\right)  }  &  :=\left(  l_{\theta,1},l_{\theta,2}%
,\cdots,l_{\theta,k^{\prime}}\right)  ,\nonumber\\
l_{\theta}^{\left(  2\right)  }  &  :=\left(  l_{\theta,k^{\prime}%
+1},l_{\theta,k^{\prime}+2},\cdots,l_{\theta,k}\right)  ,\,\,\theta
=0,1.\nonumber
\end{align}
The proof of these properties is almost parallel as the proof of analogous
properties of $\hat{F}^{Q}$, thus omitted.

By additivity (\ref{additive-c}),
\begin{align}
\hat{F}^{C}\left(  l_{0},l_{1}\right)   &  =\min_{i}\hat{F}^{C}\left(
l_{0,i},l_{1,i}\right) \nonumber\\
&  =\min_{i}\inf\left\{  \frac{1}{\,\sqrt{xy}}\left(  l_{0,i}x+l_{1,i}%
y\right)  ;\,x,y\in%
\mathbb{R}
_{>0}\right\} \nonumber\\
&  =\min_{i}\inf\left\{  l_{0,i}\sqrt{\frac{x}{y}}+l_{1,i}\sqrt{\frac{y}{x}%
};\,x,y\in%
\mathbb{R}
_{>0}\right\} \nonumber\\
&  =\min_{i}2\sqrt{l_{0,i}l_{1,i}}. \label{Fcl-explicit}%
\end{align}

We say a functional over $\mathcal{P}^{\times2}$ is \textit{polarly
normalized} if, for any $\left(  l_{0},l_{1}\right)  \in%
\mathbb{R}
_{\geq0}^{k}\times%
\mathbb{R}
_{\geq0}^{k}$ and for an orthonormal basis $\left\{  \left\vert i\right\rangle
;i=1,\cdots,k\right\}  $,
\[
\hat{F}^{Q}\left(  L_{0,c},L_{1,c}\right)  =\hat{F}^{C}\left(  l_{0}%
,l_{1}\right)  ,
\]
where \
\begin{equation}
L_{0,c}:=\sum_{i=1}^{k}l_{0,i}\left\vert i\right\rangle \left\langle
i\right\vert ,\,L_{1,c}:=\sum_{i=1}^{k}l_{1,i}\left\vert i\right\rangle
\left\langle i\right\vert . \label{Lc}%
\end{equation}

Below, $\Gamma_{1}$ is the pinching by the basis $\left\{  \left\vert
i\right\rangle ;i=1,\cdots,k\right\}  $.

\begin{proposition}
\label{prop:polar-normalize}Suppose that $F^{Q}$ is a CPTP monotone and
normalized member of $\mathcal{F}_{0}$. Then, $\hat{F}^{Q}$ is polarly normalized.
\end{proposition}

\begin{proof}
Observe $L_{0,c}$ and $L_{1,c}$ as of (\ref{Lc}) are unchanged by $\Gamma_{1}%
$. Hence, by (\ref{holder-cl}), and by CPTP monotonicity and normalization of
$F^{Q}$,
\begin{align*}
&  \hat{F}^{Q}\left(  L_{0,c},L_{1,c}\right) \\
&  =\inf\left\{  \frac{1}{F^{Q}\left(  X,Y\right)  }\left(  \mathrm{tr}%
\,\Gamma_{1}\left(  L_{0,c}\right)  X+\mathrm{tr}\,\Gamma_{1}\left(
L_{1,c}\right)  Y\right)  ;\left(  X,Y\right)  \in\mathcal{P}^{\times2}%
,F^{Q}\left(  X,Y\right)  >0\right\} \\
&  =\inf\left\{  \frac{1}{F^{Q}\left(  X,Y\right)  }\left(  \mathrm{tr}%
\,L_{0,c}\Gamma_{1}\left(  X\right)  +\mathrm{tr}\,L_{1,c}\Gamma_{1}\left(
Y\right)  \right)  ;\left(  X,Y\right)  \in\mathcal{P}^{\times2},F^{Q}\left(
X,Y\right)  >0\right\} \\
&  \geq\inf\left\{  \frac{1}{F^{Q}\left(  \Gamma_{1}\left(  X\right)
,\Gamma_{1}\left(  Y\right)  \right)  }\left(  \mathrm{tr}\,L_{0,c}\Gamma
_{1}\left(  X\right)  +\mathrm{tr}\,L_{1,c}\Gamma_{1}\left(  Y\right)
\right)  ;\left(  X,Y\right)  \in\mathcal{P}^{\times2},F^{Q}\left(  \Gamma
_{1}\left(  X\right)  ,\Gamma_{1}\left(  Y\right)  \right)  >0\right\} \\
&  =\inf\left\{  \frac{1}{\,F^{C}\left(  x,y\right)  }\left(  \sum_{i=1}%
^{k}l_{0,i}x_{i}+\sum_{i=1}^{k}l_{1,i}y_{i}\right)  ;\left(  x,y\right)  \in%
\mathbb{R}
_{\geq0}^{k}\,,\,F^{C}\left(  x,y\right)  >0\right\} \\
&  =\hat{F}^{C}\left(  l_{0},l_{1}\right)  .
\end{align*}
Also, by
\begin{align*}
&  \hat{F}^{Q}\left(  L_{0,c},L_{1,c}\right) \\
&  \leq\inf\left\{  \frac{1}{F^{Q}\left(  X,Y\right)  }\left(  \mathrm{tr}%
\,L_{0,c}X+\mathrm{tr}\,L_{1,c}Y\right)  ;\left(  X,Y\right)  \in
\mathcal{P}^{\times2},F^{Q}\left(  X,Y\right)  >0,X,Y\text{: diagonal}%
\right\}
\end{align*}
and by normalization of $F^{Q}$,
\begin{align*}
&  \hat{F}^{Q}\left(  L_{0,c},L_{1,c}\right) \\
&  \leq\inf\left\{  \frac{1}{\,F^{C}\left(  x,y\right)  }\left(  \sum
_{i=1}^{k}l_{0,i}x_{i}+\sum_{i=1}^{k}l_{1,i}y_{i}\right)  ;\left(  x,y\right)
\in%
\mathbb{R}
_{\geq0}^{k}\,,\,F^{C}\left(  x,y\right)  >0\right\} \\
&  =\hat{F}^{C}\left(  l_{0},l_{1}\right)  .
\end{align*}
After all, we have polar normalization.
\end{proof}

Below, we show%

\begin{align}
\hat{F}_{\max}\left(  L_{0},L_{1}\right)   &  =2\sqrt{\min_{\left\Vert
\psi\right\Vert =1}\left\langle \psi\right\vert L_{1}^{1/2}L_{0}L_{1}%
^{1/2}\left\vert \psi\right\rangle }\nonumber\\
&  =2\sqrt{\min_{\left\Vert \psi\right\Vert =1}\left\langle \psi\right\vert
L_{0}^{1/2}L_{1}L_{0}^{1/2}\left\vert \psi\right\rangle }\nonumber\\
&  =2\min_{\left\Vert \psi\right\Vert =1,\left\Vert \varphi\right\Vert
=1}\left\vert \left\langle \psi\right\vert L_{0}^{1/2}L_{1}^{1/2}\left\vert
\varphi\right\rangle \right\vert \nonumber\\
&  =2\left\Vert L_{0}^{-1/2}L_{1}^{-1/2}\right\Vert ^{-1}.
\label{hatFmax-explicit}%
\end{align}
If $L_{0}$, $L_{1}>0$, by (\ref{MFmax}) and (\ref{hatF}), $s\leq\hat{F}_{\max
}\left(  L_{0},L_{1}\right)  $ holds if and only if there is $L>0$ such that%

\begin{align*}
&  \,2L_{0}\geq sL,\,2L_{1}\geq sL^{-1}\\
&  \Leftrightarrow2L_{0}\geq sL,\,L\geq\frac{s}{2}L_{1}^{-1}\Leftrightarrow
L_{0}\geq\frac{s^{2}}{4}L_{1}^{-1}\\
&  \Leftrightarrow\sqrt{L_{1}}L_{0}\sqrt{L_{1}}\geq\frac{s^{2}}{4}I\mathbf{,}%
\end{align*}
which leads to the asserted result. In particular, (\ref{hatFmax-explicit})
and Proposition \ref{prop:polar-normalize} leads to (\ref{Fcl-explicit}).

By (\ref{hatFmax-explicit}) and (\ref{MFmin-MFmax}), we have%

\begin{align*}
\hat{F}_{\min}\left(  L_{0},L_{1}\right)   &  =\inf\left\{  s\,;\,\left(
L_{0},L_{1}\right)  \notin s\,\left(  \mathcal{M}_{F_{\min}}\right)
,\,s>0\right\} \\
&  =\inf\left\{  s\,;\,\exists A\in\mathcal{L}_{sa,k},\,\left(  L_{0},\left(
I_{k}+\sqrt{-1}A\right)  ^{-1}L_{1}\left(  I_{k}-\sqrt{-1}A\right)
^{-1}\right)  \notin s\,\left(  \mathcal{M}_{F_{\min}}\right)  ,\,s>0\right\}
\\
&  =\sup_{A\in\mathcal{L}_{sa,k}}\hat{F}_{\max}\left(  L_{0},\left(
I_{k}+\sqrt{-1}A\right)  ^{-1}L_{1}\left(  I_{k}-\sqrt{-1}A\right)
^{-1}\right)  .
\end{align*}

\begin{theorem}
\label{th:hatF<FQ}Suppose that $\hat{F}^{Q}\in\mathcal{F}_{0}$ is monotone
increasing by any unital CP map, and polarly normalized. Then,
\[
\hat{F}_{\max}\left(  L_{0},L_{1}\right)  \leq\hat{F}^{Q}\left(  L_{0}%
,L_{1}\right)  \leq\hat{F}_{\min}\left(  L_{0},L_{1}\right)  .
\]

\end{theorem}

\begin{proof}
For each $\hat{F}^{Q}\in\mathcal{F}_{0}$ which is monotone increasing by any
unital CP map, and polarly normalized, there is a CPTP monotone and normalized
member $F^{Q}$ \ of $\mathcal{F}_{0}$ with (\ref{holder-1}), due to by
Corollary\thinspace\thinspace\ref{cor:1-to-1}. Since $F^{Q}$ is sandwiched by
$F_{\min}$ and $F_{\max}$ by Theorem\thinspace\ref{th:Fmin<FQ<F},
(\ref{holder-1}) implies
\begin{align*}
\hat{F}^{Q}\left(  L_{0},L_{1}\right)   &  \geq\inf\left\{  \frac{1}{F_{\max
}^{Q}\left(  X,Y\right)  }\left(  \mathrm{tr}\,L_{0}X+\mathrm{tr}%
\,L_{1}Y\right)  ;\left(  X,Y\right)  \in\mathcal{P}^{\times2},F^{Q}\left(
X,Y\right)  >0\right\}  ,\\
&  =\hat{F}_{\max}\left(  L_{0},L_{1}\right)  ,\\
\hat{F}^{Q}\left(  L_{0},L_{1}\right)   &  \leq\inf\left\{  \frac{1}{F_{\min
}^{Q}\left(  X,Y\right)  }\left(  \mathrm{tr}\,L_{0}X+\mathrm{tr}%
\,L_{1}Y\right)  ;\left(  X,Y\right)  \in\mathcal{P}^{\times2},F^{Q}\left(
X,Y\right)  >0\right\} \\
&  =\hat{F}_{\min}\left(  L_{0},L_{1}\right)  .
\end{align*}
Thus we have the assertion.
\end{proof}

Define
\begin{align}
\hat{F}_{\max}^{\prime}\left(  L_{0},L_{1}\right)   &  :=\sup\left\{  \hat
{F}^{C}\left(  l_{0},l_{1}\right)  \,;\,L_{\theta}=\Phi^{\ast}\left(
l_{\theta}\right)  ,\,\theta=0,1,\,\Phi^{\ast}\,\text{:\thinspace CP unital
map from }\mathcal{C}_{k}\text{\ to }\mathcal{L}_{k}\right\} \nonumber\\
&  =\sup\left\{  \hat{F}^{C}\left(  l_{0},l_{1}\right)  \,;\,L_{\theta}%
=\Phi_{M}^{\ast}\left(  l_{\theta}\right)  ,\,\theta
=0,1,\,M\,\text{:\thinspace POVM }\right\}  ,\label{hatFmax'}\\
\hat{F}_{\min}^{\prime}\left(  L_{0},L_{1}\right)   &  :=\inf\left\{  \hat
{F}^{C}\left(  l_{0},l_{1}\right)  \,;\,l_{\theta}=\Psi^{\ast}\left(
L_{\theta}\right)  ,\,\theta=0,1\,,\,\Psi^{\ast}\,\text{:\thinspace CP unital
map from }\mathcal{L}_{k}\text{\ to }\mathcal{C}_{k}\text{ }\right\}
\nonumber\\
&  =\inf\left\{  \hat{F}^{C}\left(  l_{0},l_{1}\right)  \,;\,l_{\theta}%
=\Psi_{\vec{\rho}}^{\ast}\left(  L_{\theta}\right)  ,\,\theta=0,1\,,\,\vec
{\rho}\,\text{:\thinspace\ array of density operators }\right\} \nonumber\\
&  =\min\,\left\{  2\sqrt{\mathrm{tr}\,\rho L_{0}\,\mathrm{tr}\,\rho L_{1}%
}\,;\,\rho\geq0,\,\mathrm{tr}\,\rho=1\right\}  . \label{hatFmin'}%
\end{align}

\begin{lemma}
\label{lem:hatFmin'-F0}$\hat{F}_{\min}^{\prime}$ is a member of $\mathcal{F}%
_{0}$ , monotone increasing by any unital CP map, and is polarly normalized.
Also, $\hat{F}_{\min}^{\prime}$ is continuous. Namely, if $\left(
L_{0,\infty},L_{1,\infty}\right)  $ is at the (relative) boundary of
$\mathcal{P}^{\times2}$ and $\lim_{i\rightarrow\infty}L_{\theta,i}%
=L_{\theta,\infty}$, $\theta=0$, $1$, \
\begin{equation}
\lim_{i\rightarrow\infty}\hat{F}_{\min}^{\prime}\left(  L_{0,i},L_{1,i}%
\right)  =\hat{F}_{\min}^{\prime}\left(  L_{0,\infty},L_{1,\infty}\right)  =0.
\label{hatFmin'=0}%
\end{equation}

\end{lemma}

\begin{proof}
Observe $\hat{F}_{\min}^{\prime}$ is infimum of the map
\[
\left(  L_{0},L_{1}\right)  \rightarrow\hat{F}^{C}\left(  \Psi_{\vec{\rho}%
}^{\ast}\left(  L_{0}\right)  ,\Psi_{\vec{\rho}}^{\ast}\left(  L_{1}\right)
\right)  .
\]
Since this map is a member of $\mathcal{F}_{0}$ for each $\vec{\rho}$, so is
$\hat{F}_{\min}^{\prime}$ by Lemma\thinspace\ref{lem:lim-F0}. Also, if
$\Lambda^{\ast}$ is a CP unital map,
\begin{align*}
\hat{F}_{\min}^{\prime}\left(  \Lambda^{\ast}\left(  L_{0}\right)
,\Lambda^{\ast}\left(  L_{1}\right)  \right)   &  =\min\,\left\{
2\sqrt{\mathrm{tr}\,\rho\Lambda^{\ast}\left(  L_{0}\right)  \,\mathrm{tr}%
\,\rho\Lambda^{\ast}\left(  L_{1}\right)  }\,;\,\rho\geq0,\,\mathrm{tr}%
\,\rho=1\right\} \\
&  =\min\,\left\{  2\sqrt{\mathrm{tr}\,\Lambda\left(  \rho\right)
L_{0}\,\mathrm{tr}\,\Lambda\left(  \rho\right)  L_{1}}\,;\,\rho\geq
0,\,\mathrm{tr}\,\rho=1\right\} \\
&  \geq\min\,\left\{  2\sqrt{\mathrm{tr}\,\rho L_{0}\,\mathrm{tr}\,\rho L_{1}%
}\,;\,\rho\geq0,\,\mathrm{tr}\,\rho=1\right\} \\
&  =\hat{F}_{\min}^{\prime}\left(  L_{0},L_{1}\right)  .
\end{align*}
Thus $\hat{F}_{\min}^{\prime}$ is monotone increasing by any unital CP map.

Also, for any state $\rho$, \
\[
\sqrt{\mathrm{tr}\,\rho L_{0,c}\,\mathrm{tr}\,\rho L_{1,c}}=\sqrt
{\mathrm{tr}\,\Gamma_{1}\left(  \rho\right)  L_{0,c}\,\mathrm{tr\,}\Gamma
_{1}\left(  \rho\right)  L_{1,c}},
\]
where $L_{\theta,c}$, $\theta=0,1$ are as of (\ref{Lc}) and $\Gamma_{1}$ is
the pinching with respect to the basis $\left\{  \left\vert i\right\rangle
;i=1,\cdots,k\right\}  $. Thus, in the minimum of (\ref{hatFmin'}), $\rho$ can
be restricted to those which commute with $L_{0}$ and $L_{1}$. Therefore,
\begin{align*}
\hat{F}_{\max}^{\prime}\left(  L_{0,c},L_{1,c}\right)   &  =\min\left\{
2\sqrt{\sum_{i}p_{i}\,l_{0,i}\sum_{i}p_{i}\,l_{1,i}},\,\sum_{i}p_{i}%
=1,p_{i}\geq0\text{ }\right\} \\
&  =\min\left\{  2\sum_{i}p_{i}\sqrt{\,l_{0,i}l_{1,i}},\,\sum_{i}p_{i}%
=1,p_{i}\geq0\text{ }\right\} \\
&  =\hat{F}^{C}\left(  l_{0},l_{1}\right)  ,
\end{align*}
where the second equality holds because of the concavity of $\left(
a,b\right)  \rightarrow\sqrt{ab}$. Thus, $\hat{F}_{\min}^{\prime}$ is also
polarly normalized.

By Lemma\thinspace\ref{lem:conv-cont}, $\hat{F}_{\min}^{\prime}$ is continuous
on $\mathrm{ri}\,\mathcal{P}^{\times2}$. Also, for any sequence $\left\{
\left(  L_{0,i},L_{1,i}\right)  \right\}  $\thinspace which converges to
$\left(  L_{0,\infty},L_{1,\infty}\right)  $,
\begin{align*}
\varlimsup_{i\rightarrow\infty}\hat{F}_{\min}^{\prime}\left(  L_{0,i}%
,L_{1,i}\right)   &  \leq\lim_{i\rightarrow\infty}2\sqrt{\left\langle
\psi\right\vert L_{0,i}\,\left\vert \psi\right\rangle \left\langle
\psi\right\vert L_{1,i}\left\vert \psi\right\rangle }=0,\\
\hat{F}_{\min}^{\prime}\left(  L_{0,\infty},L_{1,\infty}\right)   &
\leq2\sqrt{\left\langle \psi\right\vert L_{0,\infty}\,\left\vert
\psi\right\rangle \left\langle \psi\right\vert L_{1,\infty}\left\vert
\psi\right\rangle }=0,
\end{align*}
where $\left\vert \psi\right\rangle $ satisfies $L_{0,\infty}\left\vert
\psi\right\rangle =0$ or $L_{1,\infty}\left\vert \psi\right\rangle =0$. Since
$\hat{F}_{\min}^{\prime}$ is non-negative on $\mathcal{P}^{\times2}$, we have
\[
\lim_{i\rightarrow\infty}\hat{F}_{\min}^{\prime}\left(  L_{0,i},L_{1,i}%
\right)  =0=\hat{F}_{\min}^{\prime}\left(  L_{0,\infty},L_{1,\infty}\right)
.
\]
Therefore, $\hat{F}_{\min}^{\prime}$ is continuous in $\mathcal{P}^{\times2}$.
\end{proof}

\begin{lemma}
\label{lem:hatFmax'-F0}$\hat{F}_{\max}^{\prime}$ is a member of $\mathcal{F}%
_{0}$ , monotone increasing by any unital CP map, and is polarly normalized.
Also, $\hat{F}_{\min}^{\prime}$ is continuous.
\end{lemma}

\begin{proof}
$\hat{F}_{\max}^{\prime}$ is obviously positively homogeneous by definition.
Observe that $L_{\theta}=\sum_{i=1}^{n}l_{\theta,i}M_{i}$ and $L_{\theta
}^{\prime}=\sum_{i=1}^{n^{\prime}}l_{\theta,i}^{\prime}M_{i}^{\prime}$ imply
\[
L_{\theta}+L_{\theta}^{\prime}=\sum_{j=1}^{n+n^{\prime}}l_{\theta,j}^{\lambda
}M_{j}^{\lambda}=\Phi_{M^{\lambda}}^{\ast}\left(  l_{\theta}^{\lambda}\right)
,
\]
where
\begin{align*}
l_{\theta,j}^{\lambda}  &  =\left\{
\begin{array}
[c]{cc}%
\frac{1}{\lambda}l_{\theta,j}, & j=1,\cdots,n\\
\frac{1}{1-\lambda}l_{\theta,j-n}^{\prime}, & j=n+1,\cdots,n+n^{\prime}%
\end{array}
\right.  ,\,\\
M_{j}^{\lambda}  &  =\left\{
\begin{array}
[c]{cc}%
\lambda M_{j}, & j=1,\cdots,n\\
\left(  1-\lambda\right)  M_{j}^{\prime}, & j=n+1,\cdots,n+n^{\prime}%
\end{array}
\right.  .
\end{align*}
Therefore,
\begin{align*}
&  \hat{F}_{\max}^{\prime}\left(  L_{0}+L_{0}^{\prime},L_{1}+L_{1}^{\prime
}\right) \\
&  \geq\sup_{l_{\theta},l_{\theta}^{\prime},M,M^{\prime}}\max_{\lambda
\in\lbrack0,1]}\hat{F}^{C}\left(  l_{0}^{\lambda},l_{1}^{\lambda}\right) \\
&  =\sup_{l_{\theta},l_{\theta}^{\prime},M,M^{\prime}}\max_{\lambda\in
\lbrack0,1]}\min\left\{  \frac{1}{\lambda}\hat{F}^{C}\left(  l_{0}%
,l_{1}\right)  ,\frac{1}{1-\lambda}\hat{F}^{C}\left(  l_{0}^{\prime}%
,l_{1}^{\prime}\right)  \right\} \\
&  =\sup_{l_{\theta},l_{\theta}^{\prime},M,M^{\prime}}\left\{  \hat{F}%
^{C}\left(  l_{0},l_{1}\right)  +\hat{F}^{C}\left(  l_{0}^{\prime}%
,l_{1}^{\prime}\right)  \right\} \\
&  =\hat{F}_{\max}^{\prime}\left(  L_{0},L_{1}\right)  +\hat{F}_{\max}%
^{\prime}\left(  L_{0}^{\prime},L_{1}^{\prime}\right)  .
\end{align*}
Therefore, combined with positive homogeneity, we have concavity of $\hat
{F}_{\max}^{\prime}$.

Also,
\begin{equation}
\hat{F}_{\max}^{\prime}\leq\hat{F}_{\min}^{\prime} \label{hatFmax'<hatFmin'}%
\end{equation}
as is shown below. For any CP unital map $\Phi^{\ast}$\thinspace from
$\mathcal{C}_{k}$\ to $\mathcal{L}_{k}$ and $\Psi^{\ast}$ from $\mathcal{L}%
_{k}$ to $\mathcal{C}_{k}$, \ $\Psi^{\ast}\circ\Phi^{\ast}$ is transpose of a
stochastic map. So,
\[
\hat{F}^{C}\left(  l_{0},l_{1}\right)  \leq\hat{F}^{C}\left(  \Psi^{\ast}%
\circ\Phi^{\ast}\left(  l_{0}\right)  ,\Psi^{\ast}\circ\Phi^{\ast}\left(
l_{1}\right)  \right)  .
\]
Hence, if $L_{\theta}=\Phi^{\ast}\left(  l_{\theta}\right)  $, $\theta=0,1$,%
\[
\hat{F}^{C}\left(  l_{0},l_{1}\right)  \leq\hat{F}^{C}\left(  \Psi^{\ast
}\left(  L_{0}\right)  ,\Psi^{\ast}\left(  L_{1}\right)  \right)  .
\]
Taking supremum of the LHS and infimum of the RHS, we obtain $\hat{F}_{\max
}^{\prime}\leq\hat{F}_{\min}^{\prime}$. Therefore, $\hat{F}_{\max}^{\prime}$
nowhere takes the value $\infty$. Also, that $\hat{F}_{\max}^{\prime}$ does
not take the value $-\infty$ in $\mathcal{P}^{\times2}$ is obvious by definition.

By Lemma\thinspace\ref{lem:conv-cont}, $\hat{F}_{\max}^{\prime}$ is continuous
on $\mathrm{ri}\,\mathcal{P}^{\times2}$. Also, for any sequence $\left\{
\left(  L_{0,i},L_{1,i}\right)  \right\}  $\thinspace which converges to
$\left(  L_{0,\infty},L_{1,\infty}\right)  $, by (\ref{hatFmax'<hatFmin'}) and
(\ref{hatFmin'=0})%
\begin{align*}
\varlimsup_{i\rightarrow\infty}\hat{F}_{\max}^{\prime}\left(  L_{0,i}%
,L_{1,i}\right)   &  \leq\varlimsup_{i\rightarrow\infty}\hat{F}_{\min}%
^{\prime}\left(  L_{0,i},L_{1,i}\right)  =0,\\
\hat{F}_{\max}^{\prime}\left(  L_{0,\infty},L_{1,\infty}\right)   &  \leq
\hat{F}_{\min}^{\prime}\left(  L_{0,\infty},L_{1,\infty}\right)  =0.
\end{align*}
Since $\hat{F}_{\max}^{\prime}$ is non-negative on $\mathcal{P}^{\times2}$, we
have
\[
\lim_{i\rightarrow\infty}\hat{F}_{\max}^{\prime}\left(  L_{0,i},L_{1,i}%
\right)  =0=\hat{F}_{\max}^{\prime}\left(  L_{0,\infty},L_{1,\infty}\right)
.
\]
Therefore, $\hat{F}_{\max}^{\prime}$ is continuous, and thus it is closed.
After all, $\hat{F}_{\max}^{\prime}$ is a member of $\mathcal{F}_{0}$.

$\hat{F}_{\max}^{\prime}$ is monotone increasing by any unital CP map
$\Lambda^{\ast}$, proved as follows. If a POVM $M=\left\{  M_{i}\right\}  $
satisfies $L_{\theta}=\sum_{i}l_{\theta,i}M_{i}$, \ the set $\left\{
\Lambda^{\ast}\left(  M_{i}\right)  \right\}  $ of operators is POVM and
satisfies $\Lambda^{\ast}\left(  L_{\theta}\right)  =\sum_{i}l_{\theta
,i}\Lambda^{\ast}\left(  M_{i}\right)  $. Therefore,
\begin{align*}
&  \hat{F}_{\max}^{\prime}\left(  \Lambda^{\ast}\left(  L_{0}\right)
,\Lambda^{\ast}\left(  L_{1}\right)  \right) \\
&  =\sup\left\{  \hat{F}^{C}\left(  l_{0},l_{1}\right)  ;\,\Lambda^{\ast
}\left(  L_{\theta}\right)  =\sum_{i}l_{\theta,i}M_{i},\,\theta
=0,1,\,M\text{:\thinspace POVM }\right\} \\
&  \geq\sup\left\{  \hat{F}^{C}\left(  l_{0},l_{1}\right)  \,;\,\Lambda^{\ast
}\left(  L_{\theta}\right)  =\sum_{i}l_{\theta,i}\Lambda^{\ast}\left(
M_{i}\right)  ,\,\theta=0,1,\,M\text{:\thinspace POVM }\right\} \\
&  \geq\sup\left\{  \hat{F}^{C}\left(  l_{0},l_{1}\right)  \,;\,L_{\theta
}=\sum_{i}l_{\theta,i}M_{i},\,\,\theta=0,1,\,M\text{:\thinspace POVM }\right\}
\\
&  =\hat{F}_{\max}^{\prime}\left(  L_{0},L_{1}\right)  .
\end{align*}

Suppose that the triple $l_{0}$, $l_{1}$, $M=\left\{  M_{i}\right\}  $
satisfies the constrain given in the RHS of (\ref{hatFmax'}) with $\left(
L_{0},L_{1}\right)  =\left(  L_{0,c},L_{1,c}\right)  $, where $\left(
L_{0,c},L_{1,c}\right)  $ is as of (\ref{Lc}). Then \ the triple $l_{0}$,
$l_{1}$, $\left\{  \Gamma_{1}\left(  M_{i}\right)  \right\}  $ also satisfies
the constrain. Therefore, without changing the maximum, we may restrict the
range of POVM to the ones which are diagonalized in the basis $\left\{
\left\vert i\right\rangle \right\}  _{i=1}^{k}$. Such a POVM corresponds to a
transpose of a stochastic matrix. Therefore,
\[
\hat{F}_{\max}^{\prime}\left(  L_{0,c},L_{1,c}\right)  =\sup\left\{  \hat
{F}^{C}\left(  \lambda_{0},\lambda_{1}\right)  ;l_{\theta}=T\,^{t}%
\lambda_{\theta},\theta=0,1,\,T\text{: column stochastic matrix }\right\}
\]
Since $\hat{F}^{C}$ is monotone increasing by the application of a transpose
of a stochastic matrix, $\hat{F}^{C}\left(  \lambda_{0},\lambda_{1}\right)  $
cannot exceed $\hat{F}^{C}\left(  l_{0},l_{1}\right)  $ if $l_{\theta}%
=T\,^{t}\lambda_{\theta}$ ($\theta=0,1$). Therefore, $\hat{F}_{\max}^{\prime
}\left(  L_{0,c},L_{1,c}\right)  =\hat{F}^{C}\left(  l_{0},l_{1}\right)  $,
and $\hat{F}_{\max}^{\prime}$ is polarly normalized.\ 
\end{proof}

The following theorem gives `operational' meaning of $\hat{F}_{\max}$ and
$\hat{F}_{\min}^{\prime}\,$.

\begin{theorem}
\label{th:hatFmax=max}%
\begin{equation}
\hat{F}_{\max}=\hat{F}_{\max}^{\prime},\,\hat{F}_{\min}=\hat{F}_{\min}%
^{\prime}. \label{hatFmax=max}%
\end{equation}

\end{theorem}

\begin{proof}
First, we show
\begin{equation}
\hat{F}_{\max}^{\prime}\left(  L_{0},L_{1}\right)  \leq\hat{F}_{\max}\left(
L_{0},L_{1}\right)  \leq\hat{F}_{\min}\left(  L_{0},L_{1}\right)  \leq\hat
{F}_{\min}^{\prime}\left(  L_{0},L_{1}\right)  . \label{hatF'<hatF<hatF<hatF'}%
\end{equation}
Suppose $\hat{F}^{Q}$ is polarly normalized and monotone increasing by any
unital CP maps. Then,
\begin{align*}
\hat{F}^{Q}\left(  L_{0},L_{1}\right)   &  \geq\sup_{l_{\theta},M}\left\{
\,\,\hat{F}^{Q}\left(  L_{0,c},L_{1,c}\right)  ;\,L_{\theta}=\Phi_{M}^{\ast
}\left(  L_{\theta,c}\right)  \right\} \\
&  =\sup_{l_{\theta},M}\left\{  \,\,\hat{F}^{C}\left(  l_{0},l_{1}\right)
;\,L_{\theta}=\Phi_{M}^{\ast}\left(  l_{\theta}\right)  \right\}  =\hat
{F}_{\max}^{\prime}\left(  L_{0,}L_{1}\right)  ,\\
\hat{F}^{Q}\left(  L_{0},L_{1}\right)   &  \leq\inf_{l_{\theta},\vec{\rho}%
}\left\{  \hat{F}^{Q}\left(  L_{0,c},L_{1,c}\right)  ;L_{\theta,c}=\Psi
_{\vec{\rho}}^{\ast}\left(  L_{\theta}\right)  \right\} \\
&  =\inf_{l_{\theta},\vec{\rho}}\left\{  \hat{F}^{C}\left(  l_{0}%
,l_{1}\right)  ;l_{\theta}=\Psi_{\vec{\rho}}^{\ast}\left(  L_{\theta}\right)
\right\}  =\hat{F}_{\min}^{\prime}\left(  L_{0},L_{1}\right)  .
\end{align*}
Since $\hat{F}_{\max}$ and $\hat{F}_{\min}$ are examples of such $\hat{F}^{Q}%
$, we have inequalities (\ref{hatF'<hatF<hatF<hatF'}).

By Lemmas\thinspace\ref{lem:hatFmax'-F0}-\ref{lem:hatFmin'-F0}, $\hat{F}%
_{\max}$ and $\hat{F}_{\min}$ are members of $\mathcal{F}_{0}$ which is
polarly normalized and monotone increasing by any unital CP map. Therefore,
\ by Theorem\thinspace\ref{th:hatF<FQ}, we have $\hat{F}_{\max}^{\prime}%
\leq\hat{F}_{\max}$ and $\hat{F}_{\min}\leq\hat{F}_{\min}^{\prime}$, which,
combined with (\ref{hatF'<hatF<hatF<hatF'}), lead to the assertion.
\end{proof}

\subsection{$\hat{F}_{1/2}$}

Below, we give an expression of $\hat{F}_{1/2}$. By (\ref{L=SY}), the optimal
$\left(  L_{0,\ast},L_{1,\ast}\right)  $ is given by the simultaneous linear equations%

\begin{align*}
\sqrt{Y}  &  =\sqrt{X}L_{0,\ast}+L_{0,\ast}\sqrt{X},\\
\sqrt{X}  &  =\sqrt{Y}L_{1,\ast}+L_{1,\ast}\sqrt{Y}.
\end{align*}
Therefore,
\begin{equation}
\sqrt{X}=\mathbf{S}_{L_{0,\ast}}\left(  \sqrt{Y}\right)  ,\,\sqrt
{Y}=\mathbf{S}_{L_{1,\ast}}\left(  \sqrt{X}\right)  . \label{X=SL}%
\end{equation}
So
\[
\partial\mathcal{M}_{F_{1/2}}=\left\{  \left(  L_{0,\ast},L_{1,\ast}\right)
;\exists X,Y\geq0\text{ with (\ref{X=SL})}\right\}  .
\]

Suppose there is $\sqrt{X}$ such that
\begin{equation}
\sqrt{X}=\mathbf{S}_{L_{0,\ast}}\circ\mathbf{S}_{L_{1,\ast}}\left(  \sqrt
{X}\right)  . \label{X=SSX}%
\end{equation}
Observe $\mathbf{S}_{L_{1,\ast}}$ is positive. Then defining $Y\geq0$ by
$\sqrt{Y}=\mathbf{S}_{L_{1,\ast}}\left(  \sqrt{X}\right)  $, $X,Y\geq0$
satisfies (\ref{X=SL}). \ Therefore,
\[
\partial\mathcal{M}_{F_{1/2}}=\left\{  \left(  L_{0,\ast},L_{1,\ast}\right)
;\exists X\geq0\text{ with (\ref{X=SSX})}\right\}  .
\]

For any $L_{0}>0$ and $L_{1}>0$, the map $\mathbf{S}_{L_{0}}\circ
\mathbf{S}_{L_{1}}$ is strictly positive. Also, a map
\begin{equation}
A\rightarrow\frac{1}{\mathrm{tr}\,\mathbf{S}_{L_{0}}\circ\mathbf{S}_{L_{1}%
}\left(  A\right)  }\mathbf{S}_{L_{0}}\circ\mathbf{S}_{L_{1}}\left(  A\right)
\label{SSA}%
\end{equation}
defined on the compact convex set
\[
\mathcal{P}_{k}^{\times2}\cap\left\{  A;\mathrm{tr}\,A=1\right\}  ,
\]
is continuous. Thus, by Tychonoff's fixed point theorem, there is $A_{\ast
}\geq0$ fixed by the map (\ref{SSA}), or equivalently,
\begin{equation}
\mathbf{S}_{L_{0}}\circ\mathbf{S}_{L_{1}}\left(  A_{\ast}\right)
=\alpha_{\ast}A_{\ast},\,\alpha_{\ast}>0, \label{eigen-SS}%
\end{equation}
or
\[
\mathbf{S}_{\sqrt{\alpha}L_{0}}\circ\mathbf{S}_{\sqrt{\alpha}L_{1}}\left(
A_{\ast}\right)  =A_{\ast}.
\]
Therefore, $\left(  \hat{F}\left(  L_{0},L_{1}\right)  \right)  ^{-2}$ is an
eigenvalue of $\mathbf{S}_{L_{0}}\circ\mathbf{S}_{L_{1}}$ corresponding to the
eigenvector $A_{\ast}\geq0$.

But there can be two or \ more eigenvalues of $\mathbf{S}_{L_{0}}%
\circ\mathbf{S}_{L_{1}}$ which corresponding eigenvectors are positive. Also,
in this way one has to compute eigenvectors in addition to eigenvalues of
$\mathbf{S}_{L_{0}}\circ\mathbf{S}_{L_{1}}$. So we further investigate the
nature of $\mathbf{S}_{L_{0}}\circ\mathbf{S}_{L_{1}}$.

\begin{proposition}
\label{prop:diagonalize}Suppose $L_{0}>0$ and $L_{1}>0$. Then, $\mathbf{S}%
_{L_{0}}\circ\mathbf{S}_{L_{1}}$ is diagonalizable and all the eigenvalues are positive.
\end{proposition}

\begin{proof}
First, $\mathbf{S}_{L_{0}}$ and $\mathbf{S}_{L_{1}}$ is self-adjoint, and all
the eigenvalues of them are positive. In fact, let $\left\vert \varphi
_{j}\right\rangle $ be the eigenvector of $L_{1}$ with corresponding
eigenvalue $\lambda_{j}$ ($>0$, by the assumption\thinspace$L_{1}>0$). Then,
\[
\mathbf{S}_{L_{1}}\left(  \left\vert \varphi_{i}\right\rangle \left\langle
\varphi_{j}\right\vert \right)  =\frac{1}{\lambda_{i}+\lambda_{j}}\left\vert
\varphi_{i}\right\rangle \left\langle \varphi_{j}\right\vert .
\]
Since $\left\{  \left\vert \varphi_{i}\right\rangle \left\langle \varphi
_{j}\right\vert \right\}  _{i,j}$ forms a complete basis of $\mathcal{L}_{k}$,
they are the only eigenvectors. Thus, all the eigenvalues of $\mathbf{S}%
_{L_{1}}$ are positive. Second, $\mathbf{S}_{L_{0}}\circ\mathbf{S}_{L_{1}}$
has the same Jordan standard form as $\mathbf{S}_{L_{1}}^{1/2}\circ
\mathbf{S}_{L_{0}}\circ\mathbf{S}_{L_{1}}\circ\mathbf{S}_{L_{1}}%
^{-1/2}=\mathbf{S}_{L_{1}}^{1/2}\circ\mathbf{S}_{L_{0}}\circ\mathbf{S}_{L_{1}%
}^{1/2}$ . Thus, $\mathbf{S}_{L_{0}}\circ\mathbf{S}_{L_{1}}$ is
diagonalizable, and all of its eigenvalues are positive.\ 
\end{proof}

\begin{proposition}
\label{prop:++-eigenvec}Suppose $L_{0}>0$ and $L_{1}>0$. Let $A$ and
$A^{\prime}$ be an eigenvector of $\mathbf{S}_{L_{0}}\circ\mathbf{S}_{L_{1}}$,
with corresponding eigenvalue $\alpha$ and $a^{\prime}$, respectively.
Further, suppose $A>0$. Then $\alpha\geq\alpha^{\prime}$. Especially, if
$A^{\prime}$ is also strictly positive, $\alpha=\alpha^{\prime}$.
\end{proposition}

\begin{proof}
Without loss of generality, we suppose $A^{\prime}$ is not a member of
$-\mathcal{P}$. (Otherwise, we name $-A^{\prime}$ as $A^{\prime}$.) Since
$A>0$, there is a positive $\lambda>0$ such that $A-\lambda A^{\prime}$ is on
$\partial\mathcal{P}_{k}$. Suppose $\alpha^{\prime}>0$. Then since
\[
\mathbf{S}_{L_{0}}\circ\mathbf{S}_{L_{1}}\left(  A-\lambda A^{\prime}\right)
=\alpha\left(  A-\frac{\alpha^{\prime}}{\alpha}\lambda A^{\prime}\right)
\]
is positive and $\alpha>0$ by Proposition\thinspace\ref{prop:diagonalize},
$\alpha^{\prime}$ cannot exceed $\alpha$. The second statement is proved by
interchanging $A$ and $A^{\prime}$.
\end{proof}

\begin{proposition}
\label{prop:+0-eigenvec}Let $A$ be an eigenvector of $\mathbf{S}_{L_{0}}%
\circ\mathbf{S}_{L_{1}}$, $L_{0}>0$, $L_{1}>0$. If $A$ is positive but may not
be strictly positive, then the subspace $\mathrm{supp}\,A$ \ is invariant by
$L_{0}$ and $L_{1}$.
\end{proposition}

\begin{proof}
Since $\mathbf{S}_{L_{1}}$ and $\mathbf{S}_{L_{0}}$ are completely positive
map,
\begin{align*}
\mathrm{supp}\,A  &  \subset\mathrm{supp}\,\mathbf{S}_{L_{1}}\left(  A\right)
\subset\mathrm{supp}\,\mathbf{S}_{L_{0}}\circ\mathbf{S}_{L_{1}}\left(
A\right) \\
&  =\mathrm{supp}\,A.
\end{align*}
Therefore,
\[
\mathrm{supp}\,\mathbf{S}_{L_{1}}\left(  A\right)  =\mathrm{supp}\,A.
\]

Recall%
\[
\mathbf{S}_{L_{1}}\left(  A\right)  =\int_{0}^{\infty}e^{-tL_{1}}Ae^{-tL_{1}%
}\,\mathrm{d}t.
\]
So, as is proved in the following,
\[
\mathrm{supp}\,e^{-tL_{1}}Ae^{-tL_{1}}\subset\mathrm{supp}\,A,\,\forall
t\geq0\text{. }%
\]
Suppose otherwise, or there is $t_{0}\geq0$ such that
\[
\mathrm{supp}\,e^{-t_{0}L_{1}}Ae^{-t_{0}L_{1}}\not \subset \mathrm{supp}\,A
\]
Then, by continuity of $t\rightarrow e^{-tL_{1}}Ae^{-tL_{1}}$, $\varepsilon>0$
such that
\[
\mathrm{supp}\,e^{-tL_{1}}Ae^{-tL_{1}}\not \subset \mathrm{supp}\,A,\,\forall
t\in\left[  t_{0},t_{0}+\varepsilon\right]  .
\]
Thus
\[
\mathrm{supp}\,\mathbf{S}_{L_{1}}\left(  A\right)  \subset\mathrm{supp}%
\,\int_{t_{0}}^{t_{0}+\varepsilon}e^{-tL_{1}}Ae^{-tL_{1}}\,\mathrm{d}%
t\not \subset \mathrm{supp}\,A,
\]
which leads to contradiction.

Since
\[
\dim\mathrm{supp}\,e^{-tL_{1}}Ae^{-tL_{1}}=\dim\mathrm{supp}\,A,
\]
we should have
\[
\mathrm{supp}\,e^{-tL_{1}}Ae^{-tL_{1}}=\mathrm{supp}\,A,\,\forall
t\geq0\text{. }%
\]
Therefore,
\[
\mathrm{supp}\,A=\mathrm{supp}\,\frac{e^{-tL_{1}}Ae^{-tL_{1}}-A}{-t},\,\forall
t>0\text{, }%
\]
which means
\[
\mathrm{supp\,}L_{1}A+AL_{1}=\mathrm{supp}\,A.
\]

Since
\begin{align*}
L_{1}A+AL_{1}  &  =\left[
\begin{array}
[c]{cc}%
L_{1,11} & L_{1,12}\\
L_{1,21} & L_{1,22}%
\end{array}
\right]  \left[
\begin{array}
[c]{cc}%
A_{11} & 0\\
0 & 0
\end{array}
\right]  +\left[
\begin{array}
[c]{cc}%
A_{11} & 0\\
0 & 0
\end{array}
\right]  \left[
\begin{array}
[c]{cc}%
L_{1,11} & L_{1,12}\\
L_{1,21} & L_{1,22}%
\end{array}
\right] \\
&  =\left[
\begin{array}
[c]{cc}%
L_{1,11}A_{11}+A_{11}L_{1,11} & A_{11}L_{1,12}\\
L_{1,21}A_{11} & 0
\end{array}
\right]  ,
\end{align*}
we have $L_{1,21}A_{11}=0$. Since $A_{11}$ is strictly positive, $L_{1,21}=0$.
Therefore, $\mathrm{supp}\,A$ is invariant by $L_{1}$.

Replacing $A$ by $\mathbf{S}_{L_{1}}\left(  A\right)  $ and $L_{1}$ by $L_{0}$
in the above argument, we can conclude that $\mathrm{supp}\,\mathbf{S}_{L_{1}%
}\left(  A\right)  =\mathrm{supp}\,A$ is invariant also by $L_{0}$.
\end{proof}

Using these propositions, $\hat{F}_{1/2}\left(  L_{0},L_{1}\right)  $
($L_{0}>0$, $L_{1}>0$) is computed as follows.

\begin{theorem}%
\begin{equation}
\hat{F}_{1/2}\left(  L_{0},L_{1}\right)  =2\left\Vert \mathbf{S}_{L_{0}}%
^{1/2}\circ\mathbf{S}_{L_{1}}\circ\mathbf{S}_{L_{0}}^{1/2}\right\Vert ^{-1/2}.
\label{hatF1/2}%
\end{equation}

\end{theorem}

In practice, it is easier to compute the square root of the smallest
eigenvalue of $\mathbf{S}_{L_{1}}^{-1}\circ\mathbf{S}_{L_{0}}^{-1}$, which is
the linear map sending $X$ to $\left\{  L_{1},\left\{  L_{0},X\right\}
\right\}  $.

\begin{proof}
Decompose $L_{\theta}$ ($\theta=0,1$) into
\[
L_{\theta}:=L_{\theta}^{\left(  1\right)  }\oplus L_{\theta}^{\left(
2\right)  }\oplus L_{\theta}^{\left(  3\right)  }\oplus\cdots\left(
\theta=0,1\right)
\]
so that $L_{0}^{\left(  i\right)  }$ and $L_{0}^{\left(  i\right)  }$ are
acting on the same subspace $\mathcal{H}^{\left(  i\right)  }$, and do not
have any smaller common invariant subspace. Then by (\ref{direct-min}), the
problem reduces to the computation of each $\hat{F}_{1/2}\left(
L_{0}^{\left(  i\right)  },L_{1}^{\left(  i\right)  }\right)  $. Since
$L_{1}^{\left(  i\right)  }$ and $L_{0}^{\left(  i\right)  }$ has no smaller
common invariant subspace, Proposition\thinspace\ref{prop:+0-eigenvec} implies
the following; Let us view $\mathbf{S}_{L_{0}^{\left(  i\right)  }}$ as a
linear transform on $\mathcal{L}\left(  \mathcal{H}^{\left(  i\right)
}\right)  $. If the eigenvector $A^{\left(  i\right)  }\in\mathcal{L}\left(
\mathcal{H}^{\left(  i\right)  }\right)  $ of $\mathbf{S}_{L_{0}^{\left(
i\right)  }}\circ\mathbf{S}_{L_{1}^{\left(  i\right)  }}$ is a positive
operator, it is strictly positive. Therefore, by Proposition\thinspace
\ref{prop:++-eigenvec}, the largest eigenvalue $\alpha^{\left(  i\right)  }$
of $\mathbf{S}_{L_{0}^{\left(  i\right)  }}\circ\mathbf{S}_{L_{1}^{\left(
i\right)  }}$ is the only eigenvalue whose corresponding eigenvector positive
definite operator. Thus, we have%
\begin{align*}
\hat{F}_{1/2}\left(  L_{0}^{\left(  i\right)  },L_{1}^{\left(  i\right)
}\right)   &  =2\left(  \alpha^{\left(  i\right)  }\right)  ^{-1/2},\\
\hat{F}_{1/2}\left(  L_{0},L_{1}\right)   &  =2\min_{i}\left(  \alpha^{\left(
i\right)  }\right)  ^{-1/2}=2\left(  \max_{i}\alpha^{\left(  i\right)
}\right)  ^{-1/2}.
\end{align*}

Denote by $\left[  X^{\left(  i,j\right)  }\right]  $ the matrix whose
$\left(  i,j\right)  $ block is a linear map $X^{\left(  i,j\right)  }$ from
$\mathcal{H}^{\left(  j\right)  }$ to $\mathcal{H}^{\left(  i\right)  }$.
Then, by (\ref{S-int}),
\[
\mathbf{S}_{L_{0}}\circ\mathbf{S}_{L_{1}}\left(  X\right)  =\int_{0}^{\infty
}\int_{0}^{\infty}\left[  e^{-sL_{0}^{\left(  i\right)  }}e^{-tL_{1}^{\left(
i\right)  }}X^{\left(  i,j\right)  }e^{-tL_{1}^{\left(  j\right)  }}%
e^{-sL_{0}^{\left(  j\right)  }}\right]  \mathrm{d\,}t\,\mathrm{d\,}s.
\]
Therefore, \ the $\left(  i,i\right)  $ block of $\mathbf{S}_{L_{0}}%
\circ\mathbf{S}_{L_{1}}\left(  X\right)  $ is $\mathbf{S}_{L_{0}^{\left(
i\right)  }}\circ\mathbf{S}_{L_{1}^{\left(  i\right)  }}\left(  X^{\left(
i,i\right)  }\right)  $. Therefore, if $X$ is an eigenvector of $\mathbf{S}%
_{L_{0}}\circ\mathbf{S}_{L_{1}}\left(  X\right)  $ corresponding to the
eigenvalue $\alpha$, $X^{\left(  i,i\right)  }$ is an eigenvector of
$\mathbf{S}_{L_{0}^{\left(  i\right)  }}\circ\mathbf{S}_{L_{1}^{\left(
i\right)  }}$ corresponding to the eigenvalue $\alpha$. Therefore, an
eigenvalue of $\mathbf{S}_{L_{0}}\circ\mathbf{S}_{L_{1}}$ cannot exceed
$\max_{i}\alpha^{\left(  i\right)  }$. On the other hand,
\[
0\oplus\cdots0\oplus A^{\left(  i_{\ast}\right)  }\oplus0\cdots\oplus0
\]
is an eigenvector of $\mathbf{S}_{L_{0}}\circ\mathbf{S}_{L_{1}}$ corresponding
to the eigenvalue $\max_{i}\alpha^{\left(  i\right)  }$, where $i_{\ast
}:=\mathrm{argmax}_{i}\,\alpha^{\left(  i\right)  }$. Therefore,
\begin{align*}
\max_{i}\alpha^{\left(  i\right)  }  &  =\text{the largest eigenvalue of
}\mathbf{S}_{L_{0}}\circ\mathbf{S}_{L_{1}}\\
&  =\left\Vert \mathbf{S}_{L_{0}}^{1/2}\circ\mathbf{S}_{L_{1}}\circ
\mathbf{S}_{L_{0}}^{1/2}\right\Vert ,
\end{align*}
and we have the asserted result.
\end{proof}

\section{Extreme points and boundary}

The set of all extreme points $\mathrm{ext\,}\mathcal{M}_{F^{Q}}$ of
$\mathcal{M}_{F^{Q}}$ , by Lemma\thinspace\ref{lem:ext+recess}, satisfies
\begin{equation}
\mathcal{M}_{F^{Q}}=\mathrm{conv\,ext\,}\mathcal{M}_{F^{Q}}+\mathcal{P}%
^{\times2}. \label{MFQ=convext+P}%
\end{equation}
Thus, $\mathrm{ext\,}\mathcal{M}_{F^{Q}}$ is the key part in considering
minimization of $\mathrm{tr}\,L_{0}X+\mathrm{tr}\,L_{1}Y$. By (\ref{F=L+1/L}%
),
\[
\mathrm{ext\,}\mathcal{M}_{F_{\min}}=\left\{  \left(  L_{0},L_{1}\right)
;L_{0}>0,L_{1}=\left(  4L_{0}\right)  ^{-1}\right\}  .
\]

To see geometry of $\mathrm{ext\,}\mathcal{M}_{F^{Q}}$ and $\mathcal{M}%
_{F^{Q}}$, for each $L_{0}$, define
\[
\mathcal{M}_{F^{Q}}\left(  L_{0}\right)  :=\left\{  L_{1};\left(  L_{0}%
,L_{1}\right)  \in\mathcal{M}_{F^{Q}}\right\}  .
\]
By Lemma\thinspace\ref{lem:ext+recess},
\begin{equation}
\mathcal{M}_{F^{Q}}\left(  L_{0}\right)  =\mathrm{conv\,ext\,}\mathcal{M}%
_{F^{Q}}\left(  L_{0}\right)  +\mathcal{P}. \label{MFQ=conv+P-2}%
\end{equation}
$\mathcal{M}_{F^{Q}}$ is specified if we specify $\mathcal{M}_{F^{Q}}\left(
L_{0}\right)  $, because of the following.

\begin{proposition}
Suppose $F^{Q}\in\mathcal{F}_{0}$ is CPTP monotone and normalized. Then,
\[
\left\{  L_{0};\left(  L_{0},L_{1}\right)  \in\mathcal{M}_{F^{Q}}\right\}
=\left\{  L_{0};L_{0}>0\right\}
\]

\end{proposition}

\begin{proof}
By Theorem\thinspace\ref{th:Fmin<FQ<F}, $\mathcal{M}_{F_{\max}}\subset
\mathcal{M}_{F^{Q}}\subset\mathcal{M}_{F_{\min}}$. Therefore, (\ref{MFmax})
implies
\[
\left\{  L_{0};L_{0}>0\right\}  \subset\left\{  L_{0};\left(  L_{0}%
,L_{1}\right)  \in\mathcal{M}_{F^{Q}}\right\}  ,
\]
and (\ref{MFmin-2}) implies opposite inclusion. Thus we have the assertion.
\end{proof}

Due to this Proposition, we have \
\begin{align*}
\mathcal{M}_{F^{Q}}  &  =\left\{  \left(  L_{0},L_{1}\right)  ;L_{0}%
>0,\,L_{1}\in\mathcal{M}_{F^{Q}}\left(  L_{0}\right)  \right\} \\
\mathrm{ext\,}\mathcal{M}_{F^{Q}}  &  =\left\{  \left(  L_{0},L_{1}\right)
;L_{0}>0,\,L_{1}\in\,\mathrm{ext}\mathcal{M}_{F^{Q}}\left(  L_{0}\right)
\right\}  .
\end{align*}
Thus, \ $\mathcal{M}_{F^{Q}}\left(  L_{0}\right)  $ and $\mathrm{ext}%
\mathcal{M}_{F^{Q}}\left(  L_{0}\right)  $ determines $\mathcal{M}_{F^{Q}}$
and $\mathrm{ext\,}\mathcal{M}_{F^{Q}}$, respectively. By (\ref{F=L+1/L}),%

\begin{align*}
\mathcal{M}_{F_{\max}}\left(  L_{0}\right)   &  =\left\{  L_{1};L_{1}%
\geq\left(  4L_{0}\right)  ^{-1}\right\}  ,\\
\mathrm{ext\,}\mathcal{M}_{F_{\max}}\left(  L_{0}\right)   &  =\left\{
\left(  4L_{0}\right)  ^{-1}\right\}  ,
\end{align*}
and the latter is consists of only a single point.

As is shown below, if $L_{0}$ and $L_{1}$ are strictly positive and have no
common non-trivial invariant subspace,%
\begin{equation}
L_{1,\#}\in\mathrm{conv\,ext\,}\mathcal{M}_{F_{1/2}}\left(  L_{0}\right)  ,
\label{L/F-in-M}%
\end{equation}
where
\[
L_{1,\#}:=\frac{1}{\left\{  \hat{F}^{1/2}\left(  L_{0},L_{1}\right)  \right\}
^{2}}L_{1}.
\]
This means the dimension of $\mathrm{conv\,ext\,}\mathcal{M}_{F_{1/2}}\left(
L_{0}\right)  \cap\mathcal{L}_{k}$ is full, i.e., $k^{2}$.

First, by the definition of $\hat{F}^{1/2}\left(  L_{0},L_{1}\right)  $,
Proposition\thinspace\ref{prop:t-1/t} and Proposition\thinspace
\ref{prop:hatF-s-hom}, $\left(  L_{0},L_{1,\#}\right)  $ is a member of
$\partial\mathcal{M}_{F_{1/2}}$, where Suppose $\left(  M_{0},M_{1}\right)
\in\mathcal{M}_{F_{1/2}}$ that is not identical to $\left(  L_{0}%
,L_{1,\#}\right)  $ satisfies $M_{0}\leq L_{0}$, $\,M_{1}\leq L_{1,\#}$ .
Then, there is $\left(  X,Y\right)  \in\mathcal{P}^{\times2}$ such that
\begin{align*}
\min_{\left(  L_{0}^{\prime},,L_{1}^{\prime}\right)  \in\mathcal{M}_{F_{1/2}}%
}\left(  \mathrm{tr}\,X\,L_{0}^{\prime}+\mathrm{tr}\,YL_{1}^{\prime}\right)
&  =\mathrm{tr}\,X\,L_{0}+\mathrm{tr}\,YL_{1,\#}\\
&  \geq\mathrm{tr}\,X\,M_{0}+\mathrm{tr}\,YM_{1}\\
&  \geq\min_{\left(  L_{0}^{\prime},,L_{1}^{\prime}\right)  \in\mathcal{M}%
_{F_{1/2}}}\left(  \mathrm{tr}\,X\,L_{0}^{\prime}+\mathrm{tr}\,YL_{1}^{\prime
}\right)  .
\end{align*}
Therefore,
\[
\mathrm{tr}\,X\,L_{0}+\mathrm{tr}\,YL_{1,\#}=\mathrm{tr}\,X\,M_{0}%
+\mathrm{tr}\,YM_{1}.
\]
This can hold true only if $X$ and $Y$ has null eigenspace. By
Proposition\thinspace\ref{prop:+0-eigenvec}, in turn this means $L_{0}$, and
$L_{1}$ has a common non-trivial invariant subspace, contradicting with the
assumption. Therefore, there is no $\left(  M_{0},M_{1}\right)  $ with
$M_{0}\leq L_{0}$, $\,M_{1}\leq L_{1,\#}$. So if $\left(  L_{0},L_{1,\#}%
\right)  $ is not a member of $\mathrm{conv\,ext\,}\mathcal{M}_{F_{1/2}}$, it
contradicts with (\ref{MFQ=convext+P}). Therefore, we have (\ref{L/F-in-M}).

On the other hand, as is shown in detail in the next section,
\[
\dim\mathrm{conv\,ext\,}\mathcal{M}_{F_{\min}}\left(  L_{0}\right)
\cap\mathcal{L}_{2}=3.
\]
So while $\mathrm{conv\,ext\,}\mathcal{M}_{F_{\max}}\left(  L_{0}\right)  $
and \ $\mathrm{conv\,ext\,}\mathcal{M}_{F_{\max}}\left(  L_{0}\right)  $ are
confined to lower dimensional subspace, $\mathrm{conv\,ext\,}\mathcal{M}%
_{F_{1/2}}\left(  L_{0}\right)  $, which lies between them, is extends to the
full space.

\section{ Qubit case}

\subsection{$\mathcal{M}_{F_{\min}}\left(  L_{0}\right)  $}

In this section, we determine $\mathcal{M}_{F_{\min}}\left(  L_{0}\right)  $ ,
when $L_{0}$ is living in qubit space, $L_{0}\in\mathcal{L}_{2}$. \ In what
follows, $\sigma_{x}$, $\sigma_{y}$, and $\sigma_{z}$ are Pauli matrices.

\ By (\ref{MFmin}),%
\[
\mathcal{M}_{F_{\min}}\left(  L_{0}\right)  \cap\mathcal{L}_{2}=\left\{
L_{1};L_{1}\geq\frac{1}{4}\left(  I_{2}-\sqrt{-1}A\right)  L_{0}^{-1}\left(
I_{2}+\sqrt{-1}A\right)  ,A\in\mathcal{L}_{sa,2}\right\}  .
\]
Define
\begin{align*}
\mathcal{M}_{0}\left(  M\right)   &  :=\left\{  L;L\geq\left(  M+\sqrt
{-1}B\right)  \left(  M-\sqrt{-1}B\right)  -M^{2},B\in\mathcal{L}%
_{sa,2}\right\} \\
&  =\left\{  L;L\geq\sqrt{-1}\left[  B,M\right]  +B^{2},B\in\mathcal{L}%
_{sa,2}\right\}  ,
\end{align*}
and suppose, without loss of generality, $L_{0}^{-1}=l\sigma_{z}+mI_{2}$. Then
by Lemma\thinspace\ref{lem:extM0}, \
\begin{align*}
&  \mathrm{ext}\mathcal{M}_{F_{\min}}\left(  L_{0}\right)  \cap\mathcal{L}%
_{2}\\
&  =\frac{1}{4}\left(  L_{0}^{-1}+\sqrt{L_{0}}\mathrm{ext}\mathcal{M}%
_{0}\left(  L_{0}^{-1}\right)  \sqrt{L_{0}}\right) \\
&  =\left\{  \frac{1}{4}L_{0}^{-1}+\frac{l^{2}}{4}\sqrt{L_{0}}\left(  s\left(
\cos\alpha\,\sigma_{x}+\sin\alpha\sigma_{y}\right)  +\frac{s^{2}}{4}%
I_{2}\right)  \sqrt{L_{0}}\,;\alpha\in\mathbb{R},\,\,s\in\left[  -2,2\right]
\right\} \\
&  =\left\{  \frac{1}{4}L_{0}^{-1}+\frac{l^{2}}{4}\left(  \frac{s}{\sqrt
{l^{2}-m^{2}}}\left(  \cos\alpha\,\sigma_{x}+\sin\alpha\sigma_{y}\right)
+\frac{s^{2}}{4}L_{0}\right)  \,;\alpha\in\mathbb{R},\,\,s\in\left[
-2,2\right]  \right\}  .
\end{align*}
Thus the dimension of the smallest affine plane spanned by $\mathrm{ext}%
\mathcal{M}_{F_{\min}}\left(  L_{0}\right)  $ is 3. By Lemma\thinspace
\ref{lem:M0-2},
\begin{align*}
&  \mathcal{M}_{F_{\min}}\left(  L_{0}\right)  \cap\mathcal{L}_{2}\\
&  =\frac{1}{4}\left(  L_{0}^{-1}+\sqrt{L_{0}}\mathcal{M}_{0}\left(
L_{0}^{-1}\right)  \sqrt{L_{0}}\right) \\
&  =\left\{  \frac{1}{4}\left(  I_{2}-\sqrt{-1}A\right)  L_{0}^{-1}\left(
I_{2}+\sqrt{-1}A\right)  ;A\in\mathcal{L}_{sa,2}\right\}  .
\end{align*}

\subsection{$\hat{F}_{\max}$ and $\,\hat{F}_{\min}$}

In this subsection, we deal with $\hat{F}_{\max}$, $\,\hat{F}_{\min}$ and
$\hat{F}_{1/2}$ .

\ First, we compute $\hat{F}_{\max}$ by (\ref{hatFmax-explicit}). Since
$\mathrm{tr}\,L_{0}^{1/2}L_{1}L_{0}^{1/2}=\mathrm{tr}\,L_{0}L_{1}$ and $\,\det
L_{0}^{1/2}L_{1}L_{0}^{1/2}=\det L_{0}\det L_{1}$,
\begin{align*}
\hat{F}_{\max}\left(  L_{0},L_{1}\right)   &  =\frac{\mathrm{tr}\,L_{0}%
L_{1}-\sqrt{\left(  \mathrm{tr}\,L_{0}L_{1}\right)  ^{2}-4\det L_{0}\det
L_{1}}}{2}\\
&  =\frac{2\det L_{0}\det L_{1}}{\mathrm{tr}\,L_{0}L_{1}+\sqrt{\left(
\mathrm{tr}\,L_{0}L_{1}\right)  ^{2}-4\det L_{0}\det L_{1}}}.
\end{align*}

Suppose
\begin{align}
\mathrm{tr}\,L_{0}  &  =\mathrm{tr}\,L_{1}=2,\label{trL=2}\\
\frac{1}{2}\mathrm{tr}\,L_{0}^{2}-1  &  =\frac{1}{2}\mathrm{tr}\,L_{1}%
^{2}-1=r^{2}. \label{trL2=r2}%
\end{align}
Then
\[
\det L_{0}=\det L_{1}=1-r^{2}.
\]
Thus
\[
\hat{F}_{\max}\left(  L_{0},L_{1}\right)  =\frac{2\left(  1-r^{2}\right)
^{2}}{\mathrm{tr}\,L_{0}L_{1}+\sqrt{\left(  \mathrm{tr}\,L_{0}L_{1}\right)
^{2}-4\left(  1-r^{2}\right)  ^{2}}}.
\]

The general explicit formula for $\hat{F}_{\min}$ is awfully complicated even
for qubit case. \ But when (\ref{trL=2}) and (\ref{trL2=r2}) hold, it takes
very simple form. Let
\[
L_{0}=I_{2}+a\sigma_{x}+b\sigma_{z},L_{1}=I_{2}+a\sigma_{x}-b\sigma_{z},
\]
where $a^{2}+b^{2}=r^{2}\leq1$. Then by (\ref{hatFmin'}) and Theorem\thinspace
\ref{th:hatFmax=max},\thinspace\ \
\begin{align*}
\hat{F}_{\min}\left(  L_{0},L_{1}\right)   &  =2\sqrt{\min_{x,z:x^{2}+z^{2}%
=1}\left(  1+ax+bz\right)  \left(  1+ax-bz\right)  }\\
&  =2\sqrt{\min_{x,z:x^{2}+z^{2}=1}\left(  r^{2}x^{2}+2ax+1-b^{2}\right)  }\\
&  =2\sqrt{\min_{x\in\left[  -1,1\right]  }\left(  r^{2}x^{2}+2ax+1-b^{2}%
\right)  }\\
&  =\left\{
\begin{array}
[c]{cc}%
2\sqrt{\left(  1-r^{2}\right)  \left(  1-\frac{a^{2}}{r^{2}}\right)  }, &
\left(  r^{2}\geq\left\vert a\right\vert \right)  ,\\
2\left(  1-\left\vert a\right\vert \right)  , & \left(  r^{2}<\left\vert
a\right\vert \right)  ,
\end{array}
\right. \\
&  =\left\{
\begin{array}
[c]{cc}%
2\sqrt{\left(  1-r^{2}\right)  \left(  1-\frac{1}{4r^{2}}\left(
\mathrm{tr}\,L_{0}L_{1}-2\left(  1-r^{2}\right)  \right)  \right)  }, &
\left(  r^{2}\geq\frac{1}{2}\sqrt{\mathrm{tr}\,L_{0}L_{1}-2\left(
1-r^{2}\right)  }\right)  ,\\
2\left(  1-\frac{1}{2}\sqrt{\mathrm{tr}\,L_{0}L_{1}-2\left(  1-r^{2}\right)
}\right)  , & \left(  r^{2}<\frac{1}{2}\sqrt{\mathrm{tr}\,L_{0}L_{1}-2\left(
1-r^{2}\right)  }\right)  ,
\end{array}
\right.
\end{align*}
where the first equality is due to $x^{2}+z^{2}=1$.

So if (\ref{trL=2}) and (\ref{trL2=r2}) hold, $\hat{F}_{\max}\left(
L_{0},L_{1}\right)  $ and $\hat{F}_{\min}\left(  L_{0},L_{1}\right)  $ are
decreasing in the overlap $\mathrm{tr}\,L_{0}L_{1}$. While fidelity is
increasing in the overlap between the states, its dual is decreasing in the
overlap of observables. One may wonder whether such quantity can be of any
use. But, since $\hat{F}_{\max}$ and $\hat{F}_{\min}$ are CPTP monotone
increasing by CP unital map $\Lambda^{\ast}$, $\ \hat{F}_{\max}\left(
L_{0},L_{1}\right)  \leq\hat{F}_{\max}\left(  L_{0}^{\prime},L_{1}^{\prime
}\right)  $ or $\hat{F}_{\min}\left(  L_{0},L_{1}\right)  \leq\hat{F}_{\min
}\left(  L_{0}^{\prime},L_{1}^{\prime}\right)  $ or is a necessary condition
for
\begin{equation}
\Lambda^{\ast}\left(  L_{\theta}\right)  =L_{\theta}^{\prime},\,\theta
\in\left\{  0,1\right\}  \label{L*-convert}%
\end{equation}
to hold for some CP unital map $\Lambda^{\ast}$. In fact, using these
conditions, one can prove the following assertion.

\begin{proposition}
Consider a qubit system, and suppose $\left(  L_{0},L_{1}\right)  $ and
$\left(  L_{0}^{\prime},L_{1}^{\prime}\right)  $ satisfy (\ref{trL=2}) and
\[
\mathrm{tr}\,L_{0}^{2}=\mathrm{tr}\,L_{1}^{2}=\mathrm{tr}\,\left(
L_{0}^{\prime}\right)  ^{2}=\mathrm{tr}\,\left(  L_{1}^{\prime}\right)  ^{2}.
\]
Then, there is a CP unital map $\Lambda^{\ast}$ with (\ref{L*-convert}) if and
only if
\[
\mathrm{tr}\,L_{0}L_{1}=\mathrm{tr}\,L_{0}^{\prime}L_{1}^{\prime}.
\]
In other words, $\left(  L_{0},L_{1}\right)  $ and $\left(  L_{0}^{\prime
},L_{1}^{\prime}\right)  $ are unitary equivalent.
\end{proposition}

\begin{proof}
Without loss of generality, we put
\[
L_{0}=I_{2}+a\sigma_{x}+b\sigma_{z},L_{1}=I_{2}+a\sigma_{x}-b\sigma_{z}.
\]
Then, under the condition of the present proposition,
\[
\left\Vert L_{0}-L_{1}\right\Vert =\left\Vert 2b\sigma_{z}\right\Vert
=2\left\vert b\right\vert =\sqrt{\mathrm{tr}\,L_{0}^{2}-\mathrm{tr}%
\,L_{0}L_{1}}%
\]
is monotone decreasing in $\mathrm{tr}\,L_{0}L_{1}$. Since $\left\Vert
L_{0}-L_{1}\right\Vert $ is monotone decreasing by application of any CP
unital map, $\mathrm{tr}\,L_{0}L_{1}$ is monotone increasing by any CP unital map.

$\hat{F}_{\max}\left(  L_{0},L_{1}\right)  $ and $\hat{F}_{\min}\left(
L_{0},L_{1}\right)  $ are decreasing in $\mathrm{tr}\,L_{0}L_{1}$ and
increasing by application of any CP unital map. Therefore, $\mathrm{tr}%
\,L_{0}L_{1}$ is monotone decreasing by any CP unital map. Combining the above
argument, $\mathrm{tr}\,L_{0}L_{1}$ is invariant by any CP unital map under
the condition of the present proposition. Thus we have the assertion.
\end{proof}

\begin{proposition}
Consider a qubit system, and suppose $\mathrm{rank}\,L_{0}=1$ or
$\mathrm{rank}\,L_{1}=1$. There is a CP unital map $\Lambda^{\ast}$ with
(\ref{L*-convert}) only if $\mathrm{rank}\,L_{0}^{\prime}=1$ or $\mathrm{rank}%
\,L_{1}^{\prime}=1$.
\end{proposition}

\begin{proof}
Since $\hat{F}_{\max}\left(  L_{0},L_{1}\right)  =\hat{F}_{\min}\left(
L_{0},L_{1}\right)  =0$ and $\hat{F}^{Q}$ is monotone decreasing by CP unital
map, $\hat{F}_{\max}\left(  L_{0}^{\prime},L_{1}^{\prime}\right)  =\hat
{F}_{\min}\left(  L_{0}^{\prime},L_{1}^{\prime}\right)  =0$. Therefore, either
$L_{0}^{\prime}$ or $L_{1}^{\prime}$ is not full-rank.\bigskip
\end{proof}

\appendix

\section{Matrix}

\begin{lemma}
\label{lem:block-positive} Let $X$, $Y$ be a positive definite matrices.
Then,
\begin{equation}
\left[
\begin{array}
[c]{cc}%
X & C\\
C^{\dagger} & Y
\end{array}
\right]  \geq0 \label{block}%
\end{equation}
if and only if
\begin{equation}
\left(  I-\pi_{X}\right)  C=0,\,\,C\left(  I\mathbf{-}\pi_{Y}\right)  =0.
\label{PCP=0}%
\end{equation}
and%
\begin{equation}
X\geq CY^{-1}C^{\dagger} \label{X>Y}%
\end{equation}

\end{lemma}

\begin{proof}
Suppose (\ref{block}) holds. To prove $\left(  I-\pi_{X}\right)  C=0$, suppose
$\left(  I-\pi_{X}\right)  C\neq0$. Then, there is a unit vector $\left\vert
\varphi\right\rangle $ in the support of $I-\pi_{X}$ such that $\left\langle
\varphi\right\vert C\neq0$. Therefore, for a sufficiently large $c>0$,
\[
\left[
\begin{array}
[c]{cc}%
-c\left\langle \varphi\right\vert  & \left\langle \varphi\right\vert C
\end{array}
\right]  \left[
\begin{array}
[c]{cc}%
X & C\\
C^{\dagger} & Y
\end{array}
\right]  \left[
\begin{array}
[c]{c}%
-c\left\vert \varphi\right\rangle \\
C^{\dagger}\left\vert \varphi\right\rangle
\end{array}
\right]  =0-2c\left\langle \varphi\right\vert CC^{\dagger}\left\vert
\varphi\right\rangle +\left\langle \varphi\right\vert C^{\dagger}YC\left\vert
\varphi\right\rangle <0.
\]
This contradicts with (\ref{block}). Therefore, we have $\left(  I-\pi
_{X}\right)  C=0$. The proof of $C\left(  I\mathbf{-}\pi_{Y}\right)  =0$ is
almost parallel.

If (\ref{block}) holds,
\begin{align}
\left[
\begin{array}
[c]{cc}%
I & -CY^{-1}\\
0 & I
\end{array}
\right]  \left[
\begin{array}
[c]{cc}%
X & C\\
C^{\dagger} & Y
\end{array}
\right]  \left[
\begin{array}
[c]{cc}%
I & 0\\
-Y^{-1}C^{\dagger} & I
\end{array}
\right]   &  =\left[
\begin{array}
[c]{cc}%
X-CY^{-1}C^{\dagger} & C-C\pi_{Y}\\
C^{\dagger}-\pi_{Y}C^{\dagger} & Y
\end{array}
\right] \nonumber\\
&  =\left[
\begin{array}
[c]{cc}%
X-CY^{-1}C^{\dagger} & 0\\
0 & Y
\end{array}
\right]  \geq0, \label{positive-equiv}%
\end{align}
which implies (\ref{X>Y}).

Suppose, on the other hand, (\ref{PCP=0}) and (\ref{X>Y}) holds. Tracking back
the chain of identities in (\ref{positive-equiv}), we have%
\[
\left[
\begin{array}
[c]{cc}%
I & -CY^{-1}\\
0 & I
\end{array}
\right]  \left[
\begin{array}
[c]{cc}%
X & C\\
C^{\dagger} & Y
\end{array}
\right]  \left[
\begin{array}
[c]{cc}%
I & 0\\
-Y^{-1}C^{\dagger} & I
\end{array}
\right]  \geq0.
\]
Since
\[
\det\left[
\begin{array}
[c]{cc}%
I & -CY^{-1}\\
0 & I
\end{array}
\right]  =1\neq0,
\]
this matrix is invertible. Therefore, we have (\ref{block}).
\end{proof}

\section{Convex analysis}

Below, unless otherwise mentioned, a function $f$ is defined on $%
\mathbb{R}
^{n}$ and takes values in $%
\mathbb{R}
\cup\left\{  \infty,-\infty\right\}  $. The \textit{epigraph} $\mathrm{epi\,}%
f$ of a function $f$ defined on $%
\mathbb{R}
^{n}$ is
\[
\mathrm{epi\,}f=\left\{  \left(  x,y\right)  ;x\in%
\mathbb{R}
^{n},y\geq f\left(  x\right)  \right\}  .
\]
$f$ is said to be \textit{convex} if $\mathrm{epi\,}f$ is convex, and
\textit{concave} if $-f$ is convex.

The effective domain $\mathrm{dom}\,f$ of a convex (concave, resp.) function
is
\[
\mathrm{dom}\,f:=\left\{  x\,;x\in%
\mathbb{R}
^{n},\,f\left(  x\right)  <\infty\right\}  \subset%
\mathbb{R}
^{n},
\]
(%
\[
\mathrm{dom}\,f:=\left\{  x\,;x\in%
\mathbb{R}
^{n},\,f\left(  x\right)  >-\infty\right\}  \subset%
\mathbb{R}
^{n},
\]
resp. ). A convex (concave, resp.) function $f$ is said to be \textit{proper}
if $f\left(  x\right)  \neq-\infty$ ($f\left(  x\right)  \neq\infty$, resp.)
for any $x$ and $f\left(  x\right)  \neq\infty$ ($f\left(  x\right)
\neq-\infty$, resp.) for some $x$. A sublinear function $f$ is a function
which is convex and homogeneous.

A function $f$ is said to be \textit{lower semi continuous} (\textit{upper
semi continuous}, resp.)\ if
\[
f\left(  x\right)  =\varliminf_{y\rightarrow x}\,f\left(  y\right)
=\lim_{\varepsilon\downarrow0}\,(\inf\left\{  f\left(  y\right)
\,;\,\left\Vert y-x\right\Vert <\varepsilon\right\}  )
\]
(
\[
f\left(  x\right)  =\varlimsup_{y\rightarrow x}\,f\left(  y\right)
=\lim_{\varepsilon\downarrow0}\,(\sup\left\{  f\left(  y\right)
\,;\,\left\Vert y-x\right\Vert <\varepsilon\right\}  ),
\]
resp.). The \textit{lower semicontinuous hull} (upper \textit{semicontinuous
hull}, resp.) of $f$ is the greatest lower semicontinuous (the smallest upper
semicontinuous) function which is not larger than (not smaller than, resp.)
$f$ .$\,$

\begin{lemma}
\label{lem:semicont}For any family of functions $\left\{  f_{i};i\in
I\right\}  $,
\begin{align*}
\varliminf_{y\rightarrow x}\,\sup_{i\in I}f_{i}\left(  y\right)   &
=\sup_{i\in I}\,\varliminf_{y\rightarrow x}\,f_{i}\left(  y\right)  ,\\
\varlimsup_{y\rightarrow x}\inf_{i\in I}\,f_{i}\left(  y\right)   &
=\inf_{i\in I}\,\varlimsup_{y\rightarrow x}\,f_{i}\left(  y\right)  .
\end{align*}
Therefore, if each $f_{i}$ is lower semicontinuous (upper semicontinuous,
resp.), \ so is $\sup_{i\in I}f_{i}$ ($\inf_{i\in I}\,f_{i}$).
\end{lemma}

\begin{proof}
Observe
\begin{align*}
\varliminf_{y\rightarrow x}\,\sup_{i\in I}f_{i}\left(  y\right)   &
=\lim_{\varepsilon\downarrow0}\,\inf\left\{  \sup_{i\in I}f_{i}\left(
y\right)  \,;\,\left\Vert y-x\right\Vert <\varepsilon\right\} \\
&  =\sup_{\varepsilon>0}\inf\left\{  \sup_{i\in I}f_{i}\left(  y\right)
\,;\,\left\Vert y-x\right\Vert <\varepsilon\right\} \\
&  \geq\sup_{i\in I}\sup_{\varepsilon>0}\inf\left\{  f_{i}\left(  y\right)
\,;\,\left\Vert y-x\right\Vert <\varepsilon\right\} \\
&  =\sup_{i\in I}\,\varliminf_{y\rightarrow x}\,f_{i}\left(  y\right)  .
\end{align*}
Since $\varliminf_{y\rightarrow x}f\left(  y\right)  \leq f\left(  x\right)  $
by definition, this means
\[
\varliminf_{y\rightarrow x}\,\sup_{i\in I}f_{i}\left(  y\right)  =\sup_{i\in
I}\,\varliminf_{y\rightarrow x}\,f_{i}\left(  y\right)  .
\]
The second identity is shown in almost parallel manner.
\end{proof}

The \textit{closure} $\mathrm{cl}\,\,f$ of a convex (concave, resp.) function
$f$ is defined as follows. If $f$ nowhere has the value $-\infty$ ($\infty$,
resp.), $\mathrm{cl}\,\,f$ is the lower semicontinuous hull (upper
semicontinuous hull, resp.) of $f$. If $f\left(  x\right)  =-\infty$
($=\infty$, resp.) for some $x$, $\mathrm{cl}\,\,f$ is the constant function
$-\infty$ ($\infty$, resp.). A convex or concave function $f$ is said to be
\textit{closed} if $\mathrm{cl}\,\,f=f$. If $f$ nowhere has the value
$-\infty$ and $f$ is convex, $f$ is closed if and only if $\mathrm{epi}\,f$ is closed.

The \textit{affine hull} $\mathrm{aff}$ $C$ of a set $C$ is the smallest
affine set which includes $C$. The \textit{relative interior} $\mathrm{ri}\,C$
of a convex set $C$ is
\[
\mathrm{ri}\,C=\left\{  x\in\mathrm{aff\,}C\,;\,\exists\varepsilon
>0,\,\,\left(  x+B_{\varepsilon}\right)  \cap\mathrm{aff\,}C\subset C\right\}
,
\]
where $B_{\varepsilon}$ is $\varepsilon-$ball centered at $0$. The
\textit{relative boundary} of $C$ is $\mathrm{cl}\,C\,\backslash
\,\mathrm{ri}\,C$\thinspace.

\begin{lemma}
\label{lem:conv-cont}(Theorem\thinspace10.1 and Theorem 7.4,
\cite{Rockafellar}) A convex function $f$ on $\mathbb{R}^{n}$ is continuous on
$\mathrm{ri\,}(\mathrm{dom\,}f)$. Let $f$ be a proper convex function on
$\mathbb{R}^{n}$. Then $\mathrm{cl}\,f$ agrees with $f$ except perhaps at
relative boundary points of $\mathrm{dom}\,f$.
\end{lemma}

\begin{lemma}
\label{lem:cl-f}(Theorem 7.4, \cite{Rockafellar}) If $f$ is proper and convex,
so is $\mathrm{cl}\,f$.
\end{lemma}

\begin{lemma}
\label{lem:lim-conv-closed} If $f_{i}$ is convex, closed, and nowhere has the
value $-\infty$ for each $i\in I$, so is $\sup_{i\in I}f_{i}$. \ Also, if
$f_{i}$ is concave, closed and has nowhere has the value $\infty$\textrm{ }for
each $i\in I$, so is $\inf_{i\in I}f_{i}$.
\end{lemma}

\begin{proof}
We only have to show the first statement, since the second one follows by
considering $-f_{i}$. Observe
\[
\mathrm{epi}\,\sup_{i\in I}f_{i}=\bigcap_{i\in I}\mathrm{epi}\,\,f_{i}\text{.}%
\]
Therefore, if each $\mathrm{epi}\,\,f_{i}$ is convex and closed, so is
$\mathrm{epi}\,\,\sup_{i\in I}f_{i}$.
\end{proof}

The \textit{dual} $f^{\ast}:%
\mathbb{R}
^{n}\rightarrow%
\mathbb{R}
\cup\left\{  \infty,-\infty\right\}  $ of $f$ is
\[
f^{\ast}\left(  x^{\ast}\right)  :=\sup_{x\in%
\mathbb{R}
^{n}}\,\left\langle x^{\ast},x\right\rangle -f\left(  x\right)  .
\]

\begin{lemma}
\label{lem:conjugate-function}(Theorem 12.2 and Corollary\thinspace12.2.1,
\cite{Rockafellar}) Let $f$ be a convex function. The conjugate function
$f^{\ast}$ is \ then a closed convex function, proper if and only if $f$ is
proper. Moreover, \ $(\mathrm{cl}\,f)^{\ast}$ $=f^{\ast}$ and $f^{\ast\ast
}=\mathrm{cl}\,f$ . Thus, The conjugacy operation $f\rightarrow f^{\ast}$
induces a symmetric one-to-one correspondence in the class of all closed
proper convex functions on $\mathbb{R}^{n}$.
\end{lemma}

The indicator function $\delta\left(  x|C\right)  $ and the support function
$\delta^{\ast}\left(  x^{\ast}|C\right)  $ of a convex set $C$ is
\[
\delta\left(  x|C\right)  :=\left\{
\begin{array}
[c]{cc}%
0, & \left(  x\in C\right)  ,\\
\infty, & \left(  x\notin C\right)  ,
\end{array}
\right.
\]
and
\begin{align*}
\delta^{\ast}\left(  x^{\ast}|C\right)   &  :=\sup_{x\in C}\,\left\langle
x^{\ast},x\right\rangle \\
&  =\sup_{x\in%
\mathbb{R}
^{n}}\,\left\langle x^{\ast},x\right\rangle -\delta\left(  x|C\right) \\
&  =\left(  \delta\left(  \cdot|C\right)  \right)  ^{\ast}\,\left(  x^{\ast
}\right)  .
\end{align*}
For any convex set $C$ (p.\thinspace112 of \cite{Rockafellar})
\begin{equation}
\delta^{\ast}\left(  x^{\ast}|C\right)  =\delta^{\ast}\left(  x^{\ast
}|\mathrm{cl}\,C\right)  . \label{support-closed}%
\end{equation}
If $C$ is a closed convex set, $\delta\left(  \cdot|C\right)  $ is closed,
since it is lower semicontinuous.

\begin{lemma}
\label{lem:sublinear} (Theorem\thinspace13.2, \cite{Rockafellar}) The
indicator function and the support function of a closed convex set are
conjugate to each other. The support function of a non-empty convex set is
closed, proper, convex, and positively homogeneous. Also, any closed, proper,
convex, and positively homogeneous function is the support function of a
non-empty convex set.
\end{lemma}

\begin{lemma}
\label{lem:sublinear-2}Suppose $f^{\ast}$ is a closed proper convex functions
which are positively homogeneous. Then, there is a non-empty closed convex set
$C_{f^{\ast}}$ such that
\begin{equation}
f^{\ast}\left(  x^{\ast}\right)  =\delta^{\ast}\left(  x^{\ast}|C_{f^{\ast}%
}\right)  , \label{sublinear=support}%
\end{equation}
and the correspondence between $f^{\ast}$ and $C_{f^{\ast}}$ is one-to-one.
\end{lemma}

\begin{proof}
By Lemma\thinspace\ref{lem:sublinear}, there is a non-empty convex set with
(\ref{sublinear=support}). By (\ref{support-closed}), we can suppose that
$C_{f^{\ast}}$ is closed, and thus, that its indicator function $\delta\left(
\cdot|C_{f^{\ast}}\right)  $ is closed. Therefore, by Lemma\thinspace
\ref{lem:conjugate-function},
\[
\left(  f^{\ast}\right)  ^{\ast}=\left(  \delta^{\ast}\left(  \cdot
|C_{f^{\ast}}\right)  \right)  ^{\ast}=\left(  \delta\left(  \cdot|C_{f^{\ast
}}\right)  \right)  ^{\ast\ast}=\mathrm{cl\,}\delta\left(  \cdot|C_{f^{\ast}%
}\right)  =\delta\left(  \cdot|C_{f^{\ast}}\right)  .
\]
Thus, for each given $f^{\ast}$, $\delta\left(  \cdot|C_{f^{\ast}}\right)  $
is uniquely decided, and we have the assertion.
\end{proof}

The recession cone $0^{+}C$ of the convex set $C$ is
\[
0^{+}C:=\left\{  y\,;\,x+\lambda y\in C,\,\forall x\in C,\,\forall\lambda
\geq0\right\}  ,
\]
and the recession function $f$ $0^{+}$of the convex function $f$ is the
function such that
\[
\mathrm{epi\,}\,f0^{+}=0^{+}\mathrm{epi\,\,}f.
\]

.

Let $C_{1}$ and $C_{2}$ be non-empty sets in $\mathbb{R}^{n}$. A hyperplane
$H$ is said to separate $C_{1}$ and $C_{2}$ if $C_{1}$ is contained in one of
the closed half-spaces associated with $H$ and $C_{2}$ lies in the opposite
closed half-space. It is said to separate $C_{1}$ and $C_{2}$ properly if
$C_{1}$ and $C_{2}$ are not both actually contained in $H$ itself.

\begin{lemma}
\label{lem:separation}(Theorem 11.1 of \cite{Rockafellar}) Let $C_{1}$ and
$C_{2}$ be non-empty sets in $\mathbb{R}^{n}$. There exists a hyperplane
separating $C_{1}$ and $C_{2}$ properly if and only if there exists a vector
$b$ such that
\begin{align*}
\inf\left\{  \,\left\langle x,b\right\rangle \,;\,x\in C_{1}\right\}   &
\geq\sup\left\{  \,\left\langle x,b\right\rangle \,;\,x\in C_{2}\right\}  ,\\
\sup\left\{  \,\left\langle x,b\right\rangle \,;\,x\in C_{1}\right\}   &
>\inf\left\{  \,\left\langle x,b\right\rangle \,;\,x\in C_{2}\right\}  .
\end{align*}

\end{lemma}

A \textit{face} of a convex set $C$ is a convex subset $C\prime$ of $C$ such
that every (closed) line segment in $C$ with a relative interior point in
$C\prime$ has both endpoints in $C^{\prime}$. A face consists of a single
point is called an \textit{extreme point}. $x\in C$ is an extreme point if and
only if it cannot be expressed as a convex combination of points of $C$ other
than $x$. The set of all extreme points of $C$ is expressed as $\mathrm{ext}%
\,C$. If $C^{\prime}$ is a half-line face of a convex set $C$, we shall call
the direction of $C\prime$ an \textit{extreme direction} of $C\,$ (extreme
point of $C$ at infinity). \ Obviously, an extreme direction of $C$ is, viewed
as a point in $%
\mathbb{R}
^{n}\backslash\left\{  0\right\}  $, is a member of recession cone $0^{+}C$.

\begin{lemma}
(Theorem 18.5\thinspace of \cite{Rockafellar}) Let $C$ be a closed convex set
containing no lines. Then, any point $x\in C$ can be written as
\[
x=\sum_{i}\lambda_{i}y_{i}\,+\sum_{i}\mu_{i}z_{i},
\]
where $\lambda_{i}\geq0$, $\sum_{i}\lambda_{i}=1$, $\mu_{i}\geq0$, $y_{i}\in$
$\mathrm{ext}\,C$, and $z_{i}$ is an extreme direction of $C$, for each $i$.
\end{lemma}

\begin{lemma}
(Corollary 18.3.1\thinspace of \cite{Rockafellar}) Let $C$ be a closed convex
set. Let $S_{1}$ be a subset of $C$, and $S_{2}$ be a set of directions such
that
\[
x=\sum_{i}\lambda_{i}y_{i}\,+\sum_{i}\mu_{i}z_{i}%
\]
stands for some $\lambda_{i}\geq0$, $\sum_{i}\lambda_{i}=1$, $\mu_{i}\geq0$,
$y_{i}\in$ $S_{1}$, and $z_{i}\in$ $S_{2}$. Then, $\mathrm{ext}\,C$ is a
subset of $S_{1}$.
\end{lemma}

From these, the following lemma is immediate.

\begin{lemma}
\label{lem:ext+recess}Let $C$ be a closed convex set containing no lines.
Then,
\[
C=\mathrm{conv\,ext}\,C+0^{+}C.
\]
Also, if a subset $S$ of $C$ satisfies
\[
C=\mathrm{conv\,}S+0^{+}C,
\]
$S$ contains $\mathrm{ext}\,C\,$.
\end{lemma}

If a certain linear function $h$ achieves maximum over $C$ at $x\in C$ and not
achieved at any other point $x^{\prime}\in C$, $x$ is called an
\textit{exposed point} of $C$. Any exposed point is an extreme point, but not
vice versa.

\begin{lemma}
\label{lem:expose-dense}(Straszewicz's Theorem, Theorem 18.6\thinspace of
\cite{Rockafellar}) For any closed convex set $C$, the set of exposed points
of $C$ is a dense subset of the set of extreme points of $C$.
\end{lemma}

\section{Determination of certain convex set in $\mathcal{L}_{sa,2}$}

In this section, we determine
\begin{align*}
\mathcal{M}_{0}\left(  M\right)   &  :=\left\{  L;L\geq\left(  M+\sqrt
{-1}B\right)  \left(  M-\sqrt{-1}B\right)  -M^{2},B\in\mathcal{L}%
_{sa,2}\right\} \\
&  =\left\{  L;L\geq\sqrt{-1}\left[  B,M\right]  +B^{2},B\in\mathcal{L}%
_{sa,2}\right\}  .
\end{align*}
By Lemma\thinspace\ref{lem:ext+recess}, $\mathcal{M}_{0}\left(  M\right)  $ is
determined by $\mathrm{ext\,}\mathcal{M}_{0}\left(  M\right)  $,
\begin{equation}
\mathcal{M}_{0}\left(  M\right)  =\mathrm{conv}\,\mathrm{ext\,}\mathcal{M}%
_{0}\left(  M\right)  +\mathcal{P}_{2}. \label{M0=ext+recess}%
\end{equation}
So first we determine $\mathrm{ext\,}\mathcal{M}_{0}\left(  M\right)
$.\ \ Below, $\sigma_{x}$, $\sigma_{y}$, $\sigma_{z}$ are Pauli matrices.

\begin{lemma}
\label{lem:extM0}If $M=l\sigma_{z}+mI_{2}$,
\begin{align}
&  \mathrm{ext\,}\mathcal{M}_{0}\left(  M\right) \nonumber\\
&  =\left\{  l^{2}U\left(  s\sigma_{x}+\frac{s^{2}}{4}I_{2}\right)
U^{\dagger};s\in\left[  -2,2\right]  ,\left[  U,L_{0}^{-1}\right]  =0,U\in
SU\left(  2\right)  \right\} \label{extM0-1}\\
&  =\left\{  l^{2}\left(  s\left(  \cos\alpha\,\sigma_{x}+\sin\alpha\sigma
_{y}\right)  +\frac{s^{2}}{4}I_{2}\right)  \,;\alpha\in\mathbb{R}%
,\,\,s\in\left[  -2,2\right]  \right\}  \label{extM0-2}%
\end{align}

\end{lemma}

\begin{proof}
Lemma\thinspace\ref{lem:expose-dense}, we have to determine the set of all
exposed points of $\mathcal{M}_{0}\left(  M\right)  $. Thus, we investigate
\[
\min_{L\in\mathcal{M}_{0}\left(  M\right)  }\,\mathrm{tr}\,YL.
\]
If $Y$ is not positive, the target function is unbounded from below, thus the
minimum is never attained. Also, if $Y$ has eigenvalue $0$, the minimum is
achieved by any member of a certain convex set containing more than single
point. This means the corresponding minimum points are not exposed. Therefore,
we suppose
\[
Y>0.
\]
In this case, the above minimum equals the minimum of
\[
f_{Y}\left(  B\right)  :=\min_{A\in\mathcal{L}_{sa,2}}\mathrm{tr}\,Y\left(
\sqrt{-1}\left[  B,M\right]  +B^{2}\right)  ,
\]

Observe $\ f_{Y}\left(  B\right)  $ is a proper, convex, and differentiable
function. Hence, at the minimum point $B$, the derivative $\mathrm{D}%
f_{Y}\left(  B\right)  \left(  \dot{B}\right)  $ must vanish for any $\dot{B}%
$. Hence,
\begin{equation}
\sqrt{-1}\left(  MY-YM\right)  +BY+YB=0, \label{Fmin-marginal}%
\end{equation}

Observe that if $Y$ satisfies (\ref{Fmin-marginal}), so does real multiple of
$Y$. Therefore, without loss of generality, we suppose
\[
Y=a\sigma_{x}+b\sigma_{y}+c\sigma_{z}+I_{2},
\]
where $a^{2}+b^{2}+c^{2}<1$. Observe also if the pair $\left(  Y,B\right)  $
satisfies (\ref{Fmin-marginal}), so does $\left(  UYU^{\dagger},UBU^{\dagger
}\right)  $ , where $U$ is any unitary commutative with $M$. Therefore, we
first solve (\ref{Fmin-marginal}) fixing $b=0$, and then rotate the result by
unitaries commutative with $M$. \ 

Then, if $b=0$, after some calculations, one can easily verify that
$B=la\sigma_{y}$ satisfies (\ref{Fmin-marginal}). Since $B$ satisfying
(\ref{Fmin-marginal}) is unique for each $M$ and $Y>0$, this is the only
solution to (\ref{Fmin-marginal}). Applying above "gauge transform",
\[
B=laU\sigma_{y}U^{\dagger},
\]
where $U\in\mathrm{SU}\left(  2\right)  $ commute with $M$, are the solutions
to (\ref{Fmin-marginal}). Also, $a\in\left(  0,1\right)  $, so that $Y>0$.
Therefore, the set of all the exposed points of $\mathcal{M}_{0}\left(
M\right)  $ is
\[
\left\{  l^{2}U\left(  -2a\sigma_{x}+a^{2}I_{2}\right)  U^{\dagger}%
;a\in\left(  -1,1\right)  ,U\in\mathrm{SU}\left(  2\right)  ,\left[
U,M\right]  =0\right\}  .
\]
The closure of this, by Lemma\thinspace\ref{lem:expose-dense}, is
$\mathrm{ext\,}\mathcal{M}_{0}\left(  M\right)  $. Thus the asserted result is obtained.
\end{proof}

\bigskip

Below we use the following functions to describe the results.
\begin{align*}
f_{1}\left(  x\right)   &  :=\left\{
\begin{array}
[c]{cc}%
\left\vert x\right\vert -1, & \left(  \left\vert x\right\vert \geq2\right) \\
\frac{1}{4}x^{2}, & \left(  \left\vert x\right\vert \leq2\right)
\end{array}
\right.  ,\\
f_{2}\left(  x,z\right)   &  :=\left\{
\begin{array}
[c]{cc}%
\sqrt{x^{2}+z^{2}}-1, & \left(  x^{2}+z^{2}\geq4\right) \\
\frac{1}{4}\left\{  x^{2}+z^{2}\right\}  , & \left(  x^{2}+z^{2}\leq4\right)
\end{array}
,\,\right.
\end{align*}
and
\begin{align}
&  D\left(  x,z,w\right) \nonumber\\
&  :=16w^{4}+\left(  -8x^{2}+8z^{2}+32\right)  w^{3}+\left(  x^{4}+2x^{2}%
z^{2}-32x^{2}+z^{4}-8z^{2}+16\right)  w^{2}\\
&  +\left(  10x^{4}+2x^{2}z^{2}-8x^{2}-8z^{4}-32z^{2}\right)  w+\left(
x^{4}-3x^{4}z^{2}-x^{6}-3x^{2}z^{4}+20x^{2}z^{2}-z^{6}-8z^{4}-16z^{2}\right)
\allowbreak.\nonumber
\end{align}

\begin{lemma}
\label{lem:unique}For each $\left(  x,z\right)  \in\mathbb{R}^{2}$ with
$z\neq0$, $\underline{w}$ satisfying
\begin{equation}
D\left(  x,z,\underline{w}\right)  =0 \label{discriminant-1}%
\end{equation}
and \underline{$w$}$\geq f_{2}\left(  x,z\right)  $ is unique.
\end{lemma}

\begin{proof}
We view (\ref{discriminant-1}) as a equation for $\underline{w}$, and prove it
has a unique solution in the region defined by $\underline{w}\geq f_{2}\left(
x,z\right)  $.

First we study the case of $x\neq0$ and $x^{2}+z^{2}-4>0$. Let $D_{2}$ be the
discriminant of (\ref{discriminant-1}) viewed as a equation for $\underline
{w}$, \ then%
\begin{align*}
D_{2}  &  =c\times-x^{2}z^{2}\left(  84x^{2}z^{2}+3x^{2}z^{4}+3x^{4}%
z^{2}+48x^{2}-12x^{4}+x^{6}+48z^{2}-12z^{4}+z^{6}-64\right)  ^{3}\allowbreak\\
&  =c\times-x^{2}z^{2}\left\{  \left(  x^{2}+z^{2}\right)  ^{3}+15\left(
x^{2}+z^{2}\right)  ^{2}+48\left(  x^{2}+z^{2}\right)  -27\left(  x^{2}%
-z^{2}\right)  ^{2}-64\right\}  ^{3}\\
&  \leq c\times-x^{2}z^{2}\left\{  \left(  x^{2}+z^{2}\right)  ^{3}+15\left(
x^{2}+z^{2}\right)  ^{2}+48\left(  x^{2}+z^{2}\right)  -27\left(  x^{2}%
+z^{2}\right)  ^{2}-64\right\}  ^{3}\\
&  =c\times-x^{2}z^{2}\left\{  \left(  x^{2}+z^{2}\right)  ^{3}-12\left(
x^{2}+z^{2}\right)  ^{2}+48\left(  x^{2}+z^{2}\right)  -64\right\}  ^{3}\\
&  =c\times-x^{2}z^{2}\left(  x^{2}+z^{2}-4\right)  ^{9}<0.
\end{align*}
where $c$ is a positive constant. Therefore, (\ref{discriminant-1}) has two
distinct real roots and two (non-real) complex roots. Also,
\begin{align*}
D\left(  x,z,\sqrt{x^{2}+z^{2}}-1\right)   &  =z^{2}\left(  -2\left(
x^{2}+z^{2}+12\right)  \sqrt{x^{2}+z^{2}}+12x^{2}+16-15z^{2}\right) \\
&  \leq z^{2}\left\{  -2\left(  x^{2}+z^{2}+12\right)  \sqrt{x^{2}+z^{2}%
}+12\left(  x^{2}+z^{2}\right)  +16\right\} \\
&  =-2z^{2}\left(  \sqrt{x^{2}+z^{2}}-2\right)  ^{3}\\
&  <0.
\end{align*}
Since $\lim_{w\rightarrow\infty}D\left(  x,z,w\right)  =\lim_{w\rightarrow
-\infty}D\left(  x,z,w\right)  =\infty$, \ by intermediate value theorem, one
of two real solutions of (\ref{discriminant-1}) is smaller than $\sqrt
{x^{2}+z^{2}}-1$, and the other is larger. This means there is only one real
solution of (\ref{discriminant-1}) satisfying $\underline{w}\geq f_{2}\left(
x,z\right)  $.

Second, we consider $x=0$-case, where
\[
D\left(  x,z,\underline{w}\right)  =\left(  \underline{w}-z\right)  \left(
\underline{w}+z\right)  \left(  z^{2}+4\underline{w}+4\right)  ^{2}.
\]
Obviously, in this case, only positive solution of (\ref{discriminant-1}) is
$\underline{w}=z$, which satisfies $\underline{w}\geq f_{2}\left(  x,z\right)
$.

Third, we study the case of $x\neq0$ and $x^{2}+z^{2}\leq4$. Observe that the
third derivative is positive of the function $w\rightarrow$ $D\left(
x,z,w\right)  $ in the region $[\frac{x^{2}+z^{2}}{4},\infty)$,
\[
\frac{\partial^{3}}{\partial w^{3}}D\left(  x,z,w\right)  =384\left\{
w-\frac{1}{8}\left(  x^{2}-z^{2}-4\right)  \right\}  .
\]
and
\[
\frac{\partial^{2}}{\partial w^{2}}D\left(  x,z,\frac{x^{2}+z^{2}}{4}\right)
=2\left(  \left(  x^{2}-4\right)  ^{2}+14x^{2}z^{2}+13z^{4}+16z^{2}\right)
>0.
\]
So the second derivative is positive in the region $[\frac{x^{2}+z^{2}}%
{4},\infty)$. Since
\[
D\left(  x,z,\frac{x^{2}+z^{2}}{4}\right)  =\frac{1}{4}z^{2}\left\{  \left(
x^{2}+z^{2}-4\right)  ^{3}-108z^{2}\right\}  <0,
\]
the function $w\rightarrow D\left(  x,z,w\right)  $ is increasing at the
smallest solution $w_{0}$ of (\ref{discriminant-1}) in the region
$[\frac{x^{2}+z^{2}}{4},\infty)$. Therefore, there cannot be any larger
solution than $w_{0}$. This proves that (\ref{discriminant-1}) has only one
solution in the region $[\frac{x^{2}+z^{2}}{4},\infty)$.
\end{proof}

\begin{lemma}
\label{lem:M0}Suppose $M=l\sigma_{z}+mI_{2}$. Then
\[
l^{2}\left(  x\,\sigma_{x}+y\sigma_{y}+z\sigma_{z}+wI_{2}\right)
\in\mathcal{M}_{0}\left(  M\right)
\]
if and only if
\begin{align}
z  &  =0,\,w\geq f_{1}\left(  x^{\prime}\right)  \,\text{ }%
\label{MFmin-qubit-1}\\
\,\text{or }z  &  \neq0,w\geq f_{2}\left(  x^{\prime},z\right)  ,\,\,D\left(
x^{\prime},z,w\right)  \geq0 \label{MFmin-qubit-2}%
\end{align}
where we have defined $x^{\prime}:=\sqrt{x^{2}+y^{2}}$. \ 
\end{lemma}

\begin{proof}
By (\ref{M0=ext+recess}) and (\ref{extM0-2}),
\begin{align}
&  l^{2}\left(  x\sigma_{x}+y\sigma_{y}+z\sigma_{z}+wI_{2}\right)
+M\in\mathcal{M}\left(  M\right) \nonumber\\
&  \Leftrightarrow\exists s\in\left[  -1,1\right]  \mathbb{\,},\,\,\exists
\alpha\in\mathbb{R},\,w\geq\frac{s^{2}}{4}+\sqrt{\left(  x-s\cos\alpha\right)
^{2}+\left(  y-s\sin\alpha\right)  ^{2}+z^{2}},\\
&  \Leftrightarrow\exists s\in\left[  -1,1\right]  \mathbb{\,},\,\,w\geq
w_{2}\left(  s\right)  :=\frac{s^{2}}{4}+\sqrt{\left(  x^{\prime}-s\right)
^{2}+z^{2}},\\
&  \Leftrightarrow\exists s\in\mathbb{R\,},\,\,w\geq w_{2}\left(  s\right)
:=\frac{s^{2}}{4}+\sqrt{\left(  x^{\prime}-s\right)  ^{2}+z^{2}}, \label{m-1}%
\end{align}
where the second "$\Leftrightarrow$" is due to the fact that
\begin{align*}
s\sigma_{x}+\frac{s^{2}}{4}I_{2}  &  \geq2\sigma_{x}+\frac{2^{2}}{4}%
I_{2}\,,\,\text{if }s\geq2,\\
s\sigma_{x}+\frac{s^{2}}{4}I_{2}  &  \geq-2\sigma_{x}+\frac{\left(  -2\right)
^{2}}{4}I_{2}\,,\,\text{if }s\leq-2.
\end{align*}
For each given $\left(  x^{\prime},z\right)  $, $w_{2}\left(  s\right)  $ goes
to $+\infty$ as $s\rightarrow\infty$. So we are interested in the minimum
$\underline{w_{2}}$ of $w_{2}\left(  s\right)  $ over $\mathbb{R}$. \ First,
suppose $z=0$. Then,
\begin{align}
\underline{w_{2}}  &  =\min_{s\in\mathbb{R}}w_{2}\left(  s\right)  =\min
_{s\in\mathbb{R}}\frac{s^{2}}{4}+\left\vert x^{\prime}-s\right\vert
\nonumber\\
&  =f_{1}\left(  x^{\prime}\right)  , \label{z=0-m}%
\end{align}
verifying (\ref{MFmin-qubit-1}).

Next, suppose $z\neq0$. Since $w_{2}\left(  s\right)  $ is differentiable,
bounded below, and defined on the open set $\mathbb{R\,}$, its minimum
$\underline{w_{2}}$ should satisfy $\mathrm{d}\underline{w_{2}}/\mathrm{d}%
s=0$. The definition of $w_{2}$ is equivalent to
\begin{align}
f_{\mathcal{M}}\left(  s\right)   &  :=\left(  w_{2}-\frac{s^{2}}{4}\right)
^{2}-\left(  x^{\prime}-s\right)  ^{2}-z^{2}\nonumber\\
&  =\frac{1}{16}s^{4}-\left(  \frac{1}{2}w_{2}+1\right)  s^{2}+2x^{\prime
}s+\left(  w_{2}^{2}-x^{\prime2}-z^{2}\right)  \allowbreak=0, \label{m-2}%
\end{align}
and%
\begin{equation}
\frac{s^{2}}{4}<w_{2}. \label{m-4}%
\end{equation}
Note here $w_{2}=\frac{s^{2}}{4}$ cannot happen because of (\ref{m-1}) and
$z^{2}>0$. \ Differentiating both ends of (\ref{m-2}) by $s$,
\[
f_{\mathcal{M}}^{\prime}\left(  s\right)  +2\frac{\mathrm{d}w_{2}}%
{\mathrm{d}s}\left(  w_{2}-\frac{s^{2}}{4}\right)  =0,
\]
where
\[
f_{\mathcal{M}}^{\prime}\left(  s\right)  =\frac{s^{3}}{4}-\left(
w_{2}+2\right)  s+2x^{\prime}.
\]
is the derivative of $f_{\mathcal{M}}\left(  s\right)  $ by $s$ considering
$w_{2}$ as a constant. Thus, because of the restriction (\ref{m-4}),
$\mathrm{d}w_{2}/\mathrm{d}s=0$ is equivalent to \
\begin{equation}
f_{\mathcal{M}}^{\prime}\left(  s\right)  =0. \label{m-3}%
\end{equation}
So if $w_{2}=\underline{w_{2}}$, (\ref{m-2}) and (\ref{m-3}), viewed as
algebraic equations for $s$, has a real common root $s$ satisfying
(\ref{m-4}). Therefore, the discriminant of $f_{\mathcal{M}}\left(  s\right)
$ has to be zero. After some computation, the discriminant coincide with
$D\left(  x^{\prime},z,\underline{w_{2}}\right)  $ up to constant factor.
Therefore, we should have
\[
D\left(  x^{\prime},z,\underline{w_{2}}\right)  =0.
\]
$\left(  x^{\prime},z,\underline{w_{2}}\right)  $ should also have to satisfy
:
\begin{align*}
\underline{w_{2}}  &  \geq\min_{\left(  s,t\right)  \in\mathbb{R}^{2}}%
\frac{s^{2}+t^{2}}{4}+\sqrt{\left(  x^{\prime}-s\right)  ^{2}+\left(
z-t\right)  ^{2}}\\
&  =\min_{\left(  s,t\right)  \in\mathbb{R}^{2}}\frac{s^{2}+t^{2}}{4}%
+\sqrt{\left(  \sqrt{x^{\prime2}+z^{2}}-s\right)  ^{2}+t^{2}}\\
&  =\min_{s\in\mathbb{R}}\frac{s^{2}}{4}+\left\vert \sqrt{x^{\prime2}+z^{2}%
}-s\right\vert \\
&  =f_{2}\left(  x^{\prime},z\right)  .
\end{align*}
Therefore, $\underline{w_{2}}$ satisfies
\[
D\left(  x^{\prime},z,\underline{w_{2}}\right)  =0,\,\underline{w_{2}}\geq
f_{2}\left(  x^{\prime},z\right)  .
\]
By Lemma\thinspace\ref{lem:unique}, the above condition specifies
$\underline{w_{2}}$ uniquely.

Since $\lim_{w\rightarrow\infty}D\left(  x^{\prime},z,w\right)  =\infty$,
\ $w\geq\underline{w_{2}}$ is equivalent to $D\left(  x^{\prime},z,w\right)
\geq0$ and $\,w\geq f_{2}\left(  x^{\prime},z\right)  $. Thus we have
(\ref{MFmin-qubit-2}).
\end{proof}

\begin{lemma}
\label{lem:M0-2}If the dimension of the Hilbert space is $2$,
\[
\mathcal{M}_{0}\left(  M\right)  =\left\{  \sqrt{-1}\left[  B,M\right]
+B^{2};B\in\mathcal{L}_{sa,2}\right\}  .
\]

\end{lemma}

\begin{proof}
Without loss of generality, let $M=l\sigma_{z}+mI_{2}$. Let us parameterize
$B\in\mathcal{L}_{sa,2}$ as follows,
\[
B=lU\left(  \beta\sigma_{y}+\gamma\sigma_{z}+\delta I\right)  U^{\dagger},
\]
where $U$ is a unitary commuting with $M$. Then,%

\begin{align*}
&  \frac{1}{l^{2}}U^{\dagger}\left\{  \sqrt{-1}\left[  B,L_{0}^{-1}\right]
+\left(  B\right)  ^{2}\right\}  U\\
&  =x\sigma_{x}+y\sigma_{y}+z\sigma_{z}+wI,
\end{align*}
where
\begin{align*}
x  &  =-2\beta,y=2\beta\delta,\,\\
z  &  =2\gamma\delta,\,w=\beta^{2}+\gamma^{2}+\delta^{2}.
\end{align*}
Therefore, erasing $\beta$, $\gamma$ and replacing $t:=\delta^{2}$,
\begin{equation}
w\left(  t\right)  :=w=\frac{x^{\prime2}}{4\left(  1+t\right)  }+\frac{z^{2}%
}{4t}+t, \label{w}%
\end{equation}
where $x^{\prime}=\sqrt{x^{2}+y^{2}}$. Observe $t=\delta^{2}$ can take any
non-negative value. Observe also $\lim_{t\rightarrow\infty}w\left(  t\right)
=\infty$ , for any $x^{\prime}$, $z$. Hence, $w$ can take any value larger
than or equal to the minimum $\underline{w}$ of $w\left(  t\right)  $ over
$t\in\lbrack0,\infty)$. \ Below, we determine relation satisfied by
$\underline{w}$, $x^{\prime}$, and $z$, and shows that $w\geq\underline{w}$ is
equivalent to (\ref{MFmin-qubit-1}) and (\ref{MFmin-qubit-2}). Then, since
$x^{\prime}=\sqrt{x^{2}+y^{2}}$ is invariant by the conjugation of unitary $U$
commuting with $M$, Lemma\thinspace\ref{lem:M0} implies our assertion.

If $z=0$, (\ref{w}) is very simple and easy to minimize.%
\begin{equation}
\underline{w}=\min_{t\in\lbrack0,\infty)}w\left(  t\right)  =\min_{t\in
\lbrack0,\infty)}\frac{x^{\prime2}}{4\left(  1+t\right)  }+t=f_{1}\left(
x^{\prime}\right)  \,. \label{z=0}%
\end{equation}
Hence, in this case, $w\geq\underline{w}$ is equivalent to
(\ref{MFmin-qubit-1}).

Next, suppose $z\neq0$. Then $t$ cannot be $0$, $t\in\left(  0,\infty\right)
$. Since $w\left(  t\right)  $ is differentiable, defined on the open
interval, and bounded below, it has minimum, and at the minimum, the
derivative of $w\left(  t\right)  $ should vanish. Rearranging the terms of
(\ref{w}), $\left(  x^{\prime},z,\underline{w}\right)  $ satisfies
\begin{equation}
f_{\mathcal{N}}\left(  t\right)  :=4t^{3}+4\left(  1-\underline{w}\right)
t^{2}+\left(  x^{\prime2}+z^{2}-4\underline{w}\right)  t+z^{2}=0. \label{n-1}%
\end{equation}
Differentiating both sides by $t$,
\[
f_{\mathcal{N}}^{\prime}\left(  t\right)  -4\left(  t^{2}+t\right)
\frac{\mathrm{d}\,\underline{w}}{\mathrm{d}\,t}=0,
\]
where%
\[
f_{\mathcal{N}}^{\prime}\left(  t\right)  =12t^{2}+8\left(  1-\underline
{w}\right)  t+x^{\prime2}+z^{2}-4\underline{w}=0
\]
is the derivative $f_{\mathcal{N}}\left(  t\right)  $ by $t$ viewing
$\underline{w}$ as a constant. Since $t>0$, $\mathrm{d}\underline
{w}/\mathrm{d}t=0$ is equivalent to
\begin{equation}
f_{\mathcal{N}}^{\prime}\left(  t\right)  =0. \label{n-2}%
\end{equation}
(\ref{n-1}) has a multiple root if and only if its discriminant is zero. After
some tedious calculations (in fact done by computer algebra system), this is
equivalent to%
\begin{equation}
D\left(  x^{\prime},z,\underline{w}\right)  =0, \label{discriminant}%
\end{equation}
\ In addition to this, $\left(  x^{\prime},z,\underline{w}\right)  $ has to
obey other constrains. Since
\[
w\left(  t\right)  \geq\frac{x^{\prime2}+z^{2}}{4\left(  1+t\right)  }+t,
\]
we should have
\begin{equation}
\underline{w}\geq f_{2}\left(  x^{\prime},z\right)  . \label{w>-2}%
\end{equation}
By Lemma\thinspace\ref{lem:unique}, (\ref{discriminant}) and (\ref{w>-2})
uniquely determines $\underline{w}$. Since $\lim_{w\rightarrow\infty}D\left(
x^{\prime},z,w\right)  =\infty$, Therefore, $w\geq\underline{w}$ is equivalent
to (\ref{MFmin-qubit-2}).

After all, $w\geq\underline{w}$ is equivalent to (\ref{MFmin-qubit-1}) and
(\ref{MFmin-qubit-2}), and we have the assertion. \ 
\end{proof}

\end{document}